\documentclass[twocolumn,superscriptaddress,floatfix,letterpaper,nofootinbib]{revtex4-2}

\usepackage{graphicx}
\usepackage[export]{adjustbox}

\usepackage{dcolumn}
\usepackage{bm}
\usepackage{multirow}

\usepackage[normalem]{ulem}

\usepackage{amsmath}
\usepackage{amsthm}
\usepackage{amstext}
\usepackage{amssymb}
\usepackage{mathrsfs}
\usepackage{amsfonts}
\usepackage{amsbsy} 

\usepackage[all,matrix,cmtip]{xy}

\usepackage{color}

\usepackage{hyperref}
\hypersetup{
    colorlinks = true,
    allcolors = blue,
}

\definecolor{red}{rgb}{1,0,0}
\definecolor{blue}{rgb}{0,0,1}
\definecolor{dblue}{rgb}{0,0,0.4}
\definecolor{green}{rgb}{0,1,0}
\definecolor{black}{rgb}{0,0,0}
\definecolor{white}{rgb}{1,1,1}

\definecolor{brn}{rgb}{.8,.4,.0}
\definecolor{redo}{rgb}{1,.5,.0}
\definecolor{ddgrn}{rgb}{0,0.4,0}
\definecolor{dgrn}{rgb}{0,0.55,0}
\definecolor{dbl}{rgb}{0,0,0.5}

\usepackage[bbgreekl]{mathbbol}
\newcommand{\id}{\text{id}}

\newcommand{\Z}{\mathbb{Z}}
\newcommand{\C}{\mathbb{C}}
\newcommand{\R}{\mathbb{R}}

\newcommand{\M}{\mathbb{M}}

\renewcommand{\t}[1]{\widetilde{#1}} 
\newcommand{\ii}{\hspace{1pt}\mathrm{i}\hspace{1pt}}
\newcommand{\ee}{\hspace{1pt}\mathrm{e}}
\newcommand{\dd}{\hspace{1pt}\mathrm{d}}
 
\renewcommand{\>}{\rangle} 

\newcommand{\Rf}[1]{Ref.~\onlinecite{#1}}

\newcommand{\eqn}[1]{eqn.~(\ref{#1})}

\newcommand{\prt}{\partial}

\newcommand{\ie}{{\it i.e.~}} 
\newcommand{\eg}{{\it e.g.~}} 
\newcommand{\etc}{{\it etc.}}

\newcommand{\bpm}{\begin{pmatrix}}
\newcommand{\epm}{\end{pmatrix}}
\newcommand{\bmm}{\begin{matrix}}
\newcommand{\emm}{\end{matrix}}

\newcommand{\cA}{ {\cal A} } 
\newcommand{\cB}{ {\cal B} }
\newcommand{\cC}{ {\cal C} }

\newcommand{\cF}{ {\cal F} } 
\newcommand{\cG}{ {\cal G} } 
\newcommand{\cH}{ {\cal H} } 
\newcommand{\cK}{ {\cal K} } 
\newcommand{\cL}{ {\cal L} } 
\newcommand{\cM}{ {\cal M} }

\newcommand{\cV}{ {\cal V} } 
\newcommand{\cX}{ {\cal X} }

\usepackage{euscript}

\newcommand\eM          {\EuScript{M}}

\newcommand\eZ         {\EuScript{Z}}

\newcommand{\al}{\alpha}

\newcommand{\Del}{\Delta} 
\newcommand{\eps}{\epsilon}

\newcommand{\om}{\omega} 
\newcommand{\Om}{\Omega}

\newcommand{\frmbox}[1]{\begin{center}\fbox{\parbox{3.3in}{\parindent=0pt #1}}\end{center}}

\usepackage{pgf,tikz}
\usepackage{tikz-3dplot}
\usetikzlibrary{shapes.geometric}
\usetikzlibrary{decorations.pathmorphing}
\usetikzlibrary{positioning}
\usetikzlibrary{calc}
\usetikzlibrary{decorations.markings, arrows.meta}
\usetikzlibrary{decorations.pathreplacing} 

\tikzset{
  midarrow/.style={
    postaction={
      decorate,
      decoration={
        markings,
        mark=at position 0.6 with {\arrow{Latex[length=2.5mm,width=1.6mm]}}
      }
    }
  },
  vertex/.style={circle,fill=black,inner sep=1.6pt},
  vlabel/.style={font=\scriptsize, yshift=-1.8mm}
}

\usepackage[all]{xy}

\usepackage{euscript}

\newcommand{\w}{{\rm w}}

\newcommand{\Aut}{\mathrm{Aut}}
\renewcommand{\mod}{\ \mathrm{mod}\ }
\newcommand{\Sq}{\mathrm{Sq}}
\newcommand{\gSq}{\mathbb{Sq}}
\newcommand{\Bs}{\cB}
\newcommand{\RZ}{{\mathbb{R}/\mathbb{Z}}}

\newcommand\se[1]{\overset{\scriptscriptstyle #1}{=}}
\newcommand\hcup[1]{\underset{{\scriptscriptstyle #1}}{\smile}}

\newcommand\toZ[1]{\lfloor #1 \rceil}

\newcommand{\confstate}[1]{|{#1}\rangle^\mathrm{conf}}

\newtheoremstyle{plainupright}
{\topsep}      
{1.4\baselineskip}      
{\normalfont}  
{}             
{\bfseries}    
{.}            
{.5em}         
{}             

\theoremstyle{plainupright}

\newtheorem{theorem}{Theorem}[section]

\newtheorem{conjecture}[theorem]{Conjecture}
\newtheorem{lemma}[theorem]{Lemma}

\newtheorem{definition}[theorem]{Definition}

\newtheorem{remark}[theorem]{Remark}

\newcommand{\ZZ}{\mathbb Z}
\newcommand{\RR}{\mathbb R}
\newcommand{\supp}{\operatorname{supp}}

\makeatletter
\def\l@subsubsection#1#2{}
\makeatother

\begin{document}

\begin{titlepage}

\title{Holographic Theory of Mixed-Dimensional Statistics\\
and Conservation-Encoding Hopping-Operator Algebras}

\author{Hanyu Xue} \affiliation{Department of Physics, Massachusetts
Institute of Technology, Cambridge, Massachusetts 02139, USA}

\author{Xiao-Gang Wen} \affiliation{Department of Physics,
Massachusetts Institute of Technology, Cambridge, Massachusetts 02139,
USA}

\begin{abstract}

We develop a general framework for the statistics of mixed-dimensional
excitations subject to intertwined conservation laws, extending the familiar
Fermi statistics with conserved particle number. We define
statistics microscopically through a \emph{hopping-operator algebra}: a local
operator subalgebra (LOsA) generated by operators that locally move or deform
excitations while preserving the conservation law. Nontrivial statistics arise
when this subalgebra is nontrivial.

We first focus on LOsAs that encode \emph{pointed} conservation laws. These
give rise to invertible excitations, whose fusion rules are exactly those of
the symmetry defects of a higher group $\cG$. For such
$\cG$-conserved excitations in $d$-dimensional space, we show that the
corresponding LOsA -- and hence the statistics it defines -- is classified by a
cohomology class $[\omega] \in H^{d+2}(B\cG;\R/\Z)$, where changing $[\omega]$
by a coboundary corresponds merely to a rephasing of the local operators. We
further provide a holographic realization: excitations with this prescribed
conservation law and statistics live on the boundary of a $\cG$ higher-group
gauge theory in $(d+1)$-dimensional space, twisted by $[\omega]$.

More generally, non-pointed conservation laws and the associated statistics of
non-invertible excitations are defined by a pair: a LOsA together with its
excitation-complex representation. This is equivalent to the pair consisting of
a LOsA and its Hilbert-space representation, which is the data
defining a generalized symmetry. Consequently, non-pointed conservation laws
and their statistics in $d$-dimensional space are classified by fusion
$d$-categories, just as generalized symmetries are. The higher-group results
above are the fully-pointed special cases of this more general
classification.

\end{abstract}


\maketitle

\end{titlepage}

{
\small \tableofcontents
}

\

\

\section{Introduction}

Quantum statistics is a fundamental property of quantum systems. Particles in
$3$-dimensional space obey either Bose or Fermi statistics; the fermionic
electron, together with particle-number conservation, is the most familiar
example. In $2$-dimensional space, particles may instead be Abelian anyons
\cite{LM7701,W8257,H8483,ASW8422} or non-Abelian anyons
\cite{W8411,W8951,MR9162,W9102}. More generally, extended excitations, such as
strings, membranes, and higher-dimensional objects, can also exhibit
nontrivial statistics. For instance, loop excitations in $3$-dimensional space
support rich braiding processes
\cite{WL1437,JMR1462,WW1454,T14044385,FH211014654,WW211212148,KC241201886,
Xue2026StatisticsAbelian}.

In this paper, we study the statistics of point-like, string-like, and
higher-dimensional excitations, with particular emphasis on systems in which
excitations of different dimensions coexist and obey intertwined conservation
laws. To formulate the problem in full generality, we place the excitations on
a triangulated spatial complex $X$, mathematically known as a
\textit{combinatorial manifold}. In a microscopic lattice realization, the
degrees of freedom may be qudits placed on vertices, links, and higher
simplices. We take the qudit state $0$ to represent the absence of an
excitation, while the other qudit states represent the presence of various
excitations.

Mathematically, such excitations are labeled by $0$-chains, $1$-chains, and
higher chains in the spatial complex. This abstraction allows us to ignore the
details of the underlying qudit Hilbert space and, in particular, to avoid
assuming a tensor-product decomposition. We assume only that excitation
configurations label quantum states by chains in the spatial complex; we do
not even require that states labeled by different chains be orthogonal. This
paper emphasizes this more general point of view.

\begin{figure}[t]
\begin{center}
\includegraphics[scale=1.0]{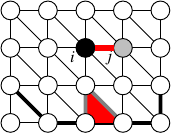}
\end{center}
\caption{A particle on site $i$ can be moved to site $j$ by a hopping operator
$t_{ji}$ supported on the link $(ij)$. A string can be locally deformed by a
hopping operator $t_\triangle$ supported on a triangle.}
\label{hop}
\end{figure}

The dynamics of these excitations are generated by hopping operators, which
locally move or deform them. Concretely, the hopping operators for particles
are supported on links, those for strings are supported on triangles, and so
on; see Fig.~\ref{hop}. We require hopping operators with disjoint supports to
commute. This condition defines a weak notion of locality: hopping operators
are \emph{weakly local} if they have finite support and commute whenever their
supports do not overlap. Because we do not assume a tensor-product structure, \textit{support} is an additional input data assigned to a hopping operator.
 Weakly local hopping operators generate a local
operator subalgebra, abbreviated LOsA, that encodes a conservation law. The
use of a proper \emph{subalgebra} is essential: the algebra of \emph{all}
local operators would encode the trivial case with no conservation law. When a
selected set of local operators generates a proper subalgebra, that subalgebra
encodes a nontrivial conservation law (equivalently, a nontrivial generalized
symmetry up to holo-equivalence \cite{JW191213492,KZ200514178}).

In analogy with the distinction between groups and their representations, a
LOsA also admits representations. One class consists of the
\textit{realizations} of a LOsA as operators on a Hilbert space. A more
abstract class of representation is given by purely geometric data, such as
chains and cycles, which encode the conservation law. Such a representation is
called an \emph{excitation complex}\footnote{This notion was introduced in
Ref.~\cite{Xue2026StatisticsAbelian} under the name \textit{excitation
pattern}.}. The term \textit{statistics} then refers to an invariant of a LOsA
that classifies its various realizations relative to a fixed excitation
complex (i.e., a fixed conservation law).

If the total Hilbert space admits a tensor-product decomposition
\begin{align}
\cV_{\mathrm{tot}}=\bigotimes_i \cV_i ,
\end{align}
then one can define \emph{strongly local} hopping operators as those generated
by operators acting only on the tensor factors within their support. Strongly
local hopping operators are, of course, weakly local. In this paper, however,
we work mainly with weakly local hopping operators, which we will simply call
hopping operators.

Statistics then arise as a classification of \emph{hopping algebras}, or
equivalently of LOsAs. Without a conservation law, the allowed local operators
form the full local operator algebra and describe trivial bosonic statistics.
With a conservation law, however, the allowed hopping operators are restricted
to a nontrivial LOsA. Such a LOsA encodes not only the conservation law but
also, in general, nontrivial statistics. In this sense, the statistics of
excitations are determined by the LOsA of their hopping operators
\cite{LW0316,FH211014654,KC241201886,Xue2026StatisticsAbelian}.

In spacetime, conserved excitations sweep out worldlines, worldsheets, and
higher-dimensional worldvolumes. These trajectories are chains in the
spacetime complex. Equivalently, they can be described by their
Poincar\'e-dual cochains $f_i$, which represent excitation \emph{currents}.
For example, a conserved particle in $d$-dimensional space has a conserved
current represented by a $d$-cocycle $f_d$ in spacetime. More generally, a
conserved codimension-$i$ excitation has a conserved current represented by an
$i$-cocycle $f_i$.

The fusion rules that encode conservation laws fall into several classes, for
which we adopt the following terminology:
\begin{enumerate}
\item \textit{Abelian.}
These are the fusion rules of symmetry defects of higher Abelian groups, whose
classifying spaces have the form
$K(A_1,1)\times K(A_2,2)\times \cdots$. In this case, excitations of different
dimensions have independent fusion rules and are described by cocycles $f_q\in Z^q(M,A_q)$, although their statistics may still
be mixed.

\item \textit{Fully pointed, or invertible.}
These are the fusion rules of symmetry defects of higher groups. Here,
intersections or junctions of higher-dimensional excitations may act as
sources for lower-dimensional ones, while the fusion rules remain pointed.
Equivalently, the conserved excitations are invertible. An excitation $a$ is
invertible if fusion with $a$ merely permutes the set of excitation types,
or, equivalently, if there exists an excitation $\bar a$ such that
\begin{align}
a\otimes \bar a=\mathbf 1 .
\end{align}
If all excitations in a system are invertible, their fusion rules are
group-like, or more generally higher-group-like, and the conservation law is
described by a higher group $\cG$ \cite{BLm0307200,KT13094721}. The Postnikov
data of $\cG$ encode how currents of different cochain degrees are coupled:
\begin{align}
\dd f_1 &=0, \nonumber\\
\dd f_2 &=k_3(f_1), \nonumber\\
\dd f_3 &=k_4(f_1,f_2), \nonumber\\
&\hspace{1em}\vdots
\end{align}
Thus invertible excitations can be viewed as symmetry defects of the higher
group $\cG$, and their conservation laws are precisely the fully-pointed
fusion rules of those defects.

\item \textit{Non-invertible.}
This is the most general case. Fusion may proceed through multiple channels
and is described by higher categories. Typical examples include non-Abelian
anyons. The traditional terminology is somewhat misleading: historically,
``non-Abelian anyon'' refers to the fact that braiding can act by matrices
rather than by mere phase factors. This characterization is difficult to
generalize to extended excitations, whose statistical processes are richer
than braiding. We therefore find it more natural to understand non-Abelian
statistics in terms of non-invertible fusion. For ordinary anyons, ``Abelian''
and ``pointed'' fusion rules coincide, but in general mixed-dimensional
systems the two notions are distinct.
\end{enumerate}

This paper studies primarily conservation laws described by fully-pointed
fusion rules. Our main result is a systematic construction of the statistics,
defined by the LOsA of hopping operators, for invertible excitations of mixed
dimensionalities. Let the excitations live in a $D=d+1$ dimensional spacetime
$M^{D}$, and let their fully-pointed conservation laws be modeled by symmetry
defects of a higher group $\cG$. Then their statistics are described by an
$\mathbb R/\mathbb Z$-valued $(d+2)$-cocycle on the classifying space $B\cG$:
\begin{align}
[\om_{d+2}]\in H^{d+2}(B\cG,\mathbb R/\mathbb Z).
\end{align}
The corresponding statistical phase is controlled by the Wess-Zumino-Witten
term
\begin{align}
\exp\left(
2\pi \ii \int_{N^{d+2}}\om_{d+2}(\{f_i\})
\right),
\qquad
M^{D}=\partial N^{d+2}.
\end{align}
This description combines the geometry of the conserved excitations, encoded by
the higher group $\cG$, with their statistical phases, encoded by
$[\om_{d+2}]$.

We know that the cohomology class $[\om_{d+2}]$ describes the anomaly of
the higher group $\cG$ symmetry. Our result indicates that the statistics of
the $\cG$-symmetry defects are described by precisely the same cohomology
class, $[\om_{d+2}]$. In other words, a symmetry anomaly manifests itself as
nontrivial statistics of symmetry defects \cite{LG1209,BZ14104540,W181202517}
(see Section \ref{anosta}).

The formulation is intrinsically holographic. The cocycle $\om_{d+2}(\{f_i\})$
is the topological action of a higher-group gauge theory in one higher
dimension. Thus the dynamics of excitations in $(d+1)$-dimensional spacetime,
together with their conservation laws and statistics, can be realized
holographically on the boundary of an $\om_{d+2}$-twisted $\cG$ higher-group
gauge theory in $(d+2)$-dimensional spacetime. This construction generalizes
Dijkgraaf-Witten theory \cite{DW9093,FreedQuinn1993} to higher groups
\cite{KT13094721}, and is consistent with previous results on string
statistics
\cite{T14044385,FH211014654,WW211212148,KC241201886,Xue2026StatisticsAbelian}.
Whereas recent works have developed an operator or Hilbert-space approach to
these statistics \cite{Xue2026StatisticsAbelian}, our approach is
field-theoretic and holographic, making it particularly well suited to
mixed-dimensional excitations.

This perspective has several consequences:
\begin{enumerate}
\item
\textbf{A unified language in arbitrary dimensions.} In our formulation, the
statistics of point particles in $(d+1)$-dimensional spacetime are described by
$H^{d+2}(K(A,d),\mathbb R/\mathbb Z)$, where $A$ is the Abelian fusion group
and $B\cG=K(A,d)$ encodes the conservation law. For spatial dimension $d=2$,
this gives $H^4(K(A,2),\mathbb R/\mathbb Z)$, which reproduces the group of
quadratic functions on $A$. This agrees with the familiar classification of
pointed braided tensor categories (PBTC) with fusion group $A$
\cite{JoyalStreet1993,EtingofGelakiNikshychOstrik2015}:
\begin{align}
H^4(K(A,2),\mathbb R/\mathbb Z)
&\cong \mathrm{Quad}(A,\mathbb R/\mathbb Z) \nonumber\\
&\cong \mathrm{PBTC}(A).
\end{align}

\item
\textbf{String statistics in $d$-dimensional space.}
For strings with Abelian fusion group $A$, the statistics are governed by
\begin{align}
H^{d+2}(K(A,d-1),\mathbb R/\mathbb Z).
\end{align}
This agrees with the classification obtained in
Refs.~\cite{KC241201886,Xue2026StatisticsAbelian}.

\item
\textbf{String statistics and the $\w_3$-structure of spacetime.}
For $A=\mathbb Z_2$ in $4$-dimensional spacetime ($d=3$), the nontrivial
string statistics are described by
\begin{align}
\frac12 f_2\Sq^1 f_2 \in H^5(K(\Z_2,2),\mathbb R/\mathbb Z).
\end{align}
The extension ambiguity of the corresponding WZW term on a closed $5$-manifold
$W^5$ is controlled by
\begin{align}
\int_{W^5} f_2\Sq^1 f_2
\se{2}
\int_{W^5} f_2\,\w_3(TW)
 .
\end{align}
Thus nontrivial $\mathbb Z_2$ string statistics require an oriented spacetime
equipped with a trivialization of $\w_3$ (see Section \ref{z2strw3}). This is
analogous to the spin structure required for Fermi statistics.

\item
\textbf{Mixed-dimensional conservation laws.}
Higher groups naturally couple excitations of different dimensions. A key
example is a twisted $\mathbb Z_2$ particle-string system in $d$-dimensional
space (assuming $d>1$), whose currents satisfy
\begin{align}
\dd f_{d-1} & \se{2} 0, \nonumber\\
\dd f_d & \se{2} \Sq^2 f_{d-1} .
\end{align}
We believe that the possible statistics form a $\mathbb Z_4$ group for
$d>2$, and a $\mathbb Z_2$ group for $d=2$
(see Section \ref{partstr}). This example is closely related to $p$-wave
topological superconducting strings and to fermionic bosonization
\cite{GuWen2014,KTT1429,W161201418,KT170108264,WangGu2018}.

\item
\textbf{Hopping algebras and WZW terms.}
Our result can also be read as a classification of LOsAs that encode
invertible conservation laws in $d$-dimensional space. If the conservation law
is described by a higher group $\cG$, then the corresponding LOsAs are
classified by
\begin{align}
[\om_{d+2}]\in H^{d+2}(B\cG,\mathbb R/\mathbb Z).
\end{align}
We will present two complementary approaches to these statistics: one based on
WZW terms, and one based on the LOsA of hopping operators, the latter serving
as our microscopic definition.
\end{enumerate}

More broadly, the statistics of potentially non-invertible excitations are
still defined by the LOsA of their hopping operators. In this case, however,
the LOsA encodes a general conservation law whose fusion rules need not be
pointed. Strongly local LOsAs in $d$-dimensional space are conjectured to be
classified by braided fusion $d$-categories in the trivial Witt class
\cite{KZ200514178,CW220303596,FT220907471}. A collection of excitations,
together with its conservation laws (encoded in the excitation complex) and its
top-dimensional coherence data (encoded in the statistics), is described by a
fusion $d$-category. This leads to the general result that generalized
statistics in $d$-dimensional space are described by fusion $d$-categories,
with higher groups arising as the pointed, or invertible, special case. Since a
fusion $d$-category defines a $(d+2)$-dimensional state-sum model, the dynamics
of generalized statistical excitations can be simulated on the boundary of a
state-sum model in $(d+2)$-dimensional spacetime. Higher gauge theories are the
invertible special case of this construction.

\begin{figure}[t]
\includegraphics[scale=0.7]{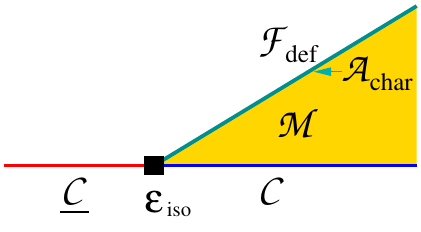}
\caption{An isomorphic holographic decomposition $\eps_{\mathrm{iso}}$ of a
system $\underline{\cC}$, or of a quantum field theory, expresses
$\underline{\cC}$ as a stacking of two boundaries -- the physical boundary
$\cC$ and the symmetry boundary $\cF_{\mathrm{def}}$ -- through a bulk
topological order $\eM$, also called the symmetry topological order:
$\underline{\cC}=\cF_{\mathrm{def}}\otimes_{\eM}\cC$ \cite{KZ150201690}. This
decomposition reveals the generalized symmetry of $\underline{\cC}$. The
boundary excitations on the symmetry boundary form a fusion higher category
$\cF_{\mathrm{def}}$, which characterizes the generalized symmetry of
$\underline{\cC}$ \cite{TW191202817}. Equivalently, $\cF_{\mathrm{def}}$
describes the fusion of symmetry defects \cite{KZ200514178}. The bulk symmetry
topological order $\eM=\eZ_2(\cF_{\mathrm{def}})$ describes the
holo-equivalence class of the symmetry \cite{JW191213492,KZ200514178}. The
boundary $\cF_{\mathrm{def}}$ is induced by a Lagrangian condensable algebra
$\cA_{\mathrm{char}}$, whose anyons are interpreted as the charges of the
symmetry $\cF_{\mathrm{def}}$ \cite{KZ200514178,FT220907471}.}
\label{CCmorph}
\end{figure}

The above categorical description of statistics echoes the
symmetry/topological-order correspondence for generalized symmetries. To make
the connection precise, let us briefly review this correspondence.

In 2014, a correspondence between non-invertible gravitational anomalies and
topological orders in one higher dimension was developed \cite{KW1458}. Using
this correspondence, one can reveal hidden gravitational anomalies in a system
through a \emph{holographic isomorphic decomposition} \cite{KZ150201690} (see
Fig.~\ref{CCmorph}), now often called the \emph{sandwich picture}
\cite{FT220907471}. It was later understood that non-invertible gravitational
anomalies are generalized symmetries (up to holo-equivalence
\cite{JW191213492,KZ200514178}), which may be anomalous, higher-group, and/or
non-invertible. Thus the sandwich picture of Fig.~\ref{CCmorph} can be used to
expose and characterize hidden generalized symmetries in a system
\cite{TW191202817,JW191213492,LB200304328,KZ200514178,GK200805960,
AS211202092,CW220303596,FT220907471,CW221214432}. In this language, the bulk
topological order -- a braided fusion $d$-category -- describes a
holo-equivalence class of symmetries \cite{KZ200514178,FT220907471}, while the
symmetry boundary specifies a particular symmetry within that class.

Traditionally, a symmetry is defined as invariance under a set of
transformations. For lattice quantum systems, however, it is more intrinsic to
define a symmetry as an isomorphism class of strongly local LOsAs
\cite{KZ200514178,CW220303596}. Refs.~\cite{DHR1969_fields_I,
DHR1969_fields_II} showed that isomorphism classes of LOsAs in one-dimensional
space are classified by braided fusion $1$-categories. More generally, it has
been proposed that isomorphism classes of strongly local LOsAs in
$d$-dimensional space are classified by braided fusion $d$-categories in the
trivial Witt class \cite{KZ200514178,CW220303596,FT220907471}. Different
Hilbert-space representations of a LOsA -- equivalently, different symmetry
boundaries of the bulk topological order determined by the LOsA -- are
described by fusion $d$-categories whose centers are the corresponding braided
fusion $d$-category. These fusion $d$-categories describe the symmetry defects
of the associated generalized symmetries. In the present paper, those symmetry
defects are precisely the mixed-dimensional excitations, whose potentially
non-invertible conservation laws and generalized statistics are described by
fusion $d$-categories.

The rest of the paper is organized as follows. We first review how Fermi
statistics arise macroscopically from a twisted higher gauge theory in one
higher dimension, and how they are defined microscopically by the algebra of
hopping operators supported on links. We then extend this holographic
framework to invertible excitations of mixed dimensionalities. Finally, we
analyze several examples: particles in $2+1$ dimensions, strings in $3+1$
dimensions, and mixed particle-string systems in $d+1$ dimensions.


\section{Notations and Conventions} \label{notation}

We utilize algebraic topology constructs such as cochains, coboundaries, cocycles, higher cup products $\hcup{k}$, Steenrod squares $\Sq^k$ (for $\Z_2$-valued cocycles), generalized Steenrod squares $\gSq^k$ (for generic coefficients), and Bockstein homomorphisms $\Bs_N$. We frequently abbreviate the ordinary cup product $a\smile b$ as $ab$. A brief review of these operations is provided in Appendix~\ref{cochain}.

We use some standard notation:
\begin{itemize}
    \item $K(A,n)$ denotes an Eilenberg-MacLane space where $\pi_n=A$ and all
other $\pi_i=0$.  We use $A^{(n)}$ to denote the $n$-group whose classifying
space $B A^{(n)}$ is $K(A,n)$.

    \item $C^n(K,A)$, $Z^n(K,A)$, and $B^n(K,A)$ denote the spaces of
$A$-valued $n$-cochains, $n$-cocycles, and $n$-coboundaries on a complex $K$.
Cohomology is $H^n(K,A) = Z^n(K,A)/B^n(K,A)$.

    \item $X \se{n} Y$ denotes equality modulo $n$ ($X-Y \equiv 0 \pmod n$).

    \item $X \se{\dd} Y$ denotes equality up to a coboundary. $X \se{n,\dd} Y$
indicates equality modulo $n$ and up to a coboundary.

    \item $\Z_n = \left\{\toZ{-\frac n2+1}, \dots, \toZ{\frac n2}\right\}$
represents a chosen integer lift of the cyclic group $Z_n$, where $\toZ{x}$ is
the integer closest to $x$ with $\toZ{1/2}=0$. Operations on $\Z_n$-valued
quantities use standard integer addition before the modulo $n$ equivalence is
applied.  In other words a $\Z_n$ valued cohain $f^{\Z_n}$ is always regarded
as $f^{\Z_n}-n\toZ{f^{\Z_n}/n}$.  In the paper, we some times use
\begin{align}
 \t f^{\Z_n} = f^{\Z_n}-n\toZ{f^{\Z_n}/n}
\end{align}
to stress that $ \t f^{\Z_n}$ is a integer lift of $\Z_n$-valued cohain
$f^{\Z_n}$.

    \item $\RZ = (-\frac12,\frac12]$ serves as a real lift of $U(1)$.  $\RZ$
additively represents $U(1)$.  In other words a $\RZ$ valued cohain $f$ is
always regarded as $f-\toZ{f}$.  Thus, $\frac{1}{n} f^{\Z_n}$ means the
integer representative of $f^{\Z_n}$ is divided by $n$, yielding an
$\RZ$-valued cochain.

\end{itemize}

\section{Holographic perspective on fermions: path-integral formalism} 
\label{bosonization}

Let us review higher-dimensional bosonization \cite{W161201418, KT170108264,
LW180901112}. We treat the fermion as a point-like excitation with
$\Z_2$ conservation. Its worldline is a $\Z_2$-valued cycle in a
$d+1$-dimensional spacetime, whose Poincar\'e dual is a $\Z_2$-valued $d$-cocycle
$f_d$.

\subsection{Spin structure from WZW term}
\label{bosonization1}

Typically, a bosonic field $f_d$ describes a bosonic excitation via the standard path integral:
\begin{align}
	\label{Zb}
	Z(M^{d+1}) = \sum_{ f_d \in Z^d(M^{d+1},\Z_2) } \ee^{- \int_{M^{d+1}} \cL( f_d)}.
\end{align}
Refs.~\cite{KT170108264,LW180901112} proposed that $f_d$ can alternatively describe a fermion by appending a WZW topological term $ \ee^{\ii \pi \int_{N^{d+2}} \Sq^2 f_d }$
\begin{align}
	\label{Zf1}
	Z(M^{d+1}) &= \sum_{ f_d \in Z^d(M^{d+1},\Z_2) } 
\ee^{-\int_{M^{d+1}} \cL( f_d)} \ee^{\ii \pi \int_{N^{d+2}} \Sq^2 f_d },
\end{align}
where $N^{d+2}$ is a \emph{regular extension} of the closed manifold $M^{d+1}$,
defined such that $\prt N^{d+2} = M^{d+1}$ and the Stiefel-Whitney classes of the
tangent bundle of $N^{d+2}$ restrict to those of $M^{d+1}$.  For \eqref{Zf1} to
be a rigorous path integral strictly defining $d+1$-dimensional physics, the bulk
amplitude must be independent of how $f_d$ is extended into $N^{d+2}$
\cite{GK150505856}.  Since $\Sq^2 f_d$ is in general a cocycle instead of a
coboundary, the term $ \ee^{\ii \pi \int_{N^{d+2}} \Sq^2 f_d } = \pm 1$ does
depend on the extension.

To fix this problem, we rewrite \eqref{Zf1} using the identity $\ee^{\ii \pi
\int_{M^{d+1}}  f_d s + \ii \pi \int_{N^{d+2}} f_d(\w_2+\w_1^2)} =1$
involving the spacetime's spin structure $s$ and Stiefel-Whitney classes
$\w_n$:
\begin{align}
	\label{Zf}
	Z(M^{d+1},s) &= \sum_{ f_d \in Z^d(M^{d+1},\Z_2) } \ee^{-\int_{M^{d+1}} \cL( f_d)} \nonumber\\
	&\quad  \ee^{\ii \pi \int_{M^{d+1}} f_d s + \ii \pi \int_{N^{d+2}} \bigl(\Sq^2 f_d + f_d(\w_2+\w_1^2)\bigr)}, \nonumber\\
	\dd s &\se{2} \w_2+\w_1^2 \quad \text{on } M^{d+1}.
\end{align}
The modified WZW term $\ee^{ \ii \pi \int_{N^{d+2}} \bigl(\Sq^2 f_d +
f_d(\w_2+\w_1^2)\bigr)}$ does not depend on the extension since $\Sq^2 f_d +
(\w_1^2+\w_2) f_d \se{2,\dd}0$ on $N^{d+2}$ due to Wu's relation. 
  Thus, this bosonization is well-defined only when $\w_2+\w_1^2$
is a coboundary.  In other words, the WZW term requires the spacetime to be
spin, consistent with fermions requiring a spin manifold.  This spin manifold
requirement from the WZW term indirectly suggests that the WZW term identifies
$f_d$ as a current of fermionic particles.

We notice that the pure bulk action amplitude
\begin{align}
	\label{LNd2}
	\ee^{\ii \pi \int_{N^{d+2}} \Sq^2 f_d}
\end{align}
defines a $(d+2)$-dimensional higher gauge theory with $f_d$ as the gauge
field.  The gauge group is a $d$-group $\Z_2^{(d)}$ describing a $(d-1)$-form
symmetry \cite{NOc0605316,GW14125148}, whose classifying space is $B\Z_2^{(d)}
= K(\Z_2,d)$.  We usually describe the $\Z_2$ conservation of the fermion by a
$\Z_2$ 0-form symmetry.  However, here, we equivalently describe the $\Z_2$
conservation of the fermion in $d$-dimensional space by a $d$-group
$\Z_2^{(d)}$.  This is why we can use a higher-form $\Z_2^{(d)}$-gauge theory
to describe the $\Z_2$-conserved fermions. In fact, the gauge potential $f_d$
of the $\Z_2^{(d)}$-gauge theory is the fermion current, and the
flat-connection condition is the $\Z_2$ conservation.

The field $f_d$ (the fermion current) can be viewed as the pullback $f_d =
\phi^* x_d$ of the canonical $d$-cocycle $x_d$ on the classifying space
$K(\Z_2,d)$ via a map $\phi: N^{d+2}\to K(\Z_2,d)$. Thus, the action represents
a nonlinear $\sigma$-model whose target space is the classifying space of the
higher group $\cG = \Omega K(\Z_2,d)=\Z_2^{(d)}$. 

When the higher gauge theory describing the conservation law is untwisted, the
conserved boundary objects are bosonic. Introducing the cocycle twist
$\Sq^2 x_d$ endows the boundary excitations with fermionic statistics.
This holographic perspective naturally combines statistics with its
prerequisite conservation law: the higher group describes the conservation
law, while the $\RZ$-valued $(d+2)$-cocycle describes the statistics. It also
allows us to generalize the above results on Fermi statistics for particles to
generalized statistics for extended objects and mixed-dimensional excitations.

\subsection{Statistics from WZW term}

In the preceding discussion, we derived the spin structure from the WZW term
$\ee^{\ii \pi \int_{N^{d+2}} \Sq^2 f_d}$. In this section, we will demonstrate
how this WZW term gives rise to Fermi braiding statistics. We first restrict
ourselves to a specific class of particle configurations where the $\Z_2$
current $f_d$ is a coboundary:
\begin{align}
 f_d \se{2} \dd a, \ \ \ \ a \in C^{d-1}(N^{d+2}, \Z_2).
\end{align}
Using the generalized Steenrod square and Eq.~\eqref{Sqd1}, we find
\begin{align}
\Sq^2 f_d = \gSq^2 f_d \se{2} \dd \gSq^2 a.
\end{align}
Thus, the WZW term reduces to an action strictly on the spacetime boundary
$M^{d+1}$ \cite{LW180901112}:
\begin{align}
 \ee^{\ii \pi \int_{N^{d+2}} \Sq^2 f_d} =
 \ee^{\ii \pi \int_{M^{d+1}} \gSq^2 a} .
\end{align}

To see how Fermi statistics emerge from the term $\ee^{\ii \pi \int_{M^{d+1}} \gSq^2 a}$,
let us specialize to $d=2$-dimensional space, where the action
simplifies to
\begin{align}
 \ee^{\ii \pi \int_{M^3} \gSq^2 a} =
 \ee^{\ii \pi \int_{M^3} a\dd a} =
 \ee^{\ii 2\pi \int_{M^3} \frac12 a\dd a} 
.
\end{align}
This is a fractional Chern-Simons term, well known to give rise to Fermi
statistics.

Specifically, we note that $f_2 = \dd a$ is the conserved $\Z_2$ current of
the particles in 2-dimensional space. The relation $f_2 = \dd a$ implies that
the particle acts as a unit flux, generating a circulating gauge field $a$
around itself. The coupling term $\frac12 a\dd a = \frac12 a f_2$ indicates
that a particle carrying current $f_2$ couples to the gauge field $a$ produced
by other particles exactly as a charge-$\frac12$ object. Thus, the particle
behaves as a bound state of a unit flux and a half-charge. 

Moving such a particle in a full loop around another particle induces an
Aharonov-Bohm phase of 
$2\pi \times \left(\frac12 \times 1\right) \times 2$.
Here, $2\pi \times (\frac12 \times 1)$ is the phase acquired from moving the
charge-$\frac12$ around the unit flux, and the overall factor of $2$ accounts
for their mutual statistics (moving the charge of the first particle around the
flux of the second, plus moving the flux of the first around the charge of the
second). Therefore, moving a particle halfway around another (\ie, exchanging
the two particles) induces a phase of $\pi\times  (\frac12 \times 1) \times 2 =
\pi$, which precisely corresponds to Fermi statistics. 

In the next section, we will formalize this heuristic argument and generalize
it to arbitrary dimensions: \frmbox{For general $d$-dimensional space, the
topological term in action amplitude 
\begin{align}
\ee^{\ii \pi \int_{M^{d+1}} \gSq^2 a},\
a \in C^{d-1}(M^{d+1}, \Z_2)  
\end{align}
endows Fermi statistics to the particle described by the $\Z_2$-conserved
current $f_d=\dd a$.  } We will show that in the Hamiltonian description of
the particles, $\ee^{\ii \pi \int_{M^{d+1}} \gSq^2 a}$ and the related WZW
term induce non-trivial phases in the hopping amplitudes. These phases lead to
a non-trivial statistical algebra for the hopping operators, strictly
enforcing Fermi statistics.

\section{Holographic perspective of fermions: Hilbert space formalism}
\label{sec:WZW_operator_fermion}

In Sec.~\ref{bosonization} we linked the fermionic statistics of a $\mathbb
Z_2$-conserved point excitation to a particular WZW term
\begin{equation}
	\exp\left(2\pi \ii \int_{N^{d+2}} \frac12 \Sq^2 f_d\right) .
\end{equation}
This statement should be read with some care.  The word ``fermion'' has several
historically different meanings: an anti-symmetric wave function,
anti-commuting operators, an anti-commuting Grassmannian field, a particle with
topological spin $\frac12$, or a particle requiring a spin structure for
spacetime.  We used the spin structure in Sec.~\ref{bosonization1} to argue that the above WZW term
gives rise to Fermi statistics; this gives an indirect argument.

In this section we use the axiomatic framework of
\Rf{Xue2026StatisticsAbelian}, which defines statistics directly from
configuration states and the algebra of hopping operators.  This framework was
first introduced for point-like fermions \cite{LW0316}, and it also generalizes
naturally to extended excitations such as loops \cite{FH211014654,KC241201886}.
Our goal is to connect this operator definition of statistics to the WZW term
in one higher dimension. 


\subsection{Review: the framework of statistics}
\label{subsec:mp_pattern}

\begin{figure}[t]
	\begin{center}
		\includegraphics[scale=1.0]{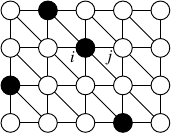}
	\end{center}
	\caption{The dots represent the sites that form the space $X$. On each
site $i$ there are two states $|0\>_i,|1\>_i$ forming a $\Z_2$ group.  The
configuration group is given by $A = \frac{\otimes_i \Z_2}{\Z_2^\text{total}}$.
The elementary subset $S$ is formed by all the links $(ij)$ (the 1-simplices):
$S =\{ (ij) \}$.
The map $\prt$ is given by $\prt (ij) = |1\>_i\otimes |1\>_j$. The support of
the link $(ij)$ consists of sites $i$ and $j$.} \label{StaAx}
\end{figure}

In the setting of Ref.~\cite{Xue2026StatisticsAbelian}, one considers $p$-dimensional excitations with invertible Abelian fusion rule, characterized by a finite Abelian group $G$. In our main theorem, the geometry of space $X$ is modeled as a combinatorial $d$-sphere, but at this moment we treat $X$ as a generic finite cell complex. The geometric shapes of excitations, called configurations, are labeled by $p$-dimensional boundaries $B_p(X,G)$.
These configurations form an Abelian group, denoted by
\begin{equation}
	A=B_p(X,G).
\end{equation}
Hopping operators are associated with $(p+1)$-dimensional objects.  Instead of using all of $C_{p+1}(X,G)$ as elementary labels, Ref.~\cite{Xue2026StatisticsAbelian} uses the elementary subset
\begin{equation}
	S=G_0\times X_{p+1}\subset C_{p+1}(X,G),
\end{equation}
where $X_{p+1}$ is the set of $(p+1)$-simplices and $G_0\subset G$ is a
generating subset.  The set $G_0\times X_{p+1}$ still generates $C_{p+1}(X,G)$.
There are two reasons for using this elementary set.  First, it reflects the
locality intuition that an extended excitation is built from local pieces.
Second, in many-body physics an operator labeled by $s\in S$ should have a
well-defined geometric support $\supp(s)$, and this is most natural for
simplex-supported operators.

These data are packaged into the following definition (see Fig.~\ref{StaAx} for a $p=0$ case).
\begin{definition}\label{defexcitationPattern}
	An \textit{excitation complex} $m$ consists of the data $(A,S,\partial,\supp)$:
	\begin{enumerate}
		\item a finite Abelian group $A$, called the configuration group;
		\item a finite set $S$ and a map $\partial:S\to A$ such that $\{\partial s\mid s\in S\}$ generates $A$;
		\item a topological space $\cX$ and a subspace $\supp(s)\subset \cX$ for each $s\in S$.
	\end{enumerate}
\end{definition}

The basic example above will be denoted by
\begin{equation}
	m_p(X,G):\ \
	(A=B_p(X,G),\; S=G_0\times X_{p+1},\; \partial,\; \supp).
\end{equation}
An excitation complex only describes geometric data and a conservation law.  A physical system realizes this excitation complex if these data label a collection of states and operators satisfying the following axioms\footnote{This definition is slightly different from that in Ref.~\cite{Xue2026StatisticsAbelian}; see Appendix \ref{sec: majorana fermions}.}.

\begin{definition}\label{defRealization}
	A \textit{realization}\footnote{In this paper, the word "realize" always refers to this specific terminology.} of the excitation complex $m=(A,S,\partial,\supp)$ consists of a Hilbert space $\mathcal H$, a collection of normalized \textit{configuration states} $\{|a\rangle\mid a\in A\}$ in $\cH$, and a collection of \textit{hopping operators} $\{U(s)\mid s\in S\}$, satisfying the following two axioms.
	\begin{itemize}
		\item \textbf{Configuration axiom:} for any $s\in S$ and $a\in A$,
		\begin{equation}\label{eqChangeConfig}
			U(s)|a\rangle=e^{i\theta(s,a)}|a+\partial s\rangle
		\end{equation}
		for some $\theta(s,a)\in\RR/2\pi\ZZ$.
		\item \textbf{Locality axiom:} for any $s_1,s_2,\ldots,s_k\in S$ satisfying
		\begin{equation}
			\supp(s_1)\cap\supp(s_2)\cap\cdots\cap\supp(s_k)=\emptyset,
		\end{equation}
		one has
		\begin{align}\label{axiomLocalityIdentity}
			[U(s_k),[\cdots,[U(s_2),U(s_1)]]]=1\in U(\mathcal H),
		\end{align}
		where $[a,b]=a^{-1}b^{-1}ab$.
	\end{itemize}
\end{definition}

The configuration axiom can be viewed as conservation rules. For example, in the quasiparticle case ($p=0$), it means that a particle and its antiparticle are always created in pairs by hopping operator $U(s)$ at $\partial s$, so the total number is always zero. Currently, it only describes Abelian fusion rules; we discuss a potential generalization in Appendix~\ref{AbNonAbElementary}.

For $k=2$, the locality axiom says that $U(s_1)$ and $U(s_2)$ commute whenever $\supp(s_1)\cap\supp(s_2)=\emptyset$.  This captures the intuition that $U(s)$ acts only on degrees of freedom inside $\supp(s)$, without requiring us to specify a bosonic tensor-product decomposition of the Hilbert space.  Moreover, $U(s)$ should be regarded as a finite-depth local quantum circuit.  Therefore $[U(s_2),U(s_1)]$ is supported near $\supp(s_1)\cap\supp(s_2)$, and hence
\begin{equation}
	[U(s_3),[U(s_2),U(s_1)]]=1
\end{equation}
whenever $\supp(s_1)\cap\supp(s_2)\cap\supp(s_3)=\emptyset$.  The higher nested-commutator conditions are interpreted similarly.

The configuration axiom assigns a collection of phases $\{\theta(s,a)\}$ to each realization of $m$, and the locality axiom imposes linear equations on them. For example, if $\supp(s_1)\cap\supp(s_2)=\emptyset$, then
\begin{equation}
	U(s_2)U(s_1)|a\rangle=U(s_1)U(s_2)|a\rangle
\end{equation}
implies
\begin{equation}
	\theta(s_1,a)+\theta(s_2,a+\partial s_1)
	=
	\theta(s_2,a)+\theta(s_1,a+\partial s_2).
\end{equation}

The solution space of these equations, denoted by $R(m)$, is a subgroup of $(\RR/2\pi\ZZ)^{S\times A}$ and can be decomposed into a continuous part and a discrete part $T^*$:
\begin{equation}
	R(m)\simeq (\RR/2\pi\ZZ)^n\oplus T^*(m).
\end{equation}
More precisely, there is a canonical surjective map $R(m)\to T^*(m)$ such that $(\RR/2\pi\ZZ)^n$ is the kernel. In other words, we have the canonical short exact sequence
\begin{equation}\label{eq:short exact sequence}
	0\to(\RR/2\pi\ZZ)^n\to R(m)\to T^*(m)\to 0.
\end{equation}

During computation, it appears that the continuous part varies for different spatial triangulations $X$, while the discrete part $T^*(m)$ only depends on the topology. Thus we define $T^*(m)$ as the \textit{classification of statistics}, and for a realization $\in R(m)$, its \textit{statistics} is defined as the image under the map $R(m)\to T^*(m)$. The discrete nature of $T^*(m)$ makes statistics very robust:
continuously modifying $|a\rangle$ or $U(s)$ may change each phase $\theta(s,a)$, but the image under $R(m)\to T^*(m)$ remains the same. For $\Z_2$ conserved particles in $2$ spatial dimensions, one expects
\begin{equation}
	T^*(m)\simeq \ZZ_4,
\end{equation}
corresponding to boson, semion, fermion, and anti-semion. 

The statistics of realizations are diagnosed by linear functions $e:R(m)\to
\RR/2\pi\ZZ$ that are zero on the continuous subgroup $(\RR/2\pi\ZZ)^n$, and
whose values are called \textit{statistical phases}. Physically, a statistical
phase is more convenient to be denoted as the phase of a sequence of excitation
operators $U(s_1)^{\epsilon_1},\cdots,U(s_n)^{\epsilon_n}$ acting on a
particular configuration state $|a\rangle$, where $\epsilon_i=\pm1$. Here we
require that $U(s)$ and $U(s)^{-1}$ always appear in pairs, so the final state
is collinear to the initial state $|a\rangle$.  In other words, the statistical
phase is
\begin{equation}\label{eq: statistical phase}
	\langle a|U(s_n)^{\epsilon_n}\cdots U(s_2)^{\epsilon_2}U(s_1)^{\epsilon_1}|a\rangle.
\end{equation}
We view Eq.~\eqref{eq: statistical phase} as the value of
\begin{equation}
	P=s_n^{\epsilon_n}\cdots s_2^{\epsilon_2}s_1^{\epsilon_1},
	\qquad \epsilon_i=\pm 1.
	\label{eq:statistical-process-word}
\end{equation}
$P$ is called a \textit{statistical process}, and the corresponding statistical phase Eq.~\eqref{eq: statistical phase} is denoted by $\langle a|U(P)|a\rangle$. As a formal definition, a statistical process is a formal sequence Eq.~\eqref{eq:statistical-process-word} (\ie, an element in the free group $\operatorname{F}(S)$) such that $s,s^{-1}$ always appear in pairs, and the evaluation
\begin{equation}
	\begin{aligned}
		&\text{a realization determined by }\left\{|a\rangle,U(s)\right\}\\\mapsto\; &\langle a_0|U(P)|a_0\rangle\in U(1).
	\end{aligned}
\end{equation}
is trivial on the continuous subgroup for some $a_0\in A$ \footnote{This implies the property for all $a\in A$.}. Thus a statistical process distinguishes different statistics via a map $T^*(m)\to \RR/\ZZ$, and we define two statistical processes as equivalent if their corresponding maps are equal. Consequently, statistical processes are classified by
\begin{equation}
	T(m):=\hom(T^*(m),\RR/\ZZ).
\end{equation}

\begin{figure*}[t]
	\begin{center}
		\includegraphics[scale=0.9]{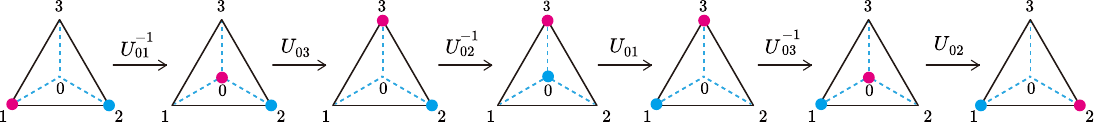}
	\end{center}
	\caption{The statistical phase of the T-junction process detects an anyon's self-statistics.}
	\label{figure: T-junction}
\end{figure*}

Statistical phases obtained from statistical processes are much more robust than is apparent from the definition. Let $P$ be a statistical process of the excitation complex $m_p(X,G)$. If $X$ is a combinatorial manifold (\cite{Xue2026StatisticsAbelian}, Theorem VI.11), then we have 
\begin{itemize}
	\item \textbf{The initial-state independence theorem}(\cite{Xue2026StatisticsAbelian}, Theorem VI.4)
	
	The statistical phase $\langle a_0|U(P)|a_0\rangle$ does not depend on the initial configuration $a_0\in A$. This usually implies that $U(P)$ is a pure phase, at least in the subspace spanned by configuration states.
	\item \textbf{The operator independence theorem}(\cite{Xue2026StatisticsAbelian}, Theorem VI.6)
	
	 As long as the Hilbert space $\cH$ and all configuration states $|a\rangle$ are fixed, the statistical phase $\langle a_0|U(P)|a_0\rangle$ does not depend on hopping operators $\{U(s)|s\in S\}$.
\end{itemize}

The simplest example of a statistical process is the T-junction process 
\cite{LW0316}
that detects the self-statistics of a $\ZZ_2$-anyon:
\begin{equation}
	P=s_{02}s_{03}^{-1}s_{01}s_{02}^{-1}s_{03}s_{01}^{-1},
\end{equation}
where $s_{ab}$ is the edge $ab$. The corresponding statistical phase is
\begin{equation}\label{eqTjunction}
	\left\langle \hbox{\raisebox{-2ex}{\includegraphics[width=1.2cm]{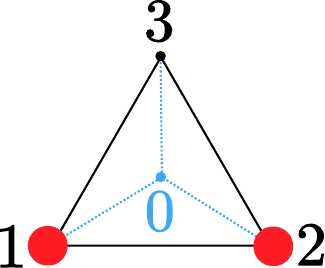}}}\right|
	U_{02}U_{03}^{-1}U_{01}U_{02}^{-1}U_{03}U_{01}^{-1}
	\left|\hbox{\raisebox{-2ex}{\includegraphics[width=1.2cm]{T_junction_state_12_detailed.pdf}}}\right\rangle.
\end{equation}
Geometrically, this T-junction process may be viewed as exchanging two anyons; see Fig.~\ref{figure: T-junction}. But in general, statistical processes may not have a very intuitive geometric explanation for higher-dimensional excitations. 

For a fixed complex $X$, the group $T^*(m_p(X,G))$ and its dual group $T(m_p(X,G))$ can be computed by a finite
algorithm.  When $X=\partial\Delta^{d+1}$, the boundary of a $(d+1)$-simplex
and the simplest triangulation of $S^d$, Ref.~\cite{Xue2026StatisticsAbelian}
proves that the statistics is classified by (see also
\Rf{LW0316,FH211014654,KC241201886})
\begin{equation}
	T^*(m_p(X,G))\simeq H^{d+2}(K(G,d-p),\RR/\ZZ).
\end{equation}
Computations suggest that the same formula holds for any generic combinatorial $d$-sphere. Although we do not have a proof, the object $K(G,d-p)$ indicates that $p$-chains are better viewed as $(d-p)$-cochains through Poincar\'e duality.  Setting $q=d-p$, we will use the cochain-version excitation complex $m^q(X,G)$ when $X$ is a $d$-dimensional combinatorial manifold.
\begin{itemize}
	\item $A=B^q(X,G)$;
	\item $S=G_0\times X_{q-1}\subset C^{q-1}(X,G)$;
	\item the boundary map $\partial:S\to A$ is replaced by the coboundary map $d:C^{q-1}(X,G)\to B^q(X,G)$;
	\item For some $s_i=(g_i,\alpha_i)$, we say 
	\begin{equation}
		\cap_{i=1}^k\supp(s_i)=\emptyset
	\end{equation}
	if and only if these simplices $\alpha_i$ are not simultaneously contained in a common $d$-simplex.\footnote{In other words, if
		$s=(g,\alpha)$, where $\alpha$ is a $(q-1)$-simplex, then $\operatorname{supp}(s)$
		is the closed barycentric dual cell of $\alpha$.}
\end{itemize}

This cochain version is more closely related to field theories and holographic constructions. This enables us to prove Theorem~\ref{thm:Realization}, which claims that:
\begin{equation}
	T^*(m^q(X,G))\supset H^{d+2}(K(G,q),\RR/\ZZ).
\end{equation}
The ``$\subset$'' direction, however, remains a conjecture.

In the next several subsections, we focus on fermions in two spatial
dimensions.  Thus we will study several realizations of $m^2(X,\ZZ_2)$, where
$X$ is a triangulation of $S^2$.  The configuration group $A=B^2(X,\ZZ_2)$
means that fermions live on $2$-simplices and that their total number is even.
The elementary set $S=X_1\subset C^1(X,\ZZ_2)$ means that fermions are created
by string operators at the ends of strings.  In practice, whenever possible, we
will define $U(s)$ for arbitrary cochains $s\in C^1(X,\ZZ_2)$; these
definitions are chosen to be consistent with the elementary operators $U(s)$
for $s\in X_1$, which already determine the statistics.

\subsection{Elementary fermions}
\label{subsec: complex fermions}

We now realize the excitation complex $m^2(X,\ZZ_2)$ in a Hilbert space of
fermions, where the space $|X|\simeq S^2$.  Assign one complex fermion operator
to each $2$-simplex $\sigma\in X_2$:
\begin{equation}
	c_\sigma,\;c_\sigma^\dagger,
	\qquad
	\{c_\sigma,c_{\sigma'}^\dagger\}=\delta_{\sigma\sigma'},
	\qquad
	\{c_\sigma,c_{\sigma'}\}=0 .
\end{equation}
Let $|0\rangle_{\rm F}$ be the Fock vacuum.  For
$a\in B^2(X,\ZZ_2)$, define the configuration state
\begin{equation}
	|a\rangle_{\rm F}
	=
	\prod_{\sigma\in X_2}^{<}(c_\sigma^\dagger)^{a_\sigma}|0\rangle_{\rm F},
\end{equation}
where $<$ is a fixed ordering of the $2$-simplices.  Since $X$ is closed and
$a$ is a coboundary, the total fermion parity of $|a\rangle_{\rm F}$ is even.
Thus all configuration states lie in the even-parity sector
$\cH^+_{\rm F}\subset\cH_{\rm F}$.

Define the Majorana operator
\begin{equation}
	\gamma_\sigma=c_\sigma+c_\sigma^\dagger,
	\qquad
	\{\gamma_\sigma,\gamma_{\sigma'}\}=2\delta_{\sigma\sigma'} .
\end{equation}
For $s\in C^1(X,\ZZ_2)$, set
\begin{equation}
	\label{eqMajoranaOperator}
	M(s)=\prod_{\sigma\in X_2}^{<}\gamma_\sigma^{(\dd s)(\sigma)} .
\end{equation}
This is a ``string operator'' that changes the fermion number by $\pm 1$ at its boundary.
Since $\dd s$ has even total parity, $M(s)$ preserves $\cH^+_{\rm F}$ and satisfies
\begin{equation}
	M(s)|a\rangle_{\rm F}
	=e^{\ii\theta_{\rm F}(s,a)}|a+\dd s\rangle_{\rm F}
\end{equation}
for a phase determined by the ordering convention.  Hence the complex-fermion
Hilbert space realizes $m^2(X,\ZZ_2)$.

These operators are multiplicative up to a phase
\begin{equation}
	M(s)M(t)\propto M(s+t).
\end{equation}
While this phase depends on the ordering of $2$-simplices, the commutator has a closed form. Two operators
$M(s)$ and $M(t)$ anticommute precisely when their endpoint configurations
overlap on an odd number of $2$-simplices. This can be expressed in terms of higher cup products:
\begin{equation}
	\label{eqMajoranaCommutator}
	M(s)M(t)
	=
	(-1)^{\int_X \dd s\hcup{2}\dd t}M(t)M(s),
\end{equation}
or, in our convenient form,
\begin{equation}
	\label{eqMajoranaCommutator2}
	[M(s),M(t)]
	=
	(-1)^{\int_X \dd s\hcup{2}\dd t}.
\end{equation}

In the T-junction process (Eq.~\eqref{eqTjunction}), we label the strings by $s_{0i}\in C^1(X,\ZZ_2)$ such that
\begin{equation}
	\dd s_{0i}=\delta_0+\delta_i,
\end{equation}
where $\delta_i\in C^2(X,\ZZ_2)$ is the Poincar\'e-dual endpoint cochain at the
vertex $i$. We evaluate the T-junction process $P$ as
\begin{equation}
	\label{eqTjunctionWordRewritten}
	M(P)
	=M_{02}M_{03}^{-1}M_{01}M_{02}^{-1}M_{03}M_{01}^{-1}.
\end{equation}

Because $[M_{0i},M_{0j}]$ are all pure phases, this simplifies to

\begin{equation}
	\begin{aligned}
		M(P)&= [M_{01},M_{02}] [M_{02},M_{03}][M_{03},M_{01}]\\&=(-1)^\tau.
	\end{aligned}
\end{equation}
where
\begin{equation}
	\tau=
	\int_X
	\dd s_{02}\hcup{2}\dd s_{03}
	+\dd s_{02}\hcup{2}\dd s_{01}
	+\dd s_{03}\hcup{2}\dd s_{01}.
\end{equation}

In the local T-junction geometry, each pair of endpoint configurations shares
only the common endpoint $0$.  Hence
\begin{align}
	\tau
	&=
	\int_X
	(\delta_0+\delta_2)\hcup{2}(\delta_0+\delta_3)
	+(\delta_0+\delta_2)\hcup{2}(\delta_0+\delta_1)
	\nonumber\\
	&\hspace{2.5cm}
	+(\delta_0+\delta_3)\hcup{2}(\delta_0+\delta_1)
	\nonumber\\
	&=1\mod 2.
\end{align}
In conclusion, we have
\begin{equation}
	M(P)=-1.
\end{equation}
Thus the fermionic statistics from the T-junction process are consistent with the traditional view from anti-commutative algebra.

\subsection{Emergent fermions in the toric code}
\label{subsec:toric_code_f_anyon}

Here, we would like to distinguish two types of excitations in a triangulated
spatial complex, whose degrees of freedom, such as qudits, live on vertices,
links, and higher-dimensional simplices.

The first type consists of \emph{elementary excitations}, namely the
excitations studied so far.  These are directly described by points, strings,
and higher-dimensional objects in space, formed by nonzero qudits, while zero
qudits are viewed as the vacuum.  Their spacetime trajectories are 1-chains,
2-chains, and higher chains.

The second type consists of \emph{bounding excitations}. These arise as
boundaries of \emph{condensed} elementary excitations. Their spacetime
trajectories are boundaries of one-higher-dimensional chains; equivalently,
their Poincar\'e-dual currents are exact cochains, namely coboundaries. 

In this section, we realize $m^2(X,\ZZ_2)$ in the bosonic toric code model,
where fermions appear as bounding excitations, namely emergent
fermions. The toric-code Hilbert space is formed by qubits
living on the edges of the spatial complex, which has basis
\begin{equation}
	\{|c\rangle\mid c\in C^1(X,\ZZ_2)\},
\end{equation}
and the vacuum is
\begin{equation}
	|\Omega\rangle=\confstate{0}=
	\sum_{\lambda\in C^0(X,\ZZ_2)}|\dd \lambda\rangle .
\end{equation}
The $m$ string translates the edge variable $c$, while the electric string
contributes a phase.  For the $f=e\times m$ anyon, we choose
the following representative of the string operator:
\begin{equation}
	\label{eq:U_f(s)}
	U(s)|c\rangle
	=
	(-1)^{\int_X c s+\dd c\hcup{1}s}|c+s\rangle,
	\qquad s\in C^1(X,\ZZ_2).
\end{equation}

\begin{remark}\label{remark: alternative solution}
	One may alternatively define
	\begin{equation}
		\label{eq:U_f'(s)}
		U'(s)|c\rangle
		=
		(-1)^{\int_X c s}|c+s\rangle.
	\end{equation}
	Both $U$ and $U'$ satisfy the two axioms and have fermionic statistics $U(P)=U'(P)=-1$. In fact, they only differ by a phase convention.
	 We choose Eq.~\eqref{eq:U_f(s)} because the commutator of $U(s)$ is the same as that of fermions, while that of $U'(s)$ is not.
\end{remark}

To get the configuration state $|a\rangle$, we choose some $s$ such that $\dd s=a$ and write
\begin{equation}
	\label{eq:configuration state}
	\confstate{\dd s}
	\sim
	U(s)|\Omega\rangle
	=
	\sum_{\lambda\in C^0(X,\ZZ_2)}
	(-1)^{\int_X(\dd\lambda)s}|\dd\lambda+s\rangle .
\end{equation}
It is easy to check the configuration axiom, so these $U(s)$ and $|a\rangle$ realize the excitation complex $m^2(X,\ZZ_2)$. Note that if $s$ is replaced by
$s+\dd\epsilon$, then $\dd s$ is unchanged but the state changes by an overall
phase,
\begin{equation}
	U(s+\dd\epsilon)|\Omega\rangle
	=
	(-1)^{\int_X(\dd\epsilon)s}U(s)|\Omega\rangle .
\end{equation}
Thus different choices of $s$ lead to an ambiguity in the phase convention for configuration states, a feature of nontrivial statistics.

From the formula of $U(s)$, a direct calculation gives
\begin{equation}
	\label{eqUnifiedProductLaw}
	U(s)U(t)
	=
	(-1)^{\int_X ts+\dd t\hcup{1}s}U(s+t).
\end{equation}
Therefore
\begin{align}
	[U(s),U(t)]
	&=
	(-1)^{\int_X ts+\dd t\hcup{1}s+st+\dd s\hcup{1}t}  
	\nonumber\\
	&=
	(-1)^{\int_X \dd s\hcup{2}\dd t}.
	\label{eqCochainCommutator}
\end{align}
The second equality follows from Eq.~\eqref{cupkrel}. This is exactly the fermionic commutator Eq.~\eqref{eqMajoranaCommutator2}; thus $U(P)=-1$.

Recall that the phase difference between $M(s)M(t)$ and $M(s+t)$ depends on the ordering of fermions, so the multiplication law is different from
Eq.~\eqref{eqUnifiedProductLaw} in general. Despite that, since they have the same commutator bicharacter, Lemma~III.1 of
Ref.~\cite{Xue2026StatisticsAbelian} implies that the difference can be
removed by a phase redefinition
\begin{equation}\label{eq: U_rm F(s)}
	U_{\rm F}(s)=\mu(s)M(s).
\end{equation}
After this redefinition, the complex-fermion operators may be chosen
to obey the same projective multiplication law as the toric-code representative
\eqref{eqUnifiedProductLaw}:
\begin{equation}
	\label{eq:Mprime_product}
	U_{\rm F}(s)U_{\rm F}(t)
	=
	(-1)^{\int_X ts+\dd t\hcup{1}s}U_{\rm F}(s+t).
\end{equation}
Thus hopping operators in the bosonic toric-code realization and the
complex-fermion realization carry the same algebra. The only difference is how
they are embedded in a total Hilbert space.

\subsection{Fermions and WZW term in one higher dimension}
\label{subsec:WZW_boundary_operator}

We now return to the WZW term and explain how the two boundary realizations above arise from the same one-higher-dimensional gauge theory.  Consider the $3+1$-dimensional $\ZZ_2$ one-form gauge theory twisted by
\begin{equation}
	\omega(f)=\frac12 f^2\in Z^4(\cdot,\RR/\ZZ),
	\qquad
	f\in Z^2(\cdot,\ZZ_2).
\end{equation}
All cochain formulas below are $\ZZ_2$-valued and appear in the path integral through the sign $(-1)^{\int(\cdots)}$.

Let $N^4$ be a spacetime whose boundary is the spatial three-manifold $M^3=\partial N^4$.  The path integral defines a state on $M$:
\begin{equation}
	Z(N)=
	\sum_{f\in Z^2(N,\ZZ_2)}
	(-1)^{\int_N f^2}|f_M\rangle .
\end{equation}
A gauge transformation with parameter $s\in C^1(M,\ZZ_2)$,
\begin{equation}
	f\longmapsto f+\dd s,
\end{equation}
changes the WZW action by a boundary term.  One convenient representative is
\begin{equation}
	(f+\dd s)^2-f^2=\dd \Theta(s,f),
	\qquad
	\Theta(s,f)=s \dd s+f\hcup{1}\dd s .
	\label{eq:Theta0}
\end{equation}

\begin{remark}
	An alternative solution is
	\begin{equation}
		\Theta'(s,f)=sf+fs+s \dd s.
	\end{equation}
	Indeed, the two solutions differ by an exact term:
	\begin{equation}
		\Theta'(s,f)-\Theta(s,f)=\dd (f\hcup{1}s).
		\label{eq:Theta_difference}
	\end{equation}
	Thus they describe the same bulk gauge transformation on a closed spatial manifold, but they induce slightly different boundary operators when $\partial M\neq\emptyset$ (Eq.~\eqref{eq:U_f(s)} versus \eqref{eq:U_f'(s)}).

\end{remark}

Thus the gauge transformation acts on the spatial state by
\begin{equation}
	T_{s}|f\rangle
	=
	(-1)^{\int_M\Theta(s,f)}|f+\dd s\rangle .
\end{equation}
The state $Z(N)$ is invariant under this transformation.

Now restrict to the coboundary sector,
\begin{equation}
	f=\dd v,
	\qquad
	v\in C^1(M,\ZZ_2),
\end{equation}
where $v$ is a choice of gauge potential trivializing $f$.  We write
\begin{equation}
	\beta(v)=v \dd v,
	\qquad
	\dd \beta(v)=(\dd v)^2.
\end{equation}
If $M$ is closed, the WZW state can be written as
\begin{equation}
	|\psi_M\rangle
	=
	\sum_{v\in C^1(M,\ZZ_2)}
	(-1)^{\int_M v\dd v}|\dd v\rangle.
	\label{eq:closed_M_WZW_state}
\end{equation}
It satisfies
\begin{equation}
	T_{s}|\psi_M\rangle=|\psi_M\rangle.
\end{equation}

The interesting point appears when $M$ itself has a boundary
\begin{equation}
	X=\partial M .
\end{equation}
Then $v$ and $v+\dd\epsilon$ represent the same bulk coboundary $\dd v$, but their WZW amplitudes differ by
\begin{equation}
	(-1)^{\int_M (v+\dd\epsilon)\dd(v+\dd\epsilon)-v\dd v}
	=
	(-1)^{\int_X\epsilon \dd v}.
\end{equation}

If one still uses Eq.~\eqref{eq:closed_M_WZW_state} as the ground state, interference will cancel all components with $\dd v|_X\ne 0$. To avoid this cancellation, the WZW term requires a boundary Hilbert space carrying a projective action of the boundary gauge transformations. From now on, to distinguish bulk and boundary cochains, we use the notation $\widetilde{v},\widetilde{s}$ to denote those in the bulk, and their boundary restrictions are denoted by $v$ and $s$.

\subsubsection{Bosonic boundary Hilbert space and the toric-code $f$ anyon}

Our first realization keeps the boundary cochain $v$ as a bosonic edge-qubit degree of freedom.  Define
\begin{equation}
	|\Psi_M^{\rm TC}\rangle
	=
	\sum_{\widetilde{v}\in C^1(M,\ZZ_2)}
	(-1)^{\int_M \widetilde{v}\dd \widetilde{v}}|\dd \widetilde{v}\rangle\otimes |v\rangle .
	\label{eq:TC_bulk_boundary_state}
\end{equation}
The identity associated with $\Theta$ is
\begin{equation}
	\beta(v+s)-\beta(v)-\Theta(s,\dd v)=\dd L(s,v),
	\label{eq:L0_descent}
\end{equation}
and a solution is
\begin{equation}
	L(s,v)=vs+\dd v\hcup{1}s.
\end{equation}
Therefore the bulk gauge transformation $T_{\widetilde{s}}$ leaves a boundary factor $(-1)^{\int_X vs}$, cancelled by the boundary operator:
\begin{equation}
	U(s)|c\rangle
	=
	(-1)^{\int_X c s+\dd c\hcup{1}s}|c+s\rangle .
	\label{eq:U0_boundary_operator}
\end{equation}
Equivalently,
\begin{equation}
	\bigl(T_{\widetilde{s}}\otimes U(s)\bigr)|\Psi_M^{\rm TC}\rangle
	=|\Psi_M^{\rm TC}\rangle .
\end{equation}
This is exactly the toric-code $f$-anyon string operator in Eq.~\eqref{eq:U_f(s)}.  Thus the emergent $f$ anyon appears as a bosonic boundary realization of the projective gauge-transformation law imposed by the bulk WZW term. 

The configuration states are constructed in a similar way.  Let $a\in B^2(X,\ZZ_2)$ be the boundary excitation configuration, and let $\widetilde a\in B^2(M,\ZZ_2)$ be a bulk extension with $\widetilde a|_X=a$. Define
\begin{equation}\label{eq:configuration state2}
	\confstate{a}_{\rm TC}
	=
	\sum_{\widetilde{v}\in C^1(M,\ZZ_2),\; \dd \widetilde{v}=\widetilde a}
	(-1)^{\int_M \widetilde{v}\dd \widetilde{v}}|v\rangle .
\end{equation}
Choose $\widetilde s\in C^1(M,\ZZ_2)$ with $\dd\widetilde s=\widetilde a$.  Then every term in the sum can be written as $\widetilde{v}=\widetilde s+\dd\widetilde\lambda$ for some $\widetilde\lambda\in C^0(M,\ZZ_2)$.  With $\lambda=\widetilde\lambda|_X$, one obtains
\begin{equation}
	\begin{aligned}
		\confstate{a}_{\rm TC}
		&=\sum_{\widetilde\lambda\in C^0(M,\ZZ_2)}
		(-1)^{\int_M(\widetilde s+\dd\widetilde\lambda)\widetilde a}|s+\dd\lambda\rangle \\
		&=(-1)^{\int_M\widetilde s\widetilde a}
		\sum_{\widetilde\lambda\in C^0(M,\ZZ_2)}
		(-1)^{\int_X\lambda a}|s+\dd\lambda\rangle \\
		&=(-1)^{\int_M\widetilde s\widetilde a}
		\sum_{\lambda\in C^0(X,\ZZ_2)}
		(-1)^{\int_X(\dd\lambda)s}|s+\dd\lambda\rangle .
	\end{aligned}
\end{equation}
Up to an overall phase, this is the state in Eq.~\eqref{eq:configuration state}.

\subsubsection{Fermionic boundary Hilbert space and complex fermions}

The same bulk WZW term also admits a boundary realization by complex fermions. The corresponding Hilbert space is constructed in Section~\ref{subsec: complex fermions}. On the operator algebra level, we replace the bosonic string operator $U(s)$ by its fermionic counterpart $U_{\rm F}(s)$ defined in Eq.~\eqref{eq: U_rm F(s)}. On the state level, we define the bulk-boundary state as
\begin{equation}
	|\Psi_M^{\rm F}\rangle
	=
	\sum_{\widetilde{v}\in C^1(M,\ZZ_2)}
	(-1)^{\int_M \widetilde{v}\dd \widetilde{v}}|\dd\widetilde{v}\rangle\otimes U_{\rm F}(v)|0\rangle_{\rm F} .
	\label{eq:F_bulk_boundary_state}
\end{equation}
Thus 
\begin{equation}
	\bigl(T_{\widetilde{s}}\otimes U_{\rm F}(s)\bigr)|\Psi_M^{\rm F}\rangle
	=|\Psi_M^{\rm F}\rangle.
\end{equation}

We therefore arrive at the following holographic picture.  The one-higher-dimensional WZW term fixes a projective boundary action of gauge transformations.  This projective action can be realized in at least two ways.  Keeping the boundary cochain $v_X$ as an edge-qubit variable gives the toric-code $f$ anyon, an emergent fermion in a bosonic Hilbert space.  Realizing the same boundary transformation law in the even-parity Hilbert space of complex fermions gives a fermion. 
In short, both bosonic and fermionic boundary realizations carry fermionic statistics as boundary manifestations of the same WZW term in one higher dimension.

\section{Statistics of Abelian excitations from one-higher-dimensional WZW terms}
\label{sec:higher-form-realization}

In this section we explain how a one-higher-dimensional WZW term induces a hopping operator algebra with nontrivial statistics. For notational simplicity,
we restrict to the homogeneous higher-form case, so there is only one cochain
degree. This construction
applies to all Abelian excitations whose fusion rules form an Abelian group and do
not mix different dimensions. The argument may potentially apply to higher-group cases, but the
axiomatic framework of statistics has not yet been fully generalized to
non-Abelian cases.

\subsection{Bulk higher-form gauge theory and descendants}\label{subsec: decendant functions}

The target hopping operator will be defined on a $d$-dimensional space $X$ with the form of Eq.~\eqref{eq:boundary-excitation-operator}, while our input is a $(d+2)$-cocycle
\begin{equation}
	\omega: Z^q(\cdot,G)\to Z^{d+2}(\cdot,\RZ).
\end{equation}
The main task of this subsection is to construct the function $L(s,v)\in C^d(X,\RZ)$ that decrease the dimension by $2$.

On a closed $(d+2)$-dimensional spacetime $N$, the associated twisted
higher-form gauge theory has field
\begin{equation}
	f\in Z^q(N,G)
\end{equation}
and topological weight
\begin{equation}
	\exp\left(2\pi\ii\int_N\omega(f)\right).
\end{equation}
Here $\omega(f)$ means that the universal cocycle $\omega$ is evaluated on the
$q$-form gauge field $f$.

A gauge transformation is parametrized by a cochain
\begin{equation}
	s\in C^{q-1}(N,G),
	\qquad f\longmapsto f+\dd s .
\end{equation}
The action is gauge invariant on a closed spacetime.  On a spacetime with
boundary, however, the variation becomes a boundary term.  We choose a local
cochain operation $\Theta$ satisfying
\begin{equation}
		\dd \Theta(s,f)=\omega(f+\dd s)-\omega(f) .
	\label{eq:Theta-descent}
\end{equation}

Thus $\Theta$ is the first descendant of the bulk WZW cocycle.  Physically,
$\Theta$ is the phase picked up by a spatial wavefunction under a bulk gauge
transformation. We also require
\begin{equation}
		\Theta(0,f)=0.
	\label{eq:Theta-descent1}
\end{equation}
This can always be done by the replacement $\Theta(s,f)\mapsto\Theta(s,f)-\Theta(0,f)$.

Next restrict to the exact sector
\begin{equation}
	f=\dd v,
	\qquad v\in C^{q-1}(\cdot,G).
\end{equation}
In this sector the WZW action itself has a spatial descendant
\begin{equation}
		\beta(v):=\Theta(v,0),
		\qquad \dd \beta(v)=\omega(\dd v) .
	\label{eq:beta-descent}
\end{equation}
Thus $\int_M\beta(v)$ is the WZW phase on a spatial filling $M$ when the bulk
$q$-form field is written as $\dd v$.

Finally, there is a second descendant $L$.  We use the convention
\begin{equation}
		\dd L(s,v)=
		\Theta(s,\dd v)+\beta(v)-\beta(s+v) .
	\label{eq:L-descent}
\end{equation}
and
\begin{equation}
		L(0,v)=0.
	\label{eq:L-descent1}
\end{equation}

The right-hand side is closed because the two terms
$\omega(\dd v+\dd s)-\omega(\dd v)$ cancel.  Therefore a local primitive $L$ exists for
standard local cocycle representatives, obtained by the usual prism or
cochain-homotopy construction.  In the following, $\Theta$, $\beta$, and $L$
are chosen once and for all as normalized local cochain operations.

The physical meaning of $L$ is the following.  Let $M$ be a spatial filling with
boundary $X$, and let $\widetilde s\in C^{q-1}(M,G)$ extend
$s\in C^{q-1}(X,G)$.  Integrating Eq.~\eqref{eq:L-descent} over $M$ gives
\begin{align}
	&\int_M \beta(\widetilde{v})
	+\int_M \Theta(\widetilde s,\dd \widetilde{v})
	-\int_M \beta(\widetilde{v}+\widetilde s)
	\nonumber\\
	&\hspace{3.0cm}=\int_X L(s,v) .
	\label{eq:L-physical-identity}
\end{align}
Thus a bulk gauge transformation changes the WZW amplitude as if the potential
were shifted from $\widetilde{v}$ to $\widetilde{v}+\widetilde s$, but leaves a boundary phase.  The
boundary hopping operator is precisely the operator that cancels this
leftover phase.

\subsection{Boundary configuration states and hopping operators}

Consider a spacetime $N^{d+2}$ with a spatial boundary $\partial N^{d+2}=M^{d+1}$. In the most fundamental form, the WZW wavefunction is defined by
\begin{equation}
	\sum_{f\in Z^q(N,G)}e^{\int_N2\pi \ii \omega(f)}|f_{M}\rangle.
\end{equation}

If we only consider the exact sector $f=\dd \widetilde{v}$, then we can get rid of the spacetime $N$, but only do integration on the space $M$:
\begin{equation}
	|\Psi_M^\omega\rangle=\sum_{v\in C^{q-1}(M,G)}e^{2\pi \ii\int_M\beta(\widetilde{v})}|\dd \widetilde{v}\rangle.
\end{equation}

Due to Eq.~\eqref{eq:L-descent}, many $v$ contribute to the same $|\dd v\rangle$ without interference. Equivalently, $|\Psi_M^\omega\rangle$ is invariant under the gauge transformation
\begin{equation}
	T^M_{\widetilde s}|\dd\widetilde{v}\rangle_{\rm bulk}
	={}e^{2\pi\ii\int_M\Theta(\widetilde s,\dd\widetilde{v})}     
	|\dd(\widetilde{v}+\widetilde s)\rangle_{\rm bulk} .
\end{equation}

 However, when $M$ has a boundary, interference may happen and $|\Psi_M^\omega\rangle$ is no longer gauge-transformation invariant. Thus we construct a specific boundary condition. We assume $M$ to be a combinatorial $(d+1)$-disk, so $X=\partial M$ is a combinatorial $d$-sphere. We assign the boundary Hilbert space
\begin{equation}
	\mathcal H_X=\mathbb C\bigl[C^{q-1}(X,G)\bigr],
\end{equation}
with basis vectors $|v\rangle$, $v\in C^{q-1}(X,G)$.  We define the boundary-bulk state as
\begin{equation}
	|\Psi_M^\omega\rangle
	={}\sum_{\widetilde{v}\in C^{q-1}(M,G)}
	e^{2\pi\ii\int_M\beta(\widetilde{v})}
	|\dd\widetilde{v}\rangle_{\rm bulk}\otimes |v\rangle .
	\label{eq:bulk-boundary-WZW-state}
\end{equation}
We define the boundary hopping operator by
\begin{equation}
	U(s)|v\rangle
	=e^{-2\pi\ii\int_X L(s,v)}|v+s\rangle .
	\label{eq:boundary-excitation-operator}
\end{equation}
This specific construction implies
\begin{equation}
	\bigl(T^M_{\widetilde s}\otimes U(s)\bigr)
	|\Psi_M^\omega\rangle
	=|\Psi_M^\omega\rangle .
	\label{eq:bulk-boundary-gauge-invariance}
\end{equation}
This is the central holographic picture: a boundary hopping operator is the
boundary remnant of a bulk higher-form gauge transformation; together they form
a gauge-invariant state $|\Psi^\omega_M\rangle$.

Similarly, a \textit{configuration} corresponds to the boundary restriction of bulk gauge potential.

For any
\begin{equation}
	a\in B^q(X,G),
\end{equation}
choose a bulk coboundary extension
\begin{equation}
	\widetilde a\in B^q(M,G),
	\qquad \widetilde a|_X=a .
\end{equation}
Analogous to Eq.~\eqref{eq:bulk-boundary-WZW-state}, we define the configuration state by
\begin{equation}
		\confstate{a}_{\widetilde a}
		=\sum_{\substack{\widetilde{v}\in C^{q-1}(M,G)\\ \dd\widetilde{v}=\widetilde a}}
		e^{2\pi\ii\int_M\beta(\widetilde{v})}|v\rangle .
	\label{eq:configuration-state-beta}
\end{equation}
 Changing
$\widetilde a$ changes only the overall phase convention of $\confstate{a}$.

\subsection{Main Theorems}
\label{maintheorem}

Here we present the main theorem connecting statistics and a WZW term in one-higher dimension. We state the theorem only for
$q$-dimensional excitations, while it automatically generalizes to all Abelian
cases. We only sketch the proof here; the full proof is given in
Appendix~\ref{app:wzw-realization-details}.

\begin{theorem}\label{thm:Realization}
	(\textbf{Standard realization of statistics}) Let $X$ be a combinatorial $d$-sphere. Pick any $M$ as a combinatorial $(d+1)$-disk with $\partial M=X$. Given a cocycle $\omega\in
Z^{d+2}(K(G,q),\mathbb R/\mathbb Z)$ and descendant functions $\Theta,\beta$,
and $L$, the configuration states Eq. \eqref{eq:configuration-state-beta} and
the hopping operators Eq. \eqref{eq:boundary-excitation-operator} \textit{realize} the
excitation complex $m^q(X,G)$. For different choices of $M,\omega,\Theta,\beta,L$, the resulting
statistics $\in T^*(m^q(X,G))$ are the same if and only if they have the same cohomology class
$[\omega]\in H^{d+2}(K(G,q),\RR/\ZZ)$. Thus we have a canonical
embedding\footnote{We believe (Conjecture IV.1 of
Ref.~\cite{Xue2026StatisticsAbelian}) this map is an isomorphism, but we do not
have a proof yet.} \begin{equation}
		H^{d+2}\bigl(K(G,q),\mathbb R/\mathbb Z\bigr)
		\hookrightarrow T^*\bigl(m^q(X,G)\bigr) .
		\label{eq:WZW-subgroup-statistics}
	\end{equation}
	Dually, there is an canonical surjection:
	\begin{equation}\label{eq: map statistical process to homology}
		T\bigl(m^q(X,G)\bigr)\twoheadrightarrow  H_{d+2}\bigl(K(G,q),\mathbb Z\bigr).
	\end{equation}
	They provide a partial classification of statistics and statistical processes.
\end{theorem}

\begin{remark}\label{remark: sphere is important}
	$X$ being a combinatorial manifold is a sufficient condition that the locality axiom works correctly. It is crucial that $|X|\simeq S^d$ be a sphere rather than a generic closed manifold. For a counterexample, suppose there is a bulk $M$ satisfying $X=\partial M$, $H^{q-1}(M,G)=0$, and $H^{q-1}(X,G)\ne 0$. Taking $\widetilde{a}=0$ in Eq.~\eqref{eq:configuration-state-beta}, we have
	\begin{equation}
		\confstate{a}_0=\sum_{\widetilde{v}\in Z^{q-1}(M,G)}e^{2\pi i \int_M\beta(\widetilde{v})}|v\rangle.
	\end{equation}
	Because $\widetilde{v}\in B^{q-1}(M,G)$, the boundary restriction $v\in B^{q-1}(X,G)$ is also in the coboundary sector. However, if we apply $U(s)$ for a nontrivial $s\in Z^{q-1}(X,G)$, the resulting cocycle $s+v$ is also nontrivial. Thus $U(s)\confstate{a}_0$ is not proportional to $\confstate{a}_0$, a contradiction.
	
	The mechanism of this phenomenon is similar to the topological ground-state degeneracy of the toric code. One should not view this counterexample as an instance where statistics imply higher-form symmetry breaking, since this counterexample also works when $\omega=0$. It only shows the failure of this particular realization.
\end{remark}

\noindent\textit{Sketch of proof.}
The configuration axiom follows by extending $s$ to
$\widetilde s$ on $M$ and using Eq.~\eqref{eq:L-descent}: the operator $U(s)$
changes the summation variable from $\widetilde{v}$ to $\widetilde{v}+\widetilde s$ and hence sends
$\confstate{a}$ to $\confstate{a+\dd s}$, up to the phase dictated by $\Theta$.  Locality
follows because $U(s)$ is a translation of $c$ followed by a phase determined by
the local cochain operation $L(s,c)$.  A nested commutator is a finite
difference of this local phase, supported only where all the dual supports of
the involved simplex labels meet, and the condition $L(0,v)=0$ is used.  If that common intersection is empty, the
nested commutator is the identity. It can be shown that when $[\omega]$ is fixed, different choices of data only lead to local or unimportant changes, so the statistics remain the same.

To prove that different cohomology classes $[\omega]\in H^{d+2}(K(G,q),\RR/\ZZ)$
correspond to different statistics in $T^*(m^q(X,G))$, it is enough to prove
that $[\omega]\ne 0$ implies the statistics is nontrivial. To do so, we
construct a statistical process $P$ of $m^q(X,G)$ and prove the
corresponding statistical phase $\Phi_\omega(P)\in U(1)$ is nontrivial. When
$X\simeq \partial \Delta^{d+1}$, the existence of such a statistical process
$P_0$ has already been proved in Ref.~\cite{Xue2026StatisticsAbelian}. The critical
argument is that $P_0$ can be transferred to any generic combinatorial
$d$-sphere $X$, producing a statistical process $P_\sigma$ of $m^q(X,G)$, local at a specific $d$-simplex $\sigma$. It
can be shown that
\begin{equation}
	\Phi_\omega(P_\sigma)=\Phi_\omega(P_0)\ne1,
\end{equation}
so the realization given by $\omega$ indeed has nontrivial statistics.

The existence of $P_\sigma$ is also physically meaningful: it says that any statistics in the classification $H^{d+2}(K(G,q),\RZ)$ can be detected by a local statistical process near $\sigma$.

\begin{theorem}
	(\textbf{Local detectability of statistics}) Let $X$ be a combinatorial $d$-sphere and $\sigma$ is one of its $d$-simplex. Any element in $H_{d+2}(K(G,q),\ZZ)$ is the image of some statistical process
	\begin{equation}
		P_\sigma=s_n^{\pm}\cdots s_1^{\pm}
	\end{equation}
	under the map Eq.~\eqref{eq: map statistical process to homology}, such that $\alpha_i\subset \sigma$ holds for any $s_i=(g_i,\alpha_i)\in G_0\times X_{q-1}$.
\end{theorem} 

\subsection{Statistics versus symmetry anomaly}\label{anosta}

In our realization of statistics, the central construction is the formula of hopping operator
\begin{equation*}
	U(s)|v\rangle
	=e^{-2\pi\ii\int_X L(s,v)}|v+s\rangle.
\end{equation*}

$v$ and $s$ in this formula are both in $C^{q-1}(X,G)$ but have different positions. We perform a strange trick: we exchange their positions and then change the notation from $s$ to $\gamma$ for $\gamma\in Z^{q-1}(X,G)$, defining a very similar family of operators

\begin{equation}
	\mathcal{S}(\gamma)|v\rangle
	=e^{-2\pi\ii\int_X L(v,\gamma)}|v+\gamma\rangle.
\end{equation}
This seems to be physical meaningless, but what we actually do is to construct global symmetry transformations. In addition to Eq.~\eqref{eq:L-descent} and \eqref{eq:L-descent1}, we impose another normalization condition so that $\mathcal{S}(\gamma)=1$:
\begin{equation}
	L(v,0)=0,\quad\forall v.
\end{equation}
This can always be done by replacing $L(v,s)$ by $L(v,s)-L(v,0)$. Then, we have the following theorem, proved in Appendix~\ref{appendix: proof of anomalous symmetry}.

\begin{theorem}\label{thm:Sgamma}
	(\textbf{Hopping operators commute with anomalous symmetry}) Let $\gamma,\gamma'\in Z^{q-1}(X,G)$. Then
	\begin{itemize}
		\item $\mathcal S(\gamma)\mathcal S(\gamma')\propto
		\mathcal S(\gamma+\gamma')$.
		\item $\mathcal S(\gamma)$ commutes with
		$\mathcal S(\gamma')$ if either $\gamma$ or $\gamma'$ is a coboundary.
		\item $U(s)$ commutes with $\mathcal S(\gamma)$ for every
		$s\in C^{q-1}(X,G)$.
		\item If $X$ is a combinatorial sphere, then there exists
		$\phi(\gamma)\in\RZ$ such that
		\begin{equation}
			\mathcal S(\gamma)\confstate{a}
			=e^{2\pi\ii\phi(\gamma)}\confstate{a},
			\qquad \forall a\in B^q(X,G).
		\end{equation}
		Hence all configuration states lie in the same $\mathcal S(\gamma)$
		eigenspace.
	\end{itemize}
\end{theorem}

\begin{remark}
	
	These $\mathcal{S}(\gamma)$ are the the anomalous symmetry associated to $[\omega]\in H^{d+2}(K(G,q),\RZ)$, and hopping operators $U(s)$ are local symmetric operators. One may feel that they are almost the same; however, their properties seem to be completely different; for example, $\{\mathcal{S}(s)\}$ are transformation patch operators \cite{CW220303596} but they cannot be treated as hopping operators, since they may not satisfy the configuration axiom.
	
	One concrete example can show the different nature between two families of operators. Consider $\omega(f_1)=\frac{1}{2}f_1^3$ that represents the nontrivial element of $H^3(B\ZZ_2,\RZ)\simeq \ZZ_2$. By solving these equations, we get
	\begin{equation}
		L(s,v)=\frac{1}{2}v(\dd v s+s \dd v+s\dd s),
	\end{equation}
	where $v$ and $s$ are both $0$-cochains. When the space $X$ is a circle, we use $\mathbf{1}$ to denote the constant field: $\mathbf{1}:11\cdots 1\in Z^1(X,\ZZ_2)$. Then we have
	\begin{equation}
		\mathcal{S}(\mathbf{1})|v\rangle=(-1)^{\int_X v\dd v}|v+\mathbf{1}\rangle.
	\end{equation}
	This is exactly the anomalous $\ZZ_2$ symmetry in one dimension. In contrast, we have
	\begin{equation}
		U(\mathbf{1})|v\rangle=|v+\mathbf{1}\rangle,
	\end{equation}
	which is \textit{not} the corresponding anomalous symmetry. In general, the operator independence theorem \cite{Xue2026StatisticsAbelian} allows any choice of hopping operators as long as they satisfy the configuration axiom, so $U(s)$ may not be like symmetry transformation inside $\supp(s)$.
	 Thus although $U(s)$ and $\mathcal{S}(\gamma)$ have very similar formula, their roles are completely different: $\mathcal{S}(\gamma)$ are symmetry transformations, while $U(s)$ are hopping operators in LOsA encoding the conservation law.
	
\end{remark}

Previously, we realize statistics at the boundary of a higher-form gauge theory twisted by cocycle $\omega\in Z^{d+2}(K(G,q),\RZ)$. In parallel, we realize this anomalous symmetry at the boundary symmetry transformation of a symmetric protected topological (SPT) phase. 

Let $M^{d+1}$ be a closed spatial manifold. In gauge theory, the WZW state in the exact sector is given by
\begin{equation}
	\sum_{\widetilde{v}\in C^{q-1}(M,G)}e^{2\pi\ii \int_M\beta(\widetilde{v})}|\dd\widetilde{v}\rangle.
\end{equation}
Here $\widetilde{v}$ represents a gauge transformation from the trivial gauge potential $0$, and $d\widetilde{v}$ is the gauge potential after the gauge transformation. The procedure to construct the SPT phase corresponding to a gauge theory, called ungauging, is simply to replace the gauge potential $d\widetilde{v}$ by the gauge transformation $\widetilde{v}$:

\begin{equation}
	|\psi_{\mathrm{SPT}}\rangle=\sum_{\widetilde{v}\in C^{q-1}(M,G)}e^{2\pi\ii \int_M\beta(\widetilde{v})}|\widetilde{v}\rangle.
\end{equation}
 We define a global symmetry by
\begin{equation}
	\mathcal{S}_{\mathrm{bulk}}|\widetilde{v}\rangle=|\widetilde{v}+\widetilde{\gamma}\rangle.
\end{equation}
For closed $M$, $|\psi_{\mathrm{SPT}}\rangle$ is indeed symmetric:
\begin{equation}
	\begin{aligned}
		&\mathcal{S}_{\mathrm{bulk}}(\widetilde{\gamma})|\psi_{\mathrm{SPT}}\rangle\\
		&=\sum_{\widetilde{v}\in C^{q-1}(M,G)}e^{2\pi\ii \int_M\beta(\widetilde{v})}|\widetilde{v}+\widetilde{\gamma}\rangle\\
		&=\sum_{\widetilde{v}\in C^{q-1}(M,G)}e^{2\pi\ii \int_M\beta(\widetilde{v}+\widetilde{\gamma})-\beta(\widetilde{\gamma})}|\widetilde{v}+\widetilde{\gamma}\rangle\\
		&=e^{-2\pi \ii \int_M\beta(\widetilde{\gamma})}|\psi_{\mathrm{SPT}}\rangle.
	\end{aligned}
\end{equation}

However, if we want to make $|\psi_{\mathrm{SPT}}\rangle$ symmetric when $\partial M\ne \emptyset$, we should to replace the boundary-bulk symmetry by
\begin{equation}
	\mathcal{S}_{\mathrm{boundary-bulk}}(\widetilde{\gamma})|\widetilde{v}\rangle=e^{-2\pi\ii\int_{\partial M}L(v,\gamma)}|\widetilde{v}+\widetilde{\gamma}\rangle.
\end{equation}

Then it can be verified that
\begin{equation}
	\mathcal{S}_{\mathrm{boundary-bulk}}(\widetilde{\gamma})|\psi_{\mathrm{SPT}}\rangle=e^{-2\pi \ii \int_M\beta(\widetilde{\gamma})}|\psi_{\mathrm{SPT}}\rangle.
\end{equation}
That is why we say $\mathcal{S}(\gamma)$ is an anomalous symmetry.
In other words, a symmetry anomaly manifests itself as non-trivial
statistics of symmetry defects \cite{LG1209,BZ14104540,W181202517}. We make Table~\ref{table: comparison of two perspectives} to compare the two perspectives, though we are not very sure about it.

\begin{table*}[htbp]
	\centering
	\renewcommand{\arraystretch}{1.35}
	\begin{tabular}{c|c|c}
		\hline
		\textbf{mathematics level}
		&
		\textbf{higher group \(\mathcal{G}\)}
		&
		\textbf{excitation complex \(m\)}
		\\
		\hline
		\textbf{kinematics level}
		&
		 symmetry defects of $\mathcal{G}$ global symmetry
		&
		hopping operator algebra (realization of \(m\))
		\\
		\hline
		\textbf{bulk counterpart}
		&
		symmetry of SPT phase
		&
		\shortstack{gauge transformation of\\Dijkgraaf-Witten theory}
		\\
		\hline
		\textbf{invariant objects}
		&
		\shortstack{algebra of\\transformation patch operator}
		&
		statistics of hopping operator
		\\
		\hline
		\textbf{allowed perturbation }
		&
		\shortstack{local symmetric operator\\ on the boundary of the patch}
		 
		&
		\shortstack{hopping operators\\respecting the configuration axiom}
		
		\\
		\hline
		\textbf{obstruction to}
		&
		onsiteness
		&
		condensation
		\\
		\hline
		\textbf{detection}
		&
		anomaly indicator
		&
		statistical process
		\\
		\hline
		\textbf{classification}
		&
		\multicolumn{2}{|c}{
			 \(H^{d+2}(K(G,q),\mathbb{R}/\mathbb{Z})\) for higher-form case
		}
		\\
		\hline
	\end{tabular}
	\caption{Two perspective of symmetry anomaly}
	\label{table: comparison of two perspectives}
\end{table*}

\subsection{Constructing solutions from prism integration}\label{subsec: prism integral}

Because the map $\Theta(s,f)\in C^{d+1}(\cdot,\RZ)$ is natural, it is sufficient to construct the value of $\Theta(s,f)\in C^{d+1}(\Delta^{d+1},\RZ)\simeq \RZ$ for any $s\in C^{q-1}(\Delta^{d+1},\ZZ_2)$ and $f\in Z^{q}(\Delta^{d+1},\ZZ_2)$. Now, we use the pull back to define $\widetilde{f}\in Z^{q}(\Delta^{d+1}\times I,\ZZ_2)$ and $\widetilde{s}\in C^{q-1}(\Delta^{d+1}\times I,\ZZ_2)$. We also define $\eta\in C^0(I,\ZZ)$ such that $\eta(0)=0$ and $\eta(1)=1$. Then we define
\begin{equation}
	h=\widetilde{f}+\dd (\eta\smile \widetilde{s})\in Z^{d+2}(\Delta^{d+1}\times I,\ZZ_2).
\end{equation}

At the lower and the upper surfaces of the prism, the restrictions of $h$ are $h_0=f$ and $h_1=f+\dd s$. Mathematically, $h$ is a simplicial homotopy between $f$ and $f+\dd s$. Next, we define $\Theta$ by
\begin{equation}
	\int_{\Delta^{d+1}}\Theta(s,f)=(-1)^{d+1}\int_{\Delta^{d+1}\times I}\omega(h).
\end{equation}
We then have the Leibniz's rule
\begin{equation}
	\begin{aligned}
		\int_{\Delta^{d+2}}d\Theta(s,f)&=(-1)^{d+1}\int_{\partial\Delta^{d+2}\times I}\omega(h)\\
		&=\int_{\Delta^{d+2}\times \partial I}\omega(h)\\
		&=\int_{\Delta^{d+2}}\omega(h_1)-\omega(h_0).
	\end{aligned}
\end{equation}
This is exactly what we want. Also, it satisfies $\Theta(0,f)=0$ automatically. 

To further construct $L(s,v)$, we will do the prism integration on $\Delta^d\times \Delta^2$. We define $\eta_1,\eta_2\in C^0(\Delta^2,\ZZ)$ such that
\begin{equation}
	\eta_1(i)=\left\{\begin{aligned}
		&0,\quad i<1;\\
		&1,\quad i\ge 1.
	\end{aligned}
	\right.
\end{equation}
and
\begin{equation}
	\eta_2(i)=\left\{\begin{aligned}
		&0,\quad i<2;\\
		&1,\quad i\ge 2.
	\end{aligned}
	\right.
\end{equation}
We define
\begin{equation}
	h=\dd(\eta_1\smile v+\eta_2\smile s)\in Z^{q}(\Delta^d\times \Delta^2,G).
\end{equation}

Intuitively, $h$ is a $2$-homotopy between the composition of gauge transformations $0\mapsto dv, dv\mapsto dv+ds$ and the whole gauge transformation $0\mapsto \dd(v+s)$. By the Leibniz's rule, we have
\begin{equation}
	\begin{aligned}
		&\int_{\partial\Delta^{d+1}\times \Delta^2}\omega(h)\\&=(-1)^d\int_{\Delta^{d+1}\times \partial\Delta^2}\omega(h)\\
		&=(-1)^d\int_{\Delta^{d+1}\times \Delta^1}\omega(h_{12})+\omega(h_{01})-\omega(h_{02})\\
		&=-\int_{\Delta^{d+1}} \Theta(s,dv)+\Theta(v,0)-\Theta(v+s,0).
	\end{aligned}
\end{equation}
Note that by the construction of $h$, three restrictions $h_{12}$, $h_{01}$, and $h_{02}$ are exactly $\dd(v+\eta s)$, $d(\eta v)$, and $d(\eta v+\eta s)$, respectively.

Still, we take emergent fermion as an example, starting from $\omega(f_2)=\frac{1}{2}f^2\in Z^4(K(\ZZ_2,2),\RZ)$. Then we have
\begin{equation}
	\langle\Theta(s,f),(0,1,2,3)\rangle=\frac{1}{2}\sum_{i=0}^3\langle h^2,(0_0,\cdots,i_0,i_1,\cdots 3_1)\rangle,
\end{equation}
where $h=f+d(\eta s)$. The four terms are listed below.
\begin{align*}
	& i=0:\;&&\langle h,(0_0,0_1,1_1)\rangle \langle h,(1_1,2_1,3_1)\rangle\\
	& i=1:\;&&\langle h,(0_0,1_0,1_1)\rangle \langle h,(1_1,2_1,3_1)\rangle\\
	& i=2:\;&&\langle h,(0_0,1_0,2_0)\rangle \langle h,(2_0,2_1,3_1)\rangle\\
	& i=3:\;&&\langle h,(0_0,1_0,2_0)\rangle \langle h,(2_0,3_0,3_1)\rangle\\
\end{align*}

Among them, we have
\begin{align}
	&\langle f+d(\eta s),(0_0,0_1,1_1)\rangle=s_{01};\\
	&\langle f+d(\eta s),(0_0,1_0,1_1)\rangle=0;\\
	&\langle f+d(\eta s),(0_0,0_1,1_1)\rangle=0;\\
	&\langle f+d(\eta s),(0_0,1_0,2_0)\rangle=f_{012};\\
	&\langle f+d(\eta s),(1_1,2_1,3_1)\rangle=(f+ds)_{123};\\
	&\langle f+d(\eta s),(2_0,2_1,3_1)\rangle=s_{23};\\
	&\langle f+d(\eta s),(2_0,3_0,3_1)\rangle=0.\\
\end{align}
Thus
\begin{equation}
	\begin{aligned}
			\langle\Theta(s,f),(0,1,2,3)\rangle&=\frac{1}{2}(s_{01}(f+ds)_{123}+f_{012}s_{23})\\&=\frac{1}{2}\langle sf+sds+fs,(0,1,2,3)\rangle.
	\end{aligned}
\end{equation}

To compute $L(s,v)$, let $h=d(\eta_1v+\eta_2s)$, we have
\begin{equation}
	\langle L(s,v),(0,1,2)\rangle=\frac{1}{2}\int_{\Delta^2\times \Delta^2}h^2.
\end{equation}

In the integral there are six terms:
\begin{align*}
	& \langle h, (0_0,0_1,0_2)\rangle \langle h, (0_2,1_2,2_2)\rangle=0;\\
	& \langle h, (0_0,1_0,2_0)\rangle \langle h, (2_0,2_1,2_2)\rangle=0;\\
	& \langle h, (0_0,1_0,1_1)\rangle \langle h, (1_1,2_1,2_2)\rangle=0;\\
	& \langle h, (0_0,1_0,1_1)\rangle \langle h, (1_1,1_2,2_2)\rangle=0;\\
	& \langle h, (0_0,0_1,1_1)\rangle \langle h, (1_1,2_1,2_2)\rangle=v_{01}v_{12};\\
	& \langle h, (0_0,0_1,1_1)\rangle \langle h, (1_1,1_2,2_2)\rangle=v_{01}(v_{12}+s_{12}).\\
\end{align*}
The sum is 
\begin{equation}
	\langle L(s,v),(0,1,2)\rangle=\frac{1}{2}\langle v\smile s,(0,1,2)\rangle.
\end{equation}

Thus $\Theta(s,f)=sf+sds+fs$ and $L(s,v)=vs$. These are exactly the solution given in Remark~\ref{remark: alternative solution}.

 Now with Theorem~\ref{thm:Realization} and this prism integration, in principle, one can realize any statistics for Abelian fusion rules, classified by
\begin{equation}
	H^{d+2}\left(\prod_{q=1}^dK(A_i,q),\RZ\right).
\end{equation}

%
%
%

\section{Statistics of invertible excitations from higher gauge theory}

We have introduced three classes of excitation conservation laws: Abelian,
invertible, and non-invertible.  In the last section, we showed that, for
excitations in $d$-dimensional space with Abelian conservation described by a
higher-form symmetry $K(G,q)$, their statistics, defined by the algebra of
hopping operators, are described by a WZW term characterized by a cohomology class
$\omega \in H^{d+2}(K(G,q);\RZ)$.  In this section, we use a WZW term to describe
statistics and generalize this description to invertible excitations with
possible mixed dimensionalities and mixed conservation laws.  The associated
conservation law is
described by a higher group. In fact, the invertible excitations correspond to
the symmetry defects of the higher group and inherit their pointed fusion rules.

\subsection{Review of higher groups and higher gauge theory} \label{hgauge}

Let $K$ be a minimal Kan complex representing a higher classifying space, with homotopy groups $\pi_1(K)=G$, $\pi_2(K)=A_2$, \dots, $\pi_n(K)=A_n$, and all higher homotopy groups trivial. The space $K$ is characterized by (see Appendix \ref{triK}):
\begin{enumerate}
    \item Canonical cochains $x_j$: A $G$-valued 1-cochain $x_1$, an
$A_2$-valued 2-cochain $x_2$, etc., which uniquely label the simplices of $K$.

    \item Automorphisms $\al_j: G \to \text{Aut}(A_j)$.

    \item Postnikov cocycles $k_{j+1}(x_1, \dots, x_{j-1})$: $A_j$-valued
$(j+1)$-cocycle on $K$, as a function of $x_1, \dots, x_{j-1}$.

\end{enumerate}
These canonical cochains satisfy the structural equations:
\begin{align}
\label{k_cond}
\dd x_1 = \id \in G, \qquad \dd_{\al_{j}} x_j =  k_{j+1}(x_1, \dots, x_{j-1}).
\end{align}

To define a higher gauge theory, a gauge field configuration on an extended
$(d+2)$-dimensional spacetime $N^{d+2}$ is a simplicial map $\phi:
N^{d+2}\rightarrow K$. This defines spacetime fields via the pullback $f_j =
\phi^* x_j$. The $G$-valued 1-cochain $f_1$ acts as a flat gauge connection,
while higher cochains $f_j$ act as higher-form gauge fields satisfying modified
flatness conditions:
\begin{align}
\label{acond}
\dd f_1 = \id \in G, \qquad \dd_{\al_{j}} f_j =  k_{j+1}(f_1, \dots, f_{j-1}).
\end{align}

The action amplitude of a twisted higher gauge theory is governed by a $(d+2)$-cocycle $\bar{\omega} \in Z^{d+2}(K, \R/\Z)$:
\[
\ee^{-S(\phi)} = \ee^{2\pi\ii \int_{N^{d+2}} \phi^*\bar\omega(x_1, x_2, \dots)} = \ee^{2\pi\ii \int_{N^{d+2}} \omega(f_1, f_2, \dots)}
\]
The partition function is
\begin{align}
\label{hgNd2}
Z = 
\hskip -3mm
\sum_{\phi: M^{d+1}\to K} 
\hskip -3mm
\ee^{-S(\phi)} 
= 
\hskip -2mm
\sum_{f_i = \phi^* x_i} 
\hskip -2mm
\ee^{2\pi\ii \int_{N^{d+2}} \om(f_1, f_2, \dots)}
.
\end{align}

\subsection{Holographic description of statistics by WZW term in one higher
dimension}

Motivated by the result of the previous section, we propose the following general
description of statistics, defined by the algebra of hopping operators, for
invertible excitations:
\begin{conjecture}\label{cnj:conservation}
The conservation law of invertible excitations in $d+1$-dimensional spacetime
$M^{d+1}$ is encoded by a higher group $\cG$.  Specifically, the excitations may be
viewed as symmetry defects of $\cG$, with their pointed fusion rules
inherited from those of the symmetry defects.  
\end{conjecture}

By Poincar\'e duality, the 1-cocycle $f_1$ describes codimension-1 defects in
$M^{d+1}$ (symmetry defects of a 0-form $G$-symmetry), $f_2$ describes
codimension-2 defects (symmetry defects of a 1-form $A_2$-symmetry), and so on.
The condition $\dd_{\al_{2}} f_2 = k_3(f_1)$ implies that 1-form symmetry
defects can terminate on the intersections of 0-form symmetry defects. Thus,
excitations of different dimensionalities exhibit mixed conservation laws,
encoding the generalized (algebraic) higher symmetry dual to $\cG$.

\begin{conjecture}\label{cnj:WZWSta}
The statistics of invertible excitations with conservation described by a higher
group $\cG$ is described by a WZW term in one-higher dimension given by a
cocycle $\om$ on the classifying space $B\cG$ of the higher group:
\begin{align}
& \ee^{2\pi\ii \int_{N^{d+2}} \om(f_1, f_2, \dots)}
\nonumber\\
& \omega \in H^{d+2}(B\cG;\RZ).
\end{align}
\end{conjecture}

As a physical implication of the above conjecture, the dynamics of these
invertible excitations are described on a boundary $M^{d+1}$ of the $\om$-twisted
higher gauge theory on $N^{d+2}$. The corresponding
path integral is
\begin{align}
\label{ZMN}
 Z = \sum_{f_i = \phi^* x_i} \ee^{-\int_{M^{d+1}} \cL(f_1, f_2, \dots)} \ee^{2\pi\ii \int_{N^{d+2}} \om(f_1, f_2, \dots)}.
\end{align}
The sum runs over fields satisfying the constraints \eqref{acond}.

In the absence of the $(d+2)$-dimensional WZW term $\om$, these boundary
excitations possess trivial (``bosonic'') statistics. Introducing
the WZW twist endows them with statistics. Although
codimension-1 excitations cannot braid, their presence strongly affects the
conservation laws and the allowed topological twists $\bar\om$, governing the
complex interplay between charges and defects.

\subsection{Generalized Spin-Statistics relations and spacetime constraints}

Because statistics are a local property, they must be well-defined on the
simplest closed spacetime, the sphere $M^{d+1} = S^{d+1}$ (see discussion in Section
\ref{maintheorem} and Appendix \ref{app:wzw-realization-details}). Given the canonical disk-like extension $N^{d+2} =
D^{d+2}$, any two extensions differ by a closed manifold $S^{d+2}$. Since the
WZW cocycle $\om(f_1,\cdots,f_d)$ on $S^{d+2}$ is exact, the phase
ambiguity $\ee^{2\pi\ii \int_{S^{d+2}} \om}$ is strictly trivial.  Hence, the
statistics are always well-defined locally.

On spacetime $M^{d+1}$ with a more general homotopy type, the WZW action
$\ee^{2\pi\ii \int_{N^{d+2}} \om}$ may depend on the choice of extension
$N^{d+2}$ if there is no unique canonical one. If an ambiguity exists, the
partition function is ill-defined on $M^{d+1}$. We then say that the homotopy type
of $M^{d+1}$ is \emph{not admissible} for the statistics specified by $\om$. 

This restriction imposes geometric constraints on spacetime,
generalizing the familiar requirement that fermions demand a spin structure.
From this perspective, generalized spin structures are not arbitrary
decorations but mandatory geometric compatibilities demanded by the intrinsic
topological statistics of the system's extended excitations.

\section{Finite fully-pointed fusion higher categories and finite higher
groups}

\label{sec:pointed-fusion-higher-cats}

For invertible excitations with mixed dimensionality and finite types, we
expect their mixed conservation laws and statistics to be classified by finite
fully-pointed fusion higher categories. In this section, we explain the
relation between these categories and finite higher groups: \frmbox{a
fully-pointed fusion $d$-category is classified, up to equivalence, by a pair
$(\cG,[\omega])$, where $\cG$ is a finite $d$-group and $[\omega]\in
H^{d+2}(B\cG;\R/\Z)$.} This mathematical result supports our Conjectures
\ref{cnj:conservation} and \ref{cnj:WZWSta}: the $d$-group $\cG$ encodes the
conservation laws, and the cohomology class $[\omega]$ encodes the statistics
for excitations in $d$-dimensional space.

Here, the phrase ``fully-pointed'' is crucial.  Pointed means that every
simple object is invertible. Fully pointed means that all simple objects and
simple morphisms at all categorical levels are invertible. 

\subsection{Warm-up: ordinary pointed fusion categories}

To build intuition, we first consider the familiar case of a fully-pointed
fusion $1$-category $\mathcal C$ over $\C$. Physically, such a category
describes a system of Abelian anyons (or topological defect worldlines) in a
2+1D topological order. 

The ``simple objects'' of $\mathcal C$ correspond to the anyon types. The
condition that the category is ``fully-pointed'' means that every anyon $g$
has an inverse (its antiparticle) such that their fusion yields the vacuum. In
other words, every anyon has quantum dimension $D=1$. Because of this, their
fusion rules are entirely deterministic and form a finite group $G$, where the
fusion of two anyons $g_1$ and $g_2$ gives a unique anyon $g_1g_2$. 

Every such pointed fusion category is equivalent to $\mathrm{Vec}_G^\omega$,
the category of $G$-graded vector spaces. Physically, the vector spaces
represent the Hilbert spaces of degenerate states when anyons are fused, and
the $G$-grading just keeps track of the total anyon charge. 

However, fusion is a physical process, and quantum mechanics allows for phase
ambiguities when processes are composed. When we fuse three anyons $g_1, g_2,$
and $g_3$, we can pair them up in two different orders: $(g_1 \otimes g_2)
\otimes g_3$ or $g_1 \otimes (g_2 \otimes g_3)$. The associativity isomorphism
relates the quantum states resulting from these two fusion trees. For Abelian
anyons, this isomorphism is simply a $U(1)$ phase, widely known in physics as
the \textbf{$F$-symbol}. Mathematically, this phase is captured by a function
$\omega(g_1, g_2, g_3) \in U(1)$, known as the associator.

For this physical theory to be consistent, macroscopic observables cannot
depend on the order in which we fuse anyons. Fusing three anyons $g_1, g_2,
g_3$ involves a unitary mapping between two different fusion trees,
represented by the $F$-symbol (or associator). Because the anyons are Abelian,
the fusion spaces are 1-dimensional, and this mapping is simply a $U(1)$
phase: \[ |(g_1 g_2) g_3 \rangle = \omega(g_1, g_2, g_3) | g_1 (g_2 g_3)
\rangle.  \]

If we fuse four anyons $g_1, g_2, g_3, g_4$, we can relate the state $|((g_1
g_2)g_3)g_4\rangle$ to $|g_1(g_2(g_3 g_4))\rangle$ through sequences of
$F$-moves. Requiring that the final quantum state be single-valued regardless
of the sequence chosen gives rise to the famous \textbf{pentagon identity}.
Evaluating the two possible paths yields the relation:
\begin{widetext}
\begin{align*}
\text{Path 1: } & |((g_1 g_2)g_3)g_4\rangle \xrightarrow{F} |(g_1(g_2 g_3))g_4\rangle \xrightarrow{F} |g_1((g_2 g_3)g_4)\rangle \xrightarrow{F} |g_1(g_2(g_3 g_4))\rangle \\
&\implies \text{Phase: } \omega(g_1, g_2, g_3) \, \omega(g_1, g_2 g_3, g_4) \, \omega(g_2, g_3, g_4) \\
\text{Path 2: } & |((g_1 g_2)g_3)g_4\rangle \xrightarrow{F} |(g_1 g_2)(g_3 g_4)\rangle \xrightarrow{F} |g_1(g_2(g_3 g_4))\rangle \\
&\implies \text{Phase: } \omega(g_1 g_2, g_3, g_4) \, \omega(g_1, g_2, g_3 g_4).
\end{align*}
\end{widetext}
Equating the phases from both paths gives the pentagon equation:
\begin{align}
&\ \ \ \
\omega(g_1, g_2, g_3) \, \omega(g_1, g_2 g_3, g_4) \, \omega(g_2, g_3, g_4) 
\nonumber\\
&= \omega(g_1 g_2, g_3, g_4) \, \omega(g_1, g_2, g_3 g_4).
\end{align}
Rearranging this equation gives:
\begin{align}
&\ \ \ \ 
(\dd\omega)(g_1, g_2, g_3, g_4) 
\nonumber\\
&\equiv \frac{\omega(g_2, g_3, g_4) \, \omega(g_1, g_2 g_3, g_4) \, \omega(g_1, g_2, g_3)}{\omega(g_1 g_2, g_3, g_4) \, \omega(g_1, g_2, g_3 g_4)} = 1.
\end{align}
Mathematically, this identity exactly matches the 3-cocycle condition for the
associator, meaning $\omega$ is a 3-cocycle, $\omega\in Z^3(G,U(1))$. (In additive
$\R/\Z$ notation, this is written as $\dd\omega = 0$).

Finally, the specific $U(1)$ phases in the $F$-symbol depend on our choice of
local basis vectors for the fusion spaces. We can perform a gauge transformation
by redefining the basis of the 1-dimensional fusion space at each vertex $(g_a, g_b)$:
\[
|g_a, g_b \rangle \longrightarrow \mu(g_a, g_b) |g_a, g_b \rangle,
\]
where $\mu(g_a, g_b) \in U(1)$ is an arbitrary 2-cochain. We must track how this
local basis change affects the macroscopic $F$-symbol $\omega(g_1, g_2, g_3)$. 
The state $|(g_1 g_2) g_3 \rangle$ contains two vertices, $(g_1, g_2)$ and $(g_1 g_2, g_3)$,
so its basis redefinition picks up a factor of $\mu(g_1, g_2) \mu(g_1 g_2, g_3)$. 
Similarly, the state $|g_1 (g_2 g_3) \rangle$ picks up a factor of $\mu(g_2, g_3) \mu(g_1, g_2 g_3)$.
Therefore, under this gauge transformation, the new associator $\omega'$ becomes:
\[
\omega'(g_1, g_2, g_3) = \omega(g_1, g_2, g_3) \frac{\mu(g_2, g_3) \mu(g_1, g_2 g_3)}{\mu(g_1 g_2, g_3) \mu(g_1, g_2)}.
\]
The fractional term on the right is exactly the standard definition of the
3-coboundary $\dd\mu$ of the 2-cochain $\mu$. Thus, the gauge transformation
modifies the associator by $\omega' = \omega \cdot (\dd\mu)$. 

Because physics must be gauge-invariant, the physical system does not depend on
the specific cocycle $\omega$, but rather on its equivalence class up to coboundaries.

Therefore, the physically distinct invertible anyon models are classified
exactly by pairs $(G,[\omega])$.  The group $G$ describes the pointed fusion
rules (the conservation law), and the cohomology class $[\omega]\in H^3(BG,U(1))$
describes the statistics. Thus, we have the well-known classification
\cite{EtingofNikshychOstrik2005,EGNO2015}: \frmbox{ pointed fusion
$1$-categories $\longleftrightarrow$ $\{(G,[\omega]) \mid [\omega]\in
H^3(BG,\RZ)\}$.  }

\subsection{Finite higher groups and the cocycle twist}

For a general fully-pointed fusion $n$-category, the invertible simple objects
and higher morphisms naturally form a finite $n$-group $\cG$.
Homotopy-theoretically, $\cG$ is described by its classifying space $B\cG$, a
finite homotopy $n$-type. Its nonzero homotopy groups are finite and given by
$\pi_1(B\cG)=G$ and $\pi_i(B\cG)=A_i$ for $2\le i\le n$.  In particular,
$\pi_1(B\cG)=G$ labels the simple objects, $\pi_2(B\cG)=A_2$ labels the simple
1-morphisms, \etc, of the fusion $n$-category.

These homotopy levels (\ie the categorical levels corresponding to excitations
with different dimensionalities) are assembled via Postnikov classes
$k_{i+1}\in H^{i+1}(B\cG_{i-1};A_i)$.  For example, a finite $2$-group is
specified by \[ G=\pi_1(B\cG), \qquad A_2=\pi_2(B\cG), \] an action \[ \al_2:
G\to \Aut(A_2), \] and a Postnikov class \[ k_3\in H^3(BG;A_2^{\al_2}).  \]
Thus \[ B\cG \simeq K_{\al_2,k_3}(G,1;A_2,2).  \]

Given a finite $n$-group $\cG$, one forms its linearization, denoted
$n\mathrm{Vec}_{\cG}$. This is the higher-categorical analogue of group-graded
vector spaces. Explicitly, the categorical levels of $n\mathrm{Vec}_{\cG}$ are
constructed from the homotopy groups of $B\cG$ as follows:
\begin{itemize}
    \item \textbf{Objects:} The simple objects are in one-to-one correspondence with the elements $g \in \pi_1(B\cG)$. A general object is a formal finite direct sum of these simple objects.
    \item \textbf{$k$-morphisms ($1 \le k \le n-1$):} There are no nonzero $k$-morphisms between distinct simple $(k-1)$-morphisms. The simple $k$-morphisms from a simple $(k-1)$-morphism to itself are labeled by the elements $a_{k+1} \in \pi_{k+1}(B\cG)$. General $k$-morphisms are formal direct sums of these.
    \item \textbf{$n$-morphisms (top level):} This is where the category is ``linearized'' over $\C$. Between distinct simple $(n-1)$-morphisms, the space of $n$-morphisms is the zero vector space. From a simple $(n-1)$-morphism to itself, the $n$-morphisms form a $1$-dimensional complex vector space $\C$. General $n$-morphisms are linear maps (matrices of complex numbers).
\end{itemize}
The fusion (tensor product) of these objects and morphisms follows the group operations in $\pi_k(B\cG)$, with the lower-level groups acting on the higher-level groups precisely as dictated by the Postnikov classes of $B\cG$. 

However, the $n$-group alone only specifies the fusion rules and the strict skeletal structure. To fully specify the fusion $n$-category, the highest-level coherence isomorphism (the generalized associator mapping between the $1$-dimensional vector spaces at the $n$-morphism level) can be twisted by a cocycle $\omega\in Z^{n+2}(B\cG;U(1))$. 

Up to coboundaries (which correspond to changing the basis of the $n$-morphisms), this twist is specified by a cohomology class $[\omega]\in H^{n+2}(B\cG;U(1))$, yielding the twisted category $n\mathrm{Vec}_{\cG}^{\omega}$.

This gives the general classification
\cite{BaezLauda2004,DouglasReutter2018,KongTianZhou2020}: 
\begin{align}
&\ \ \ \
\text{fully-pointed fusion
$n$-categories} 
\nonumber\\
&\longleftrightarrow \{ (\cG,[\omega]) \mid [\omega]\in
H^{n+2}(B\cG;\RZ) \}.  
\end{align}
(For recent classification results for non-pointed
fusion $2$-categories, see \cite{DecoppetEtAl2024Fusion2Categories}.)

We note one important distinction between the mathematics and the physics. In mathematics, one further quotients these pairs by group equivalences $f \in \Aut(\cG)$, where $(\cG,[\omega]) \sim (\cG,f^*[\omega])$. However, in physics, we do not quotient by this equivalence, because the elements of $\cG$ act as distinct physical labels for excitations.


\section{Examples} 

\subsection{Fermion statistics in general dimensions}

In this subsection, we consider fermion statistics in $d$-dimensional space 
($d\geq 3$) with $\Z_2$ conservation, corresponding to $\omega(x_d)=\frac{1}{2}\Sq^2x_d$, where $x_d\in Z^d(K(\Z_2,d),\Z_2)$ is the canonical generator. $\omega$ generates
\[
H^{d+2}\bigl(K(\Z_2,d);\R/\Z\bigr)\cong \Z_2
\]
for $d\ge 3$.

After specifying the statistics by cohomology class $[\omega]$, one can construct the corresponding hopping operators using Theorem~\ref{thm:Realization}, and the relevant function $L$ can be constructed by prism integral in Section~\ref{subsec: prism integral}. But for fermions, it is easier to solve the descendant functions in Section~\ref{subsec: decendant functions} directly.

First, we solve the equation 
\begin{equation}
	d\Theta(s_{d-1},f_d)=\omega(f_d+\dd s_{d-1})-\omega(f_d).
\end{equation}
Using Eq.~\eqref{Sq-cross}, the right hand side has
\begin{equation}
	\begin{aligned}
		&\Sq^2(f+\dd s)-\Sq^2(f)
		\\&=\dd\left(\gSq^2( s)+\dd s\hcup{d-2} f\right)
	\end{aligned}
\end{equation}
So a solution is 
\begin{equation}
	\Theta(s,f)=\frac{1}{2}\left(\gSq^2( s)+\dd s\hcup{d-2} f\right).
\end{equation}
It indeed satisfies $\Theta(0,f)=0$. Also, we have
\begin{equation}
	\beta(s)=\frac{1}{2}\gSq^2( s).
\end{equation}

Next, we have
\begin{equation}
	\begin{aligned}
		&\dd L(s,v)\\&=\Theta(s,\dd v)-\beta(s+v)+\beta(v)\\
		&= \frac{1}{2}\left(\gSq^2( s)+\dd s\hcup{d-2} \dd v+\gSq^2(v)-\gSq^2(s+v)\right)\\
		&=\frac{1}{2}\dd (\dd s\hcup{d-1}v+s\hcup{d-2}v).
	\end{aligned}
\end{equation}
So a solution is
\begin{equation}
	L(s,v)=\frac{1}{2}\dd s\hcup{d-1}v+s\hcup{d-2}v
\end{equation}
and the hopping operator is
\begin{equation}
	U(s)|v\rangle=(-1)^{\int_X \dd s\hcup{d-1}v+s\hcup{d-2}v}|v+s\rangle.
\end{equation}

Note that in $d=2$ dimensional space, $\ZZ_2$ conserved particle can also be
semion, which is described by the Pontryagin square
\begin{equation}
	\omega(x_2)
=\frac{1}{4}\left(\t x_2 ^2+\t x_2 \hcup{1}\dd \t x_2\right)
=\frac14 \gSq^2 \t x_2,
\end{equation}
where $\t x_2$ is a specific lift from
$\ZZ_2$ coefficient to $\ZZ$ coefficient, such as $0\mapsto 0$ and $1\mapsto
1$.  The above $\omega(x_2)$ generates $H^{4}\bigl(K(\Z_2,2);\R/\Z\bigr)\cong
\Z_4$.

\subsection{Statistics of Abelian anyons}

As a first example, we consider the statistics of point-like excitations in $2+1$-dimensional spacetime. The Poincar\'e dual of the particle worldlines is represented by a $2$-cocycle $f_2$ on spacetime $M^{2+1}$. Assuming the particles obey pointed fusion rules described by a finite Abelian group $A$, $f_2$ is an $A$-valued cocycle.

Equivalently, $f_2$ can be viewed as the pullback of the canonical $A$-valued $2$-cocycle $x_2$ from the Eilenberg-MacLane space $K(A,2)$:
\[
f_2=\phi^*x_2, \qquad \phi:M^{2+1}\to K(A,2).
\]

According to our construction, the possible statistics of particles with fusion group $A$ are classified by:
\[
H^4(K(A,2),\mathbb R/\mathbb Z).
\]
This cohomology group is naturally identified with the group of quadratic functions $q:A\to \mathbb R/\mathbb Z$. Here, ``quadratic'' means that $q(-a)=q(a)$, and the function
\[
b_q(a,b):=q(a+b)-q(a)-q(b)
\]
forms a symmetric bilinear pairing $b_q:A\times A\to \mathbb R/\mathbb Z$. Physically, $q(a)$ determines the topological spin of the particle labeled by $a$, while $b_q(a,b)$ gives the mutual braiding phase between particles $a$ and $b$. Thus,
\[
H^4(K(A,2),\mathbb R/\mathbb Z) \cong \mathrm{Quad}(A,\mathbb R/\mathbb Z).
\]

A representative class in $H^4(K(A,2),\mathbb R/\mathbb Z)$ is obtained by evaluating a universal quadratic expression on the canonical $A$-valued $2$-cocycle $x_2$, written schematically as $\om_4(x_2)=q(x_2)$. This indicates that the cocycle $\om_4$ is built from cup products and Pontryagin squares applied to $x_2$. The resulting spacetime action is obtained by pullback: $\phi^*\om_4(x_2)=\om_4(f_2)$.

For instance, consider $A=\mathbb Z_N$. The classification yields:
\[
H^4(K(\mathbb Z_N,2),\mathbb R/\mathbb Z) \cong
\begin{cases}
	\mathbb Z_N, & N\text{ odd},\\[2mm]
	\mathbb Z_{2N}, & N\text{ even}.
\end{cases}
\]
A convenient cocycle representative is $\om_4 = \frac{k}{2N}\,\mathfrak P(x_2)$, where $\mathfrak P(x_2)$ is the Pontryagin square of $x_2$. The parameter $k$ takes values in $\mathbb Z_{2N}$ for even $N$, effectively reducing to $\mathbb Z_N$ for odd $N$. 

If $x_2$ is lifted to an integer-valued $2$-cochain $\widetilde x_2$, the Pontryagin square at the cochain level is:
\[
\mathfrak P(x_2) = \widetilde x_2\smile \widetilde x_2 + \widetilde x_2 \hcup{1} \dd \widetilde x_2 = \gSq^2 x_2 \quad \mathrm{mod}\ 2N .
\]
Pulling this cocycle back to spacetime gives the WZW term:
\[
\om_4(f_2) = \frac{k}{2N} \left( \widetilde f_2\smile \widetilde f_2 - \widetilde f_2\hcup{1}\dd\widetilde f_2 \right),
\]
where $f_2$ is the $\mathbb Z_N$-valued particle-worldline cocycle and $\widetilde f_2$ is its integer lift.

For a general finite Abelian group $A=\bigoplus_I \mathbb Z_{N_I}$, we decompose the canonical cocycle as $x_2=\bigoplus_I x_2^{\Z_{N_I}}$. A general degree-$4$ cocycle is a sum of diagonal Pontryagin squares and off-diagonal cup products:
\[
\om_4 = \sum_I \frac{k_I}{2N_I}\,\mathfrak P(x_2^{\Z_{N_I}}) + \sum_{I<J} \frac{k_{IJ}}{N_{IJ}}\, x_2^{\Z_{N_I}}\smile x_2^{\Z_{N_J}},
\]
where $N_{IJ}=\gcd(N_I,N_J)$. The parameters are $k_I \in \mathbb Z_{2N_I}$ (or $\mathbb Z_{N_I}$ for odd $N_I$) and $k_{IJ}\in \mathbb Z_{N_{IJ}}$.

Notice that diagonal terms require the Pontryagin square because the coefficient is $1/(2N_I)$. The cup-$1$ correction ensures the expression remains closed modulo $2N_I$. Conversely, the mixed term $x_2^{\Z_{N_I}} \smile x_2^{\Z_{N_J}}$ is already a well-defined $\mathbb Z_{N_{IJ}}$-valued $4$-cocycle, so no cup-$1$ correction is needed.

In summary, statistics for point-like excitations with fusion group $A=\prod_I \Z_{N_I}$ in $2+1$-dimensional spacetime are classified by:
\[
H^4(K(A,2),\mathbb R/\mathbb Z) \cong \left(\bigoplus_I \mathbb Z_{\widetilde N_I}\right) \oplus \left(\bigoplus_{I<J}\mathbb Z_{N_{IJ}}\right),
\]
where $\widetilde N_I = 2N_I$ for even $N_I$ and $\widetilde N_I = N_I$ for odd $N_I$.

\subsection{Statistics of Abelian strings in $3+1$-dimensional spacetime}
\label{strsta}

Next, we consider the statistics of Abelian strings in $3+1$-dimensional spacetime. The Poincar\'e dual of the string worldsheets is a $2$-cocycle $f_2$ on spacetime $M^{3+1}$. Assuming pointed fusion rules described by a finite Abelian group $A$, the statistics of strings are classified by:
\[
H^5(K(A,2),\mathbb R/\mathbb Z).
\]

For a generic finite Abelian group $A=\bigoplus_I \mathbb Z_{N_I}$, we find:
\[
H^5(K(A,2),\mathbb R/\mathbb Z) \cong \left(\bigoplus_I \mathbb Z_{\gcd(N_I,2)}\right) \oplus \left(\bigoplus_{I<J}\mathbb Z_{N_{IJ}}\right).
\]
The first set of terms describes the self-statistics of strings and exists only for even $N_I$. The second set describes mutual statistics between strings carrying charges in different cyclic factors.

Using the canonical decomposition $x_2=\bigoplus_I x_2^{\Z_{N_I}}$, we
construct explicit cocycle representatives with the Bockstein operation
$\Bs_{N_I} x_2^{\Z_{N_I}} = \frac{1}{N_I}\dd{\t x_2^{\Z_{N_I}}} $ (where $\t
x_2$ is an integer lift of $x_2$). A general representative takes the form
\begin{align}
\om_5 &= \sum_I \frac{k_I}{N_I}  x_2^{\Z_{N_I}} \Bs_{N_I} x_2^{\Z_{N_I}} \nonumber\\
&\quad + \sum_{I<J} \frac{k_{IJ}}{N_{IJ}} x_2^{\Z_{N_I}} \Bs_{N_{IJ}} x_2^{\Z_{N_J}} ,
\end{align}
where $k_I\in \mathbb Z_{\gcd(N_I,2)}$ and $k_{IJ}\in \mathbb Z_{N_{IJ}}$.
The diagonal term $ \frac{k_I}{N_I}  x_2^{\Z_{N_I}} \Bs_{N_I} x_2^{\Z_{N_I}}$
can also be written as, with $d=3$,
\begin{align}
 \frac{k_I}{N_I^2} \gSq^3 x_{d-1}^{\Z_{N_I}},
\end{align}
or as
\begin{align}
 \frac{k_I}{2} \gSq^2 \Bs_n x_{d-1}^{\Z_{N_I}} .
\end{align}
The later two expressions are valid for any spatial dimension $d>1$
(see Appendix \ref{Hd2}).

Pulling back to spacetime ($f_{2,I}=\phi^*x_2^{\Z_{N_I}}$), the corresponding topological term is:
\[
\om_5 = \sum_I \frac{k_I}{N_I}\, f_{2,I} \Bs_{N_I}f_{2,I} + \sum_{I<J} \frac{k_{IJ}}{N_{IJ}}\, f_{2,I} \Bs_{N_{IJ}}f_{2,J}.
\]
These cocycles physically encode the fermionic loop statistics and Bockstein braiding statistics of Abelian strings in $3+1$-dimensional spacetime.

\subsection{Spacetime condition for nontrivial $\mathbb Z_2$ string statistics in 3+1D}
\label{z2strw3}

Consider a string with $\Z_2$ conservation in $3+1$-dimensional spacetime. The nontrivial statistics are described by the $5$-cocycle:
\[
\om_5(f_2) = \frac12\,f_2 \Sq^1 f_2 = \frac12\,f_2 \Bs_2 f_2 .
\]
The WZW amplitude is defined by choosing an extension $N^5$ (where $\partial N^5=M^4$). This amplitude is independent of the extension if the phase on any closed comparison $5$-manifold $W^5$ is trivial:
\[
\int_{W^5}f_2 \Sq^1f_2=0 \mod 2.
\]

On a closed oriented $5$-manifold $W^5$, we have 
\[
\int_{W^5}f_2 \Sq^1f_2 = \int_{W^5}f_2  \w_3(TW) \mod 2.
\]
First the left side is linear in $f_2$. Since
on a closed oriented $5$-manifold the Wu-class $v_1=\w_1=0$, so
$\int_{W^5}\Sq^1(\text{degree }4)=\int v_1(\cdot)=0$. Hence for
$\deg(f\,g)=4$,
\begin{equation}
0=\int_{W^5}\Sq^1(f_2\, g_2)
=\int_{W^5}\Sq^1 f_2\, g_2+\int_{W^5} f_2\,\Sq^1 g_2 ,
\end{equation}
so the pairing $B(f,g)=\int_{W^5} f\,\Sq^1 g$ is symmetric, and the cross terms
in $B(f+g,f+g)$ cancel mod $2$. Thus
$f_2\mapsto\int_{W^5} f_2\,\Sq^1 f_2$ is a \emph{linear} functional on
$H^2(W^5;\Z_2)$.

Due to the linear condition,
by Poincar\'e duality there is a unique $c_3\in H^3(W^5;\Z_2)$ with
\begin{equation}
\int_{W^5} f_2\,\Sq^1 f_2=\int_{W^5} f_2\,c_3
\qquad\text{for all } f_2\in H^2(W^5;\Z_2).
\end{equation}
Naturality forces $c_3$ to be a degree-$3$ polynomial in the $\w_i(TW)$; with
$\w_1(TW)=0$ the only option is $c_3=\epsilon\,\w_3(TW)$, $\epsilon\in\{0,1\}$.

Take ${W^5}=SU(3)/SO(3)$, the generator of $\Omega^{SO}_5=\Z_2$, where
$H^2=\langle \w_2\rangle$, $H^3=\langle \w_3\rangle$, $\w_3=\Sq^1 \w_2$, and
$\int_{W^5} \w_2 \w_3=1$. With $f_2=\w_2$,
\begin{align}
\int_{W^5} f_2\,\Sq^1 f_2&=\int_{W^5} \w_2\,\Sq^1 \w_2=\int_W \w_2 \w_3=1.
\end{align}
Therefore $\eps\neq 0$.  So $\epsilon=1$ \ie
\begin{align}
\label{fSqfw3}
\int_{W^5} f_2\,\Sq^1 f_2=\int_{W^5} f_2 \w_3(TW)
\end{align} 
We see that the ambiguity vanishes for all background
fields $f_2$ if and only if $\w_3(TW)=0$.  This is the direct string analogue
of the spin condition for fermionic particles. To support nontrivial
statistics for $\mathbb Z_2$ strings without ambiguity, the spacetime must
admit a $\w_3$-structure ($\w_3=0$). While a full spin structure ($\w_1=\w_2=0
\implies \w_3=\Sq^1\w_2=0$) is sufficient, the condition $\w_3=0$ is strictly
weaker and represents the precise, minimal tangential condition detected by
string statistics.

Using $\Sq^1 \w_2=\w_3$, we also have
\begin{equation}
\int_{W^5} f_2\w_3=\int_{W^5} \w_2\,\Sq^1 f_2=\int_{W^5} \Sq^2\Sq^1 f_2 ,
\end{equation}
the last step being $\int_{W^5}\Sq^2 y_3=\int_{W^5} v_2\, y_3=\int_{W^5} \w_2\, y_3$
(Wu formula, with $v_2=w_2$).
We see that
\begin{align}
 \int_{W^5} f_2\,\Sq^1 f_2 = \int_{W^5} \Sq^2\Sq^1 f_2 .
\end{align}

\subsection{Statistics of mixed $\Z_2$-particles and $\Z_2$-strings in
$d+1$-dimensional spacetime}

\label{partstr}

We now explore a system containing both point-like particles and strings in
$d$-dimensional space ($d\ge 2$), where both excitations obey $\Z_2$ fusion
rules. 

Let $f_{d-1}$ and $f_d$ denote the $\Z_2$-valued Poincar\'e duals of the
string worldsheets and particle worldlines, respectively in spacetime. Their
conservation laws are governed by a higher-group classifying space $K$ with
$\pi_{d-1}(K)=\Z_2$ and $\pi_d(K)=\Z_2$. This space is specified by a
Postnikov class $k_{d+1}\in H^{d+1}(K(\Z_2,d-1),\Z_2)$. For $d\ge 2$, this
cohomology group is generated by $\Sq^2 x_{d-1}$. Thus, we have two possible
conservation laws labeled by $\rho \in \{0,1\}$, where $k_{d+1} = \rho \Sq^2
x_{d-1}$. This gives rise to two higher-group classifying spaces 
\begin{align}
K_\rho
= K_{\rho \Sq^2
x_{d-1}}(\Z_2, d-1; \Z_2,d)
.
\end{align}
The spacetime conservation laws take the form:
\[
\dd f_{d-1}\se{2}0, \qquad \dd f_d\se{2}\rho\,\Sq^2 f_{d-1} .
\]

\subsubsection{The untwisted conservation law ($\rho=0$)}

When $\rho=0$, particles and strings are independently conserved ($\dd f_{d-1}\se{2}0$ and $\dd f_d\se{2}0$). The statistics are classified by $H^{d+2}(K_0,\RZ)$. 

We find two universal self-statistics terms for $d>3$: string self-statistics
term $\frac12\Sq^2\Sq^1 f_{d-1}$ and particle self-statistics term
$\frac12\Sq^2 f_d$, which generate $H^{d+2}(K_0,\RZ) \cong (\Z_2)^2$.

For $d=3$, a simple particle-string mutual-statistics term
($\frac12 f_2 f_3$) is also present. Thus, $H^{5}(K_0,\RZ) \cong (\Z_2)^3$.

For $d=2$, string self-statistics term $\frac12\Sq^2\Sq^1 f_{1}$ becomes a
$\RZ$-valued coboundary, and there no non-trivial string self-statistics.  On
the other hand, $\Z_2$ conserved particle can also be semion, which is
described by the Pontryagin square $\omega(f_2) =\frac14 \gSq^2 f_2$,
which generate $H^{4}(K_0,\RZ) \cong \Z_4$.

\subsubsection{The twisted conservation law ($\rho=1$)}

When $\rho=1$, the particle current is no longer independently conserved: $\dd
f_d \se{2} \Sq^2 f_{d-1}$. Physically, specific self-intersections of the
string worldsheets act as sources for the particle current.  The statistics of
the intertwined $\Z_2$-conserved particle-string system are classified by 
\begin{align}
H^{d+2}(K_1,\RZ) .
\end{align}

One of the generator of $H^{d+2}(K_1,\RZ)$
is represented by the secondary-operation cocycle:
\begin{align}
\label{omd2}
\om_{d+2}(f_{d-1},f_d) = \frac12\,\gSq^2 f_d + \frac18\,\gSq^3 f_{d-1} 
+ \frac12\,\Del_{d+2}(f_{d-1}),
\end{align}
where 
\begin{align}
\label{dDel}
\dd\Del_{d+2}(f_{d-1}) \se{2} \Sq^2\Sq^2f_{d-1} +
\Sq^3\Sq^1f_{d-1}.
\end{align}
To understand this result, let us first check if $\om_{d+2}(f_{d-1},f_d)$
is a cocycle or not.  From \eqn{Sqd1}, we find that
\begin{align}
& \dd [
\frac12 \gSq^2 f_d + \frac{1}{8} \gSq^3 f_{d-1} ] 
 \se{1}
\frac12 \gSq^2 \dd f_d +\frac{1}{8} \gSq^3 \dd f_{d-1}
\nonumber\\
&\ \ \ \
 = 
\frac12 \gSq^2 \gSq^2 f_{d-1} 
+\frac12 \gSq^3 \Bs_2 f_{d-1} 
\nonumber\\
& \ \ \ \
\se{1}
\frac12 (\gSq^2 \gSq^2 f_{d-1}+\gSq^3 \gSq^1 f_{d-1})
.
\end{align}
In last two equations, we have used
$\Bs_2 f_{d-1} = \frac12 \dd f_{d-1}$ and
\eqref{Sq1Bs2}.
Since $f_{d-1}$ is a $\Z_2$-valued cocycle, the above can be rewritten as
\begin{align}
&
 \dd [
\frac12 \gSq^2 f_d + \frac{1}{8} \gSq^3 f_{d-1} ] 
\nonumber \\
& \se{1}
\frac12 (\Sq^2 \Sq^2 f_{d-1}+\Sq^3 \Sq^1 f_{d-1})
.
\end{align}
$\Del_{d+2}(f_{d-1})$ is used to cancel the above via $\dd
\Del_{d+2}(f_{d-1})$.  Due to the Adem relation \eqref{Sq2Sq2Sq3Sq1} of the
Steendrod squares, \eqref{dDel} always has solutions.
For example, when $d=2$ 
\begin{align}
\Del_4(f_1) &\se{2} 0,
\end{align}
since for a  $\Z_2$-valued cocycle $f_1$, $\Sq^2 f_{1} \se{2} 0$ and $\Sq^1
f_{1} \se{2}0$.  Therefore,  $\om_{d+2}(f_{d-1},f_d)$ is indeed a cocycle.

For $d=2$, $\frac{1}{8} \gSq^3 f_{1} = 0$.  The above calculation and the fact
that $\Del_4(f_1)=0$ implies that $\frac12 \gSq^2 f_2$ is a $\RZ$-valued
cocycle, despite $f_2$ is not a cocycle.  Our calculation indicates that the
condition $\dd f_2 = \gSq^2 f_1$ is enough to make $\frac12 \gSq^2 f_2$ to be
a cocycle.  In this case $\om_{4}(f_{1},f_2) = \frac12\,\gSq^2 f_2$ is of
order 2, and endows the particle with the Fermi statistics.  Thus $\Z_2
\subset H^{4}(K_1,\RZ)$.  We know that the graded group
\begin{align}
\text{gr} \ H^{4}(K_1;\RZ)
& \subset
 H^{4}(K(\Z_2,1)\times K(\Z_2,2);\RZ)
\nonumber\\
&=
 H^{4}(K(\Z_2,2);\RZ) = \Z_4
,
\end{align}
since $H^{4}(K(\Z_2,1);\RZ) = 0$.  The generator  $\frac14\,\gSq^2 f_2$ of
$H^{4}(K(\Z_2,2);\RZ) = \Z_4$ endows the particle $f_2$ with a semion
statistics.  But  $\frac14\,\gSq^2 f_2$ is a cocycle when only $f_2$ is a
cocycle.  When  $\dd f_2 = \gSq^2 f_1$, $\frac14\,\gSq^2 f_2$ is not a
cocycle.  Only $\frac12\,\gSq^2 f_2$ is a  cocycle (see \eqref{Sqd1}).
Therefore,
\begin{align}
H^{4}(K_1;\RZ) = \Z_2 
\end{align}
which is generated by $\frac12\,\gSq^2 f_2$ with $\dd f_2 = \gSq^2 f_1$.

For $d= 3$, we note that twice this generator $\om_{5}(f_{2},f_3)$ yields the
nontrivial string self-statistics term (see Sections \ref{strsta} and
\ref{Hd2}): 
\begin{align}
2\om_{5}(f_2,f_3) \se{1,\dd} \frac14\,\gSq^3 f_2
= \frac14 f_2 \dd f_2
= \frac12 f_2 \Bs_2 f_2
,
\end{align}
describing a fermionic string and a bosonic particle with a mixed conservation
law: $\dd f_{d-1} \se{2} 0$ and $\dd f_d \se{2} \Sq^2 f_{d-1}$.  
Thus $\om_{5}(f_2,f_3)$ is of order 4.

Using  Serre spectral sequence, in Appendix \ref{SSK1}, we show that
$\operatorname{gr}H^5(K_1;\RZ)\cong \Z_2\oplus\Z_2$.  Thus
\begin{align}
H^{5}(K_1,\RZ) 
=
\Z_4 
.
\end{align}

For the $k=1$ or $k=3$ cases of this $\Z_4$ statistics group, the system
bears a striking resemblance to a fermion and its induced $p$-wave
topological superconducting (pTS) string. In those configurations, the particle
carries Fermi statistics and $\Z_2$ fusion. Such fermions can condense along
a string to form a pTS string \cite{K0131}, which possesses $\Z_2$
conservation and braids trivially with the fermions. Furthermore, pTS strings
and $\Z_2$ fermions exhibit the exact same mixed conservation law $\dd f_d
\se{2} \Sq^2 f_{d-1}$ \cite{KT170108264,WangGu2018}. 

For $d > 3$, the generator $\om_{d+2}(f_{d-1},f_d)$ is also of order 4.
Thus
$\Z_4 \subset
H^{d+2}(K_1,\RZ) $. We believe that
\begin{align}
H^{d+2}(K_1,\RZ) 
=\Z_4 
,
\end{align}
since if the particle $f_d$ is a fermion, it must couple to string $f_{d-1}$
in the way described by \eqref{omd2}.  But 
$H^{d+2}(K_1,\RZ)$ could be larger, since
string could have its own
independent statistics described by $ \om_{d+2}(f_{d-1},f_d) $ that depend
only on $f_{d-1}$.  Such a $ \om_{d+2}(f_{d-1}) $ is in
\begin{align}
H^{d+2}(K(\Z_2,d-1),\RZ) = \Z_2 \text{ generated by }  \frac14 \gSq^3 f_{d-1} .
\end{align}
But the order-2 $\frac14 \gSq^3 f_{d-1}$ is already included in $2
\om_{d+2}(f_{d-1},f_d)$.  Also, for $d>3$ there is no mutual statistics
between the particle and the string.

\appendix

\section{Fermionic statistics of Majorana fermion system}
\label{sec: majorana fermions}

In our definition~\ref{defRealization} of a realization, configuration states $\{|a\rangle\}$ need not be orthogonal, and they need not span the whole
Hilbert space $\mathcal H$. This is slightly different from the original definition in Ref.~\cite{Xue2026StatisticsAbelian}, where configuration states are assumed to be either orthogonal or collinear. This assumption is unnecessary for defining statistics, and there are indeed physical systems without it. 

We start from the discussion of elementary fermions in Section~\ref{subsec: complex fermions}. This realization used only the Majorana operators $\gamma_\sigma$
inside the hopping operators.  This suggests a smaller realization in which
each $2$-simplex carries a single Majorana variable.  Suppose that
$N=|X_2|$ is even.  Let $\cH_{\rm M}$ be an irreducible module of the Clifford
algebra generated by
\begin{equation}
	\gamma_\sigma,
	\qquad
	\{\gamma_\sigma,\gamma_{\sigma'}\}=2\delta_{\sigma\sigma'},
\end{equation}
so that
\begin{equation}
	\dim\cH_{\rm M}=2^{N/2}.
\end{equation}
This is much smaller than the complex-fermion Hilbert space, for which
$\dim\cH_{\rm F}=2^N$.  Since the number of configurations in
$B^2(X,\ZZ_2)$ on $S^2$ is $2^{N-1}$, the corresponding configuration states in
$\cH_{\rm M}$ cannot be linearly independent.

Choose a normalized reference state $|0\rangle_{\rm M}\in\cH_{\rm M}$ and define
\begin{equation}
	\label{eq:configuration states majorana}
	|a\rangle_{\rm M}
	=
	\prod_{\sigma\in X_2}^{<}\gamma_\sigma^{a_\sigma}|0\rangle_{\rm M},
	\qquad
	a\in B^2(X,\ZZ_2).
\end{equation}
Together with the same operators $M(s)$ in Eq.~\eqref{eqMajoranaOperator},
these states satisfy the configuration axiom and the locality axiom by the
Clifford algebra.  Hence they also define a realization of $m^2(X,\ZZ_2)$ with
fermionic statistics.  This construction works for any choice of
$|0\rangle_{\rm M}$.  For a generic choice, different configuration states may
be neither orthogonal nor collinear.

\section{Statistics for non-invertible excitations}
\label{AbNonAbElementary}

\subsection{Holographic description from state-sum theories}

The holographic framework of using higher gauge theory to describe 
statistics in one lower dimension naturally extends to non-invertible excitations.
This is achieved by replacing lattice higher gauge theory with a more general
state-sum theory. Specifically, a $(d+2)$-dimensional state-sum theory provides
the bulk description for the boundary excitations' conservation laws and
statistics. In the state-sum description of the path integral, the action amplitude
is dictated by a tensor $T$ on top-dimensional simplices, whose indices
correspond to labels from lower-dimensional faces (representing the
Poincar\'e-dual currents of excitations of varying dimensions).

Mixed conservation laws and non-pointed fusion rules (\eg, $a \otimes b = c_1
\oplus c_2 \oplus \dots$) are seamlessly incorporated into the admissible
labelings of the tensor $T$, \ie labelings that make the tensor non-zero.  The
complex values of the tensor $T$ encode the generalized statistics.
Retriangulation invariance ensures the theory is topological, which
simultaneously enforces consistency on the fusion and braiding data of the
boundary excitations.

For non-invertible excitations, braiding acts via matrices on multi-dimensional
fusion spaces, realizing non-Abelian anyons. The underlying algebraic
structure shifts from a simple higher group to a general fusion
$d$-category.  \frmbox{For excitations with
mixed dimensionalities in $d+1$-dimensional spacetime, their conservation laws and
statistics are characterized and classified by a state-sum theory in one higher
dimension, \ie by a fusion $d$-category.} A higher group merely represents
the specialized, pointed (invertible) case of this broader categorical
framework.

\subsection{Hopping algebra description: a potential generalization}

The axiomatic framework of statistics established in Ref.~\cite{Xue2026StatisticsAbelian} applies only to Abelian cases, \ie, the configurations form an Abelian group $A$ and hopping operators $S$ act additively. The corresponding statistics are classified by cohomology classes $H^{d+2}(B\cG,\RR/\ZZ)$, where $\cG=\prod_{q=1}^d K(\pi_q,q-1)$ is a higher Abelian group. 

It is natural to conjecture that in a non-Abelian (but still invertible) generalization, the classification is still $H^{d+2}(B\cG,\RR/\ZZ)$, while $\cG$ can be a generic higher group. Then excitations should be viewed as boundary descendants of a twisted $\cG$-gauge theory: $A$ corresponds to gauge potentials, and $S$ corresponds to elementary gauge transformations. In that situation, the configuration axiom might be replaced by the following formula:

\begin{equation}
	U(s)|a\rangle=e^{i\theta(s,a)}|s\cdot a\rangle.
\end{equation}
Here we replace the addition map $(s,a)\mapsto a+\partial s$ by a generic action $S\times A\to A$. 

We define a process $P$ as an element of the free group $\operatorname{F}(S)$:

\begin{equation}
	P=s_n^{\epsilon_n}\cdots s_2^{\epsilon_2}s_1^{\epsilon_1},
	\qquad \epsilon_i=\pm 1.
\end{equation}

For any point $x$ in the spatial manifold $\cX$, we define $x\notin \supp(P)$ if deleting all $s\in S$ that $x\notin \supp(s)$ in $P$ makes it trivial in $\operatorname{F}(S)$. We say $P$ has empty support if $x\notin \supp(P),\forall x\in \cX$.

\begin{definition}
	An \textit{(invertible) excitation complex} $m$ consists of the data $(A,S,\cdot,\supp)$:
	\begin{enumerate}
		\item a finite set $A$ called configurations;
		\item a finite set $S$ and a map $\cdot:S\times A\to A$ such that for every $s\in S$, $a\mapsto s\cdot a$ is a bijection;
		\item a topological space $\cX$ and a subspace $\supp(s)\subset \cX$ for each $s\in S$, such that
		\begin{equation}
			\supp(P)=\emptyset\implies P\cdot a=a,\quad\forall a\in A.
		\end{equation}
	\end{enumerate}
\end{definition}

\begin{definition}
	A realization of the (invertible) excitation complex $(A,S,\cdot,\supp)$ consists of a Hilbert space $\mathcal H$, a collection of normalized \textit{configuration states} $\{|a\rangle\mid a\in A\}$ in $\cH$, and a collection of \textit{hopping operators} $\{U(s)\mid s\in S\}$, satisfying the following two axioms.
	\begin{itemize}
		\item \textbf{Configuration axiom:} for any $s\in S$ and $a\in A$,
		\begin{equation}
			U(s)|a\rangle=e^{i\theta(s,a)}|s\cdot a\rangle
		\end{equation}
		for some $\theta(s,a)\in\RR/2\pi\ZZ$.
		\item \textbf{Locality axiom:} For any $P$ with empty support,
		\begin{equation}
			U(P)|a\rangle=|a\rangle,\quad \forall a\in A.
		\end{equation}
	\end{itemize}
\end{definition}

Based on these two axioms, one can compute the solution space of $\{\theta(s,a)\}$ and define statistics $T^*$ as its discrete part. However, enumerating all processes with empty support is not easy, and more numerical and theoretical arguments are needed to justify the robustness of $T^*$. We leave these questions to future work.

\section{Details of the WZW boundary realization}
\label{app:wzw-realization-details}

In this appendix we give the details of the proof of
Theorem~\ref{thm:Realization}. We first collect relevant notation.

Let $G$ be a finite Abelian group and let $1\le q\le d$.  Let
\begin{equation}
	\omega\in Z^{d+2}\bigl(K(G,q),\mathbb R/\mathbb Z\bigr)
\end{equation}
be represented by a normalized local cochain operation.  We choose normalized
local descendants
\begin{equation}
	\Theta(s,f)\in C^{d+1}(\cdot,\mathbb R/\mathbb Z),
	\quad
	\beta(v)\in C^{d+1}(\cdot,\mathbb R/\mathbb Z),
\end{equation}
\begin{equation}
	L(s,v)\in C^d(\cdot,\mathbb R/\mathbb Z),
\end{equation}
where $f\in Z^q(\cdot,G)$ and $s,v\in C^{q-1}(\cdot,G)$, satisfying
\begin{align}
	\dd\Theta(s,f)
	&=\omega(f+\dd s)-\omega(f),
	\label{eq:app-Theta}\\
	\Theta(0,f)&=0,\nonumber\\
	\beta(v)&:=\Theta(v,0),
	\label{eq:app-beta}\\
	\dd L(s,v)
	&=\Theta(s,\dd v)+\beta(v)-\beta(s+v),
	\label{eq:app-L}\\
	L(0,v)&=0.\nonumber
\end{align}

Let $X$ be a PL triangulation of $S^d$, and let $M$ be a triangulated
$(d+1)$-disk with
\begin{equation}
	\partial M=X .
\end{equation}
The boundary Hilbert space is
\begin{equation}
	\mathcal H_X=\mathbb C\bigl[C^{q-1}(X,G)\bigr],
	\qquad
	\{|c\rangle:c\in C^{q-1}(X,G)\}
\end{equation}
with the standard orthonormal basis.  For a boundary configuration
$a\in B^q(X,G)$, choose a bulk coboundary extension
\begin{equation}
	\widetilde a\in B^q(M,G),
	\qquad
	\widetilde a|_X=a .
\end{equation}
The configuration state is represented by
\begin{equation}
		\begin{aligned}
			\confstate{a}_{\widetilde a}
			={}&
			\sum_{\substack{\widetilde{v}\in C^{q-1}(M,G)\\ \dd\widetilde{v}=\widetilde a}}
			e^{2\pi i\int_M\beta(\widetilde{v})}
			|v\rangle .
	\end{aligned}
	\label{eq:app-config-state}
\end{equation}
Here $v$ denotes the restriction of $\widetilde{v}$ to $X$.

For an operator label $s\in C^{q-1}(X,G)$, define
\begin{equation}
		U(s)|c\rangle
		=e^{-2\pi i\int_X L(s,c)}|c+s\rangle .
	\label{eq:app-operator}
\end{equation}
For an elementary label $(g,\alpha)\in G_0\times X^{q-1}$, $s$ is the
simplex-supported cochain with value $g$ on $\alpha$.
Theorem~\ref{thm:Realization} states that these data define a realization of
$m^q(X,G)$ and that the resulting statistics depends only on the cohomology
class $[\omega]$.

\subsection{Independence of bulk extensions}
\label{app:extension-independence}

We first show that the state in Eq.~\eqref{eq:app-config-state} is independent of $\widetilde{a}\in B^q(M,G)$. Consider an alternative choice
\begin{equation}
	\widetilde{a}'=\widetilde{a}+\dd\widetilde{\mu}.
\end{equation}
Because the configuration $a\in B^q(X,G)$ should be the same, the restriction of $\widetilde{\mu}$ to $X$ satisfies
\begin{equation}
	\mu\in Z^{q-1}(X,G).
\end{equation}
We compare $\confstate{a}_{\widetilde{a}}$ with 
\begin{equation}
		\begin{aligned}
		\confstate{a}_{\widetilde a+\dd\widetilde{\mu}}
		={}&
		\sum_{\substack{\widetilde{v'}\in C^{q-1}(M,G)\\ \dd\widetilde{v'}=\widetilde a+\dd\widetilde{\mu}}}
		e^{2\pi i\int_M\beta(\widetilde{v'})}
		|v'\rangle .
	\end{aligned}
\end{equation}

We fix $v=v'\in C^{q-1}(X,G)$ and then compare the two components. Let $\widetilde{\iota}=\widetilde{v'}-\widetilde{v}\in C^{q-1}(M,G)$. It should satisfy
\begin{equation}
	\left\{
		\begin{aligned}
			&\dd\widetilde{\iota}=\dd\widetilde{\mu}\\
			&\iota=0.
		\end{aligned}
		\right.
\end{equation}
These equations always admit solutions. Indeed, because $H^{q-1}(X,G)=0$, there exists $\mu=\dd\nu$ for some $\nu\in C^{q-2}(X,G)$. We extend it to $\widetilde{\nu}\in C^{q-2}(M,G)$ and take $\widetilde{\iota}=\widetilde{\mu}-\dd\widetilde{\nu}$. Note that this procedure may fail if $H^{q-1}(X,G)\ne0$; see Remark~\ref{remark: sphere is important}.

It remains to compare the amplitude
\begin{equation}
	\int_M \beta(\widetilde{v'})-\int_M\beta(\widetilde{v}).
\end{equation}
Because of the boundary condition $v'=v$, we may glue the two copies of $M$ to a sphere $\widehat{M}\simeq S^{d+1}$. Then we only need to evaluate
\begin{equation}
	\int_{\widehat M}\beta(\widehat v),
	\qquad
	\dd\widehat v=\widehat a,
	\label{eq:app-extension-ratio}
\end{equation}
where $\widehat v$ is obtained by gluing $\widetilde{v}$ and $\widetilde{v'}$.
It remains to see that this number depends only on $\widehat a$, not on the
choice of $\widehat v$.

Let $\widehat v_0$ and $\widehat v_1$ be two cochains on $\widehat M$ satisfying
$\dd\widehat v_0=\dd\widehat v_1=\widehat a$. Because $H^{q-1}(\widehat{M},G)=0$, we may assume $\widehat v_1=\widehat{v}_0+\dd\eta$. Consider the prism $\widehat M\times I$. Let $\pi:\widehat M\times I\to \widehat M$ be the projection and $i_*\eta$ be the embedding of $\eta$ to $\widehat{M}\times 1$, extended by zero. Then
\begin{align}
	\label{eq:app-extension-cylinder}
	&\int_{\widehat M}\beta(\widehat v_1)
	-\int_{\widehat M}\beta(\widehat v_0)
	\\
	&\hspace{1.0cm}=
	\int_{\widehat M\times I}\dd\beta(\pi^*\widehat{v}_0+\dd\iota_*\eta)
	=
	\int_{\widehat M\times I}\omega(\pi^*\widehat a)
	=0 .
\nonumber 
\end{align}
The last integral vanishes because $\omega(\pi^*\widehat a)$ is pulled back
from $\widehat M$ and has no component along the interval direction.  Therefore
Eq.~\eqref{eq:app-extension-ratio} is independent of the boundary basis vector
$c$. In conclusion, we have
\begin{equation}
	\confstate{a}_{\widetilde a_2}
	=
	\exp\left(2\pi i\int_{\widehat M}\beta(\widehat v)\right)
	\confstate{a}_{\widetilde a_1}
	\label{eq:app-extension-independence}
\end{equation}
for any $\widehat v$ with $\dd\widehat v=\widehat a$.  Thus changing the bulk
extension changes only the phase convention of the configuration state.

A similar argument also applies to two different bulk triangulations $M_1$ and $M_2$. For $a\in B^q(X,G)$, we take $a=\dd s$ and use $\widetilde{s_i}$ to denote the extension to $C^{q-1}(M_i,G)$. Then one can prove Eq.~\eqref{eq:app-extension-independence}
by the same method.

\subsection{The configuration axiom}
\label{app:configuration-axiom-proof}

Let $s\in C^{q-1}(X,G)$ and choose an extension
\begin{equation}
	\widetilde s\in C^{q-1}(M,G),
	\qquad
	\widetilde s|_X=s .
\end{equation}
Applying Eq.~\eqref{eq:app-operator} to Eq.~\eqref{eq:app-config-state} gives
\begin{equation}
	\begin{aligned}
		U(s)\confstate{a}_{\widetilde a}
		={}&
		\sum_{\dd\widetilde v=\widetilde a}
		e^{2\pi i\int_M\beta(\widetilde{v})}
		e^{-2\pi i\int_X L(s,v)}
		|v+s\rangle .
	\end{aligned}
	\label{eq:app-config-axiom-1}
\end{equation}
By Stokes' theorem and Eq.~\eqref{eq:app-L},
\begin{align}
	\int_X L(s,v)
	&=\int_M \dd L(\widetilde s,\widetilde{v})
	\nonumber\\
	&=\int_M\Theta(\widetilde s,\dd\widetilde{v})
	+\int_M\beta(\widetilde{v})
	-\int_M\beta(\widetilde{v}+\widetilde s).
	\label{eq:app-stokes-L}
\end{align}
Since $\dd\widetilde{v}=\widetilde a$, the summand in
Eq.~\eqref{eq:app-config-axiom-1} becomes
\begin{equation}
	e^{-2\pi i\int_M\Theta(\widetilde s,\widetilde a)}
	e^{2\pi i\int_M\beta(\widetilde{v}+\widetilde s)}
	|v+s\rangle .
\end{equation}
Changing variables from $\widetilde{v}$ to $\widetilde{v}+\widetilde s$, we obtain
\begin{equation}
		\begin{aligned}
			U(s)\confstate{a}_{\widetilde a}
			={}&
			e^{-2\pi i\int_M\Theta(\widetilde s,\widetilde a)}
			\confstate{a+\dd s}_{\widetilde a+\dd\widetilde s} .
	\end{aligned}
	\label{eq:app-config-axiom-result}
\end{equation}
This is precisely the configuration axiom. The formula is independent of the
particular extension $\widetilde s$; changing $\widetilde s$ changes the displayed
coefficient $e^{-2\pi i\int_M\Theta(\widetilde s,\widetilde a)}$ and the final state $\confstate{a+\dd s}_{\widetilde a+\dd\widetilde s}$ by compensating phases.

\subsection{The locality axiom}
\label{app:locality-axiom-proof}

Write
\begin{equation}
	U(s)=M_sT_s,
	\qquad
	T_s|c\rangle=|c+s\rangle,
\end{equation}
where $M_s$ is the diagonal phase in Eq.~\eqref{eq:app-operator}.  The
translations $T_s$ commute.  Hence all commutators come from finite differences
of the local phase $L(s,c)$.

For example, with the convention
$[A,B]=A^{-1}B^{-1}AB$, one obtains
\begin{equation}
	[U(s_2),U(s_1)]|c\rangle
	=
	e^{-2\pi i\int_X\Gamma_2(s_1,s_2;c)}|c\rangle,
\end{equation}
where
\begin{equation}
	\begin{aligned}
		\Gamma_2(s_1,s_2;c)
		={}&L(s_1,c)+L(s_2,c+s_1)\\
		&-L(s_2,c)-L(s_1,c+s_2).
	\end{aligned}
	\label{eq:app-Gamma2}
\end{equation}
Thus $\Gamma_2$ is a second finite difference of the local cochain operation
$L$.  It is supported only where the dual supports of $s_1$ and $s_2$ meet.
Repeating the same argument, the phase of a $k$-fold nested commutator is the
integral of a local $k$-fold finite difference of $L$, supported inside
\begin{equation}
	\supp(s_1)\cap\cdots\cap\supp(s_k).
\end{equation}
If this intersection is empty, the finite difference vanishes, and the nested
commutator is the identity.  This proves the locality axiom.

\subsection{Changing cocycle representatives and descendants}
\label{app:cohomology-representative-independence}

The construction uses a cocycle representative $\omega$ together with a
normalized full descent datum $(\Theta,\beta,L)$.  We now show that the
statistics depend only on the cohomology class
\begin{equation}
	[\omega]\in H^{d+2}\bigl(K(G,q),\mathbb R/\mathbb Z\bigr).
\end{equation}
Let $(\omega',\Theta',\beta',L')$ be another normalized descent datum with
\begin{equation}
	\omega'=\omega+\dd\mu,
	\qquad
	\mu\in C^{d+1}\bigl(K(G,q),\mathbb R/\mathbb Z\bigr),
	\label{eq:app-omega-prime}
\end{equation}
where $\mu$ is normalized.  Then
\begin{equation}
	\dd\Bigl[
	\Theta'(s,f)-\Theta(s,f)
	+\mu(f)-\mu(f+\dd s)
	\Bigr]=0 .
\end{equation}
The cocycle inside brackets is zero when $s=0$. This further implies that it is not only a cocycle but also a coboundary. Hence there is a natural local $d$-cochain operation
$\rho(s,f)$, with $\rho(0,f)=0$, such that
\begin{equation}
		\Theta'(s,f)=
		\Theta(s,f)+\mu(f+\dd s)-\mu(f)+\dd\rho(s,f).
	\label{eq:app-Theta-prime}
\end{equation}
Taking $f=0$ gives
\begin{equation}
	\beta'(v)=
	\beta(v)+\mu(\dd v)+\dd\rho(v,0).
	\label{eq:app-beta-prime}
\end{equation}
Consequently, $L'$ may be chosen in the form
\begin{equation}
		\begin{aligned}
			L'(s,v)=L(s,v)
			&+\rho(s,\dd v)+\rho(v,0)\\
			&-\rho(s+v,0)+\dd\chi(s,v).
	\end{aligned}
	\label{eq:app-L-prime}
\end{equation}
The total differential term $\dd\chi$ is irrelevant in this paper. Indeed, applying $\dd$ to the
right-hand side gives
\begin{equation}
	\dd L'(s,v)=\Theta'(s,\dd v)+\beta'(v)-\beta'(s+v).
\end{equation}

Define the diagonal unitary
\begin{equation}
	W_\rho|c\rangle
	=e^{2\pi i\int_X\rho(c,0)}|c\rangle .
	\label{eq:app-W-rho}
\end{equation}
For fixed $\widetilde a$, Eq.~\eqref{eq:app-beta-prime} implies
\begin{equation}
		{\confstate{a}_{\widetilde a}}'
		=
		e^{2\pi i\int_M\mu(\widetilde a)}
		W_\rho\confstate{a}_{\widetilde a} .
	\label{eq:app-state-counterterm}
\end{equation}
Thus changing $\omega$ by a coboundary conjugates the family of states by a
local diagonal unitary and rephases each configuration state.

Similarly, using Eq.~\eqref{eq:app-L-prime} on a closed $X$, the new excitation
operator satisfies
\begin{equation}
		U'(s)=W_\rho U(s)W_\rho^{-1}M_\rho(s),
	\label{eq:app-operator-counterterm}
\end{equation}
where
\begin{equation}
	M_\rho(s)|c\rangle
	=e^{-2\pi i\int_X\rho(s,\dd c)}|c\rangle .
	\label{eq:app-local-M-rho}
\end{equation}
Here the order in Eq.~\eqref{eq:app-operator-counterterm} is the usual operator
order: $M_\rho(s)$ acts first on the input basis vector $|c\rangle$.

We now compare the corresponding statistical phases.  Let
\begin{equation}
	P=s_k^{\epsilon_k}\cdots s_1^{\epsilon_1},
	\qquad
	\epsilon_i=\pm1,
\end{equation}
be a statistical process.  Its evaluated phase is of the form
\begin{equation}
	\langle a_0|
	U(s_k)^{\epsilon_k}\cdots U(s_1)^{\epsilon_1}
	|a_0\rangle\in U(1).
	\label{eq:app-statistical-phase}
\end{equation}
Equation~\eqref{eq:app-state-counterterm} only changes the phase convention of
the configuration states, while $W_\rho$ is an automorphism of the whole
boundary Hilbert space.  These changes do not affect the statistical phase in Eq.~\eqref{eq:app-statistical-phase}.
The remaining factor $M_\rho(s)$ is a product of local diagonal phase factors,
supported near the elementary operator $s$ and depending only on the local
configuration $\dd c$.  Such local modifications are precisely the continuous redundancies (Theorem VI.9, Ref.~\cite{Xue2026StatisticsAbelian})
quotiented out in the definition of statistical processes and of
$T^*(m^q(X,G))$.  Therefore, cohomologous cocycles and their normalized descent
data define the same statistics.

We have thus proved that the WZW construction gives a well-defined map
\begin{equation}
	H^{d+2}\bigl(K(G,q),\mathbb R/\mathbb Z\bigr)
	\longrightarrow T^*\bigl(m^q(X,G)\bigr).
	\label{eq:app-well-defined-map}
\end{equation}
To prove injectivity, we compare the above realization with the simplest
spatial sphere $X=\partial\Delta^{d+1}$, where Ref.~\cite{Xue2026StatisticsAbelian}
proves
\begin{equation}
	H^{d+2}\bigl(K(G,q),\mathbb R/\mathbb Z\bigr)
	\simeq
	T^*\bigl(m^q(\partial\Delta^{d+1},G)\bigr).
\end{equation}

\subsection{The special case $X=\partial\Delta^{d+1}$}
\label{app:simplex-boundary-realization}

Now take
\begin{equation}
	X=\partial\Delta^{d+1},
	\qquad
	M=\Delta^{d+1}.
\end{equation}
This case is special because $C^{n}(\partial\Delta^{d+1},G)\simeq C^{n}(\Delta^{d+1},G)$. Therefore, when one extends a boundary cochain such as $s\in C^{q-1}(\partial\Delta^{d+1},G)$ to the bulk, no additional choice is required.

The general configuration-state formula reduces to
\begin{equation}
		\begin{aligned}
			\confstate{a}
			={}&
			\sum_{\substack{v\in C^{q-1}(\Delta^{d+1},G)\\ \dd v=a}}
			e^{2\pi i\int_{\Delta^{d+1}}\beta(v)}
			|v\rangle .
	\end{aligned}
	\label{eq:app-simplex-state}
\end{equation}
The hopping operator is
\begin{equation}
		U(s)|c\rangle
		=
		e^{-2\pi i\int_{\partial\Delta^{d+1}}L(s,c)}
		|c+s\rangle .
	\label{eq:app-simplex-operator}
\end{equation}
These formulas are intrinsic to the boundary of the simplex.  The only use of
$\Delta^{d+1}$ is that the local $(d+1)$-cochain $\beta(v)$ is evaluated on the
single top simplex; all required cochain values already live on the boundary.

Using Stokes' theorem and Eq.~\eqref{eq:app-L}, we obtain
\begin{equation}
		U(s)\confstate{a}
		=
		e^{-2\pi i\int_{\Delta^{d+1}}\Theta(s,a)}
	\confstate{a+\dd s} .
	\label{eq:app-simplex-config-axiom}
\end{equation}
Thus, on $\partial\Delta^{d+1}$, the present boundary realization agrees with
the realization used in Ref.~\cite{Xue2026StatisticsAbelian}, up to the overall
sign convention for $L$.  In particular, the statistical processes constructed
there evaluate to the same WZW phases for the present realization.  This identifies
our construction on the simplest sphere with the known realization whose
statistics are classified by
$H^{d+2}(K(G,q),\mathbb R/\mathbb Z)$.

\subsection{Transfer of statistical processes}
\label{app:transfer-statistical-processes}

In the previous subsections we have constructed the map in Eq.~\eqref{eq:app-well-defined-map}. Now we prove that this map is injective, \ie, different cohomology classes give different statistics. Because the map is evidently an Abelian-group homomorphism, it is enough to prove that for any nontrivial cohomology class $[\omega] \in H^{d+2}(K(G,q),\RR/\ZZ)$, the corresponding statistics are nontrivial. The
argument is a transfer argument from the simplest sphere
\begin{equation}
	\partial\Delta^{d+1}
\end{equation}
to an arbitrary PL triangulated sphere $X$.

We first recall how a statistical process is evaluated in the present
realization.  Roughly speaking, a statistical process is a word
\begin{equation}
	P=s_N^{\epsilon_N}\cdots s_2^{\epsilon_2}s_1^{\epsilon_1},
	\qquad
	\epsilon_i=\pm1,
	\label{eq:app-process-word}
\end{equation}
whose evaluated phase Eq.~\eqref{eq:app-statistical-phase} is independent of the choice $\{U(s)\}$. Because of the initial-state-independence theorem, this phase is equal to
\begin{equation}\label{eq: statistical phase to be evaluated}
	\langle 0|
	U(s_k)^{\epsilon_k}\cdots U(s_1)^{\epsilon_1}
	|0\rangle\in U(1),
\end{equation}
where $|0\rangle$ is not a configuration state but the basis vector labeled by $0\in C^{q-1}(X,G)$. We define a sequence of cochains
\begin{equation}
	v_0=0,
	\qquad
	v_i=\sum_{j<i}\epsilon_j s_j .
	\label{eq:app-process-vi}
\end{equation}
Expanding the product in Eq.~\eqref{eq: statistical phase to be evaluated}, one obtains
\begin{equation}
	\Phi_\omega(P;X)
	=
	\exp\left(
	-2\pi i\sum_{i=1}^N\int_X \ell_i
	\right),
	\label{eq:app-process-phase}
\end{equation}
where
\begin{equation}
	\ell_i=
	\begin{cases}
		L(s_i,v_i),& \epsilon_i=+1,\\[2mm]
		-\,L(s_i,v_i-s_i),& \epsilon_i=-1.
	\end{cases}
	\label{eq:app-ell-i}
\end{equation}
Here $v_i$ is always the cochain before the $i$-th operator acts, except that
for an inverse operator we rewrite $U(s_i)^{-1}$ as the inverse move from
$v_i$ to $v_i-s_i$.

For the simplest triangulated sphere $\partial\Delta^{d+1}$,
Ref.~\cite{Xue2026StatisticsAbelian} proves the existence of statistical processes that distinguish all classes in
$H^{d+2}(K(G,q),\mathbb R/\mathbb Z)$.  Equivalently, for every nonzero
$[\omega]$ there is a process $P_0$ such that
\begin{equation}
	\Phi_\omega(P_0;\partial\Delta^{d+1})\ne 1.
\end{equation}
We write $\Delta^{d+1}$ as $[0,\cdots,d+1]$ and let $\sigma_0=[0,\cdots,d]$ be a specific $d$-simplex of $\partial\Delta^{d+1}$. Moreover, $P_0$ can be constructed in a particular form according to the following theorem, the cochain version of Theorem VI.12, Ref.~\cite{Xue2026StatisticsAbelian}.
 \begin{theorem}
 	Let $\sigma_0$ be a specific $d$-simplex of $\partial\Delta^{d+1}$. Then
 	any element in $T(m^q(\partial\Delta^{d+1},G))\simeq H_{d+2}(K(G,q),\ZZ)$ can be represented by a statistical process $P_0$ such that for each $(g,\alpha)\in S$ appearing in $P_0$, the relation $\alpha\subset \sigma_0$ holds.
 \end{theorem}

  Next, for any generic triangulated $d$-sphere $X$, we identify $\sigma_0$ with a specific $d$-simplex $\sigma\subset X$. This construction then transfers the statistical process $P_0$ of $m^q(\partial\Delta^{d+1},G)$ to a statistical process $P_\sigma$ of $m^q(X,G)$.
  
  We define a map 
  \begin{equation}
  	r:X\longrightarrow \partial\Delta^{d+1}
  	\label{eq:app-transfer-map}
  \end{equation}
  that maps $\sigma$ isomorphically to $\sigma_0$ and collapses all remaining vertices of $X$ to the vertex $(d+1)$ of $\Delta^{d+1}$. After assigning a suitable branching structure to $X$, we can make $r$ a map of simplicial sets.
  
  Because $r$ induces an isomorphism $\sigma\simeq \sigma_0$, the transferred process $P_\sigma$ is obtained by replacing each label
  $s_i$ in $P_0$ by
  \begin{equation}
  	s_i^\sigma:=r^*s_i.
  	\label{eq:app-transferred-label}
  \end{equation}
  Similarly, in Eq.~\eqref{eq:app-process-phase}, we have $v_i^\sigma:=r^*v_i$. Because $L$ is a natural transformation, we have
  \begin{equation}
  	\ell_i^\sigma=r_\sigma^*\ell_i.
  	\label{eq:app-ell-transfer}
  \end{equation}
 Consequently,
  \begin{equation}
  	\begin{aligned}
  		&\sum_i \int_X \ell_i^\sigma
  		=\sum_i \int_X r^*\ell_i\\
  		&=\sum_i \int_{r_*X} \ell_i=\sum_i \int_{\partial\Delta^{d+1}}\ell_i
  	\end{aligned}
  \end{equation}
  and
\begin{equation}
	\Phi_\omega(P_\sigma;X)
	=
	\Phi_\omega(P_0;\partial\Delta^{d+1})\ne 1.
	\label{eq:app-transfer-phase-equality}
\end{equation}

  This proves both the injectivity of the WZW-boundary realization and the nontriviality of the transferred statistical process $P_\sigma$.

\section{Proof of Theorem \ref{thm:Sgamma}}\label{appendix: proof of anomalous symmetry}

\begin{lemma}\label{lemma: elimination}
	Let $X$ be a closed $d$-dimensional combinatorial manifold. Let
	\begin{equation}
		\mathcal O(c_1,\ldots,c_m;\gamma_1,\ldots,\gamma_n)
		\in C^d(X,\RZ)
	\end{equation}
	be a natural local cochain operation, where the $c_i$ are arbitrary cochains
	and the $\gamma_j$ are cocycles. We abbreviate this expression as
	$\mathcal O(c_i;\gamma_j)$. Suppose that, as a universal cochain identity,
	\begin{equation}\label{eq:elim-assumption-prb}
		d\mathcal O(c_i;\gamma_j)=\mathcal Q(\gamma_j),
	\end{equation}
	with the right-hand side independent of all $c_i$ and $dc_i$. Then
	\begin{equation}\label{eq:elim-conclusion-prb}
		\int_X\mathcal O(c_i;\gamma_j)=\int_X\mathcal O(0;\gamma_j).
	\end{equation}
\end{lemma}
\begin{proof}
	Indeed, take $Y=X\times I$ and extend $c_i$ to cochains $\widetilde c_i$ with
	boundary values $0$ and $c_i$. Extend each closed variable constantly,
	$\widetilde\gamma_j=\pi^*\gamma_j$, where $\pi:Y\to X$. Stokes' theorem gives
	\begin{align}
		&\int_X \mathcal O(c_i;\gamma_j)
		-\int_X \mathcal O(0;\gamma_j)  \notag\\
		&=\int_Y d\mathcal O(\widetilde c_i;\widetilde\gamma_j)
		=\int_Y \mathcal Q(\pi^*\gamma_j)=0 .
	\end{align}
	The last equality follows because $\mathcal Q(\pi^*\gamma_j)$ is pulled back
	from $X$ and has no component in the interval direction. The hypothesis in
	Eq.~\eqref{eq:elim-assumption-prb} is crucial: dependence on $dc_i$ is not
	allowed, even though $dc_i$ is itself closed.
	
\end{proof}

Now we prove Theorem~\ref{thm:Sgamma}. For $\gamma,\gamma'\in Z^{q-1}(X,G)$, it claims that
 \begin{itemize} \item $\mathcal{S}(\gamma)\mathcal{S}(\gamma')\propto \mathcal{S}(\gamma+\gamma')$ \item $\mathcal{S}(\gamma)$ commute with $\mathcal{S}(\gamma')$ if either $\gamma$ or $\gamma'$ is a coboundary. \item For any $s\in C^{q-1}(X,G)$, the hopping operator $U(s)$ commute with $\mathcal{S}(\gamma)$. \item If $X$ is a combinatorial sphere, then there exists $\phi(\gamma)\in\RZ$ such that \begin{equation} \mathcal{S}(\gamma)\confstate{a}=e^{2\pi \ii\phi(\gamma)}\confstate{a},\quad \forall a\in B^q(X,G). \end{equation}
  \end{itemize}

\begin{proof}
	We use
	\begin{equation}\label{eq:L-closed-prb}
		dL(v,\gamma)=\beta(v)+\beta(\gamma)-\beta(v+\gamma),
		\qquad d\gamma=0,
	\end{equation}
	and the normalization $L(0,v)=L(v,0)=0$.
	
	First,
	\begin{align}
		\mathcal S(\gamma)\mathcal S(\gamma')|v\rangle
		&=\exp\!\left[-2\pi\ii\int_X
		B_{\gamma,\gamma'}(v)\right]
		\mathcal S(\gamma+\gamma')|v\rangle,
		\label{eq:S-product-start-prb}
	\end{align}
	where
	\begin{equation}
		B_{\gamma,\gamma'}(v)
		=L(v,\gamma')+L(v+\gamma',\gamma)-L(v,\gamma+\gamma').
	\end{equation}
	By Eq.~\eqref{eq:L-closed-prb},
	\begin{equation}
		dB_{\gamma,\gamma'}(v)
		=\beta(\gamma')+\beta(\gamma)-\beta(\gamma+\gamma'),
	\end{equation}
	which is independent of the arbitrary cochain $v$. Lemma \ref{lemma: elimination} gives
	\begin{equation}
		\int_X B_{\gamma,\gamma'}(v)
		=\int_X B_{\gamma,\gamma'}(0)
		=\int_X L(\gamma',\gamma).
	\end{equation}
	Therefore
	\begin{equation}\label{eq:S-product-prb}
		\mathcal S(\gamma)\mathcal S(\gamma')
		=e^{-2\pi\ii\int_X L(\gamma',\gamma)}
		\mathcal S(\gamma+\gamma').
	\end{equation}
	This proves the projective product law. It also shows that the commutator
	phase is
	\begin{equation}\label{eq:comm-phase-prb}
		\exp\!\left[-2\pi\ii\int_X
		\{L(\gamma',\gamma)-L(\gamma,\gamma')\}\right].
	\end{equation}
	Assume now that $\gamma=db$; for $q=1$ this means $\gamma=0$. For $q>1$, define
	\begin{equation}
		C(b;\gamma')=L(\gamma',db)-L(db,\gamma').
	\end{equation}
	Equation~\eqref{eq:L-closed-prb} gives $dC(b;\gamma')=0$. Applying Lemma \ref{lemma: elimination} to the cochain variable $b$,
	\begin{equation}
		\int_X C(b;\gamma')=\int_X C(0;\gamma')
		=\int_X\{L(\gamma',0)-L(0,\gamma')\}=0 .
	\end{equation}
	Thus Eq.~\eqref{eq:comm-phase-prb} is unity. The case $\gamma'=db$ is identical.
	
	Next compare $\mathcal S(\gamma)U(s)|v\rangle$ and
	$U(s)\mathcal S(\gamma)|v\rangle$. Their relative exponent is the integral of
	\begin{equation}
		A(s,v;\gamma)=L(s,v)+L(v+s,\gamma)
		-L(v,\gamma)-L(s,v+\gamma).
	\end{equation}
	Using Eq.~\eqref{eq:L-closed-prb} and the defining equation for $L$, one obtains
	$dA(s,v;\gamma)=0$. Lemma \ref{lemma: elimination}, applied to $s$ and $v$, gives
	\begin{equation}
		\int_X A(s,v;\gamma)=\int_X A(0,0;\gamma)=0.
	\end{equation}
	Hence $U(s)\mathcal S(\gamma)=\mathcal S(\gamma)U(s)$.
	
	It remains to prove the statement about configuration states. Let
	$M$ be a combinatorial $(d+1)$-disk with $\partial M=X$, and choose
	$\widetilde a\in B^q(M,G)$ extending $a$. The configuration state is
	\begin{equation}
		\confstate{a}_{\widetilde a}
		=\sum_{d\widetilde v=\widetilde a}
		e^{2\pi\ii\int_M\beta(\widetilde v)}|v\rangle .
	\end{equation}
	Since $X\simeq S^d$, every $\gamma\in Z^{q-1}(X,G)$ admits a closed extension
	$\widetilde\gamma\in Z^{q-1}(M,G)$: for $q=1$ extend the constant $G$-valued
	$0$-cocycle $\gamma$ constantly, and for $q>1$ use $H^{q-1}(S^d,G)=0$ to write
	$\gamma=d\eta$ and take $\widetilde\gamma=d\widetilde\eta$. Then
	\begin{align}
		\mathcal S(\gamma)\confstate{a}_{\widetilde a}
		&=\sum_{d\widetilde v=\widetilde a}
		e^{2\pi\ii\int_M\beta(\widetilde v)-2\pi\ii\int_X L(v,\gamma)}
		|v+\gamma\rangle                                    \notag\\
		&=e^{-2\pi\ii\int_M\beta(\widetilde\gamma)}
		\sum_{d\widetilde v=\widetilde a}
		e^{2\pi\ii\int_M\beta(\widetilde v+\widetilde\gamma)}
		|v+\gamma\rangle .
	\end{align}
	In the second line we used Stokes' theorem and
	$dL(\widetilde v,\widetilde\gamma)=\beta(\widetilde v)+\beta(\widetilde\gamma)
	-\beta(\widetilde v+\widetilde\gamma)$. Since $d\widetilde\gamma=0$, the change
	of variables $\widetilde v' = \widetilde v+\widetilde\gamma$ preserves
	$d\widetilde v'=\widetilde a$. Hence
	\begin{equation}
		\mathcal S(\gamma)\confstate{a}_{\widetilde a}
		=e^{-2\pi\ii\int_M\beta(\widetilde\gamma)}
		\confstate{a}_{\widetilde a}.
	\end{equation}
	Thus one may take $\phi(\gamma)=-\int_M\beta(\widetilde\gamma)$, which is
	independent of $a$.
\end{proof}

\section{Spacetime complexes, cochains, and cocycles}
\label{cochain}

\subsection{Cochains and cocycles}

In this paper, we model the \(D\)-dimensional spacetime \(\cM^{D}\) using a
triangulated spacetime complex, denoted by \(M^{D}\), with $D=d+1$. The complex is constructed
from simplices: vertices, edges, triangles, tetrahedra, and so on. We label
vertices with indices \(i,j,k,\ldots\), edges by pairs \((i,j)\), and triangles
by triples \((i,j,k)\). 

To unambiguously define local lattice actions on these simplices, we endow the
complex with a \emph{branching structure} \cite{C0527,CGL1172,CGL1204}. A
branching structure assigns a direction to every edge such that no triangle
contains an oriented closed loop (see Fig.~\ref{mir}). This choice naturally
induces a local ordering of the vertices within each simplex: the first vertex
has no incoming edges, the second has exactly one incoming edge, and so forth.
Furthermore, this vertex ordering assigns a canonical orientation to each
simplex. 

\begin{figure}[t]
\begin{center}
\includegraphics[scale=0.5]{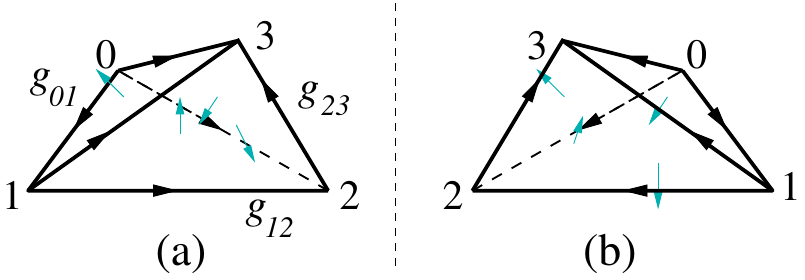}
\end{center}
\caption{
(Color online) Two branched simplices with opposite orientations.
(a) A branched simplex with a positive orientation and
(b) a branched simplex with a negative orientation.
}
\label{mir}
\end{figure}

Let \(\M\) be an Abelian group (equivalently, a \(\Z\)-module). An \(n\)-cochain \(f_n \in C^n(M^{D};\M)\) assigns an element of \(\M\) to each oriented \(n\)-simplex, denoted by \(f_{n;i_0i_1\cdots i_n}\). Physically, such a cochain represents a bosonic lattice field on the \(n\)-simplices. 

Because an oriented simplex changes sign when its orientation is reversed -- \ie, \((i_1,i_0,\ldots,i_n)=-(i_0,i_1,\ldots,i_n)\) -- an additive cochain is completely antisymmetric under permutations of its vertices:
\[
f_{n;i_1i_0\cdots i_n}=-f_{n;i_0i_1\cdots i_n}.
\]
Equivalently, an \(n\)-cochain is a linear map from the space of \(n\)-chains to \(\M\), with the pairing evaluated as:
\[
\langle f_n,(i_0,i_1,\ldots,i_n)\rangle = f_{n;i_0i_1\cdots i_n}.
\]
(If the coefficient group is a direct sum, \(\M_1\oplus \M_2\), the cochain is simply treated as a corresponding pair of cochains.)

Integration over the spacetime manifold is defined by evaluating a \(D\)-cochain \(f_D\) on the fundamental \(D\)-chain \(M^{D}\):
\begin{align}
\int_{M^{D}} f_D \equiv \langle f_D,M^{D}\rangle 
= \sum_{(i_0,i_1,\ldots,i_D)} s_{i_0i_1\cdots i_D} (f_D)_{i_0i_1\cdots i_D},
\end{align}
where \(s_{i_0i_1\cdots i_D}=\pm1\) indicates whether the simplex's canonical orientation aligns with the global orientation of \(M^{D}\).

The coboundary operator \(\dd:C^n(M^{D};\M)\to C^{n+1}(M^{D};\M)\) is defined by the standard alternating sum:
\begin{align}
\label{eq:differential}
&\ \ \ \
\left\langle \dd f_n,(i_0i_1\cdots i_{n+1})\right\rangle 
\nonumber\\
&= \sum_{m=0}^{n+1} (-1)^m \left\langle f_n,(i_0i_1\cdots \widehat{i_m}\cdots i_{n+1}) \right\rangle,
\end{align}
where \(\widehat{i_m}\) denotes the omission of the \(m\)-th vertex. This operator is nilpotent, satisfying \(\dd^2=0\). 

We denote the group of \(n\)-cocycles (cochains satisfying \(\dd f_n=0\)) by \(Z^n(M^{D};\M)\). A cochain \(f_n\) is a coboundary if it can be written as \(f_n=\dd g_{n-1}\) for some \((n-1)\)-cochain \(g_{n-1}\); the group of \(n\)-coboundaries is denoted by \(B^n(M^{D};\M)\). Since \(\dd^2=0\), every coboundary is trivially a cocycle (\(B^n \subset Z^n\)). The \(n\)-th cohomology group is defined as the quotient:
\[
H^n(M^{D};\M) = Z^n(M^{D};\M)/B^n(M^{D};\M).
\]

For a $\Z_N$-valued cocycle $x_n$, we may choose an integer lift, still
denoted by $x_n$. Because $x_n$ is closed modulo $N$, its integer coboundary
satisfies $\dd x_n \se{N} 0$. This allows us to define the integer-valued
$(n+1)$-cocycle known as the Bockstein homomorphism:
\begin{align}
\label{Bsdef}
\Bs_N x_n \equiv \frac{1}{N}\dd (x_n-n\toZ{x_n/n})=\frac{1}{N}\dd \t x_n
 ,
\end{align}
where $\t x_n = x_n-n\toZ{x_n/n}$ is a fixed integer lift of $\Z_n$-valued
$x_n$ The cohomology class of \(\Bs_N x_n\) is independent of the chosen
integer lift, as altering the lift simply shifts \(\Bs_N x_n\) by an integer
coboundary.

While the additive cochain formalism applies to Abelian groups, we will occasionally require non-Abelian \(1\)-cocycles. For a general group \(G\), a \(G\)-valued \(1\)-cochain \(g\) satisfies the flatness condition (or \(1\)-cocycle condition) on a triangle if:
\[
g_{ik}= g_{ij} g_{jk}.
\]
We do not require higher non-Abelian cohomology in this paper.

\subsection{Cup products and higher cup products}

For \(f_m\in C^m(M;\M_1)\) and \(h_n\in C^n(M;\M_2)\),
together with a bilinear pairing \(\M_1\times \M_2\to \M_3\), the
Alexander-Whitney cup product is
\begin{align}
&\ \ \ \
\langle f_m\smile h_n,(0,1,\ldots,m+n)\rangle
\nonumber\\
&=
\langle f_m,(0,1,\ldots,m)\rangle
\langle h_n,(m,m+1,\ldots,m+n)\rangle .
\end{align}
We often write \(f_mh_n\) for \(f_m\smile h_n\).

Let us check the coboundary sign explicitly. On the simplex
\((0,1,\ldots,m+n+1)\), one has
\begin{align}
&\dd(f_mh_n)_{(0,\ldots,m+n+1)}  \\
&=
\sum_{r=0}^{m}(-)^r
f_{m;(0,\ldots,\widehat r,\ldots,m+1)}\,
h_{n;(m+1,\ldots,m+n+1)}
\nonumber\\
&\ \ \ \ +
\sum_{r=m+1}^{m+n+1}(-1)^r
f_{m;(0,\ldots,m)}\,
h_{n;(m,\ldots,\widehat r,\ldots,m+n+1)}.
\nonumber 
\end{align}
On the other hand,
\begin{align}
&\ \ \ \ \
\big((\dd f_m)h_n\big)_{(0,\ldots,m+n+1}
\nonumber\\
&=
\sum_{r=0}^{m+1}(-1)^r
f_{m;(0,\ldots,\widehat r,\ldots,m+1)}
h_{n;(m+1,\ldots,m+n+1)},
\nonumber \\
&\ \ \ \
(-1)^m \big(f_m\dd h_n\big)_{(0,\ldots,m+n+1}
\nonumber\\
&=
\sum_{s=0}^{n+1}(-1)^{m+s}
f_{m;(0,\ldots,m)}
h_{n;(m,\ldots,\widehat{m+s},\ldots,m+n+1)}.
\end{align}
The two middle terms
\[
(-1)^{m+1}f_m(0,\ldots,m)h_n(m+1,\ldots,m+n+1)
\]
and
\[
(-1)^m f_m(0,\ldots,m)h_n(m+1,\ldots,m+n+1)
\]
cancel. The remaining terms are exactly those in \(\dd(f_mh_n)\).
Hence
\begin{align}
\label{cup-Leibniz-expanded}
\dd(f_mh_n)
=
(\dd f_m)h_n+(-1)^m f_m\dd h_n .
\end{align}

We next use Steenrod's higher cup products. For
\(f_m\in C^m(M;\M_1)\) and \(h_n\in C^n(M;\M_2)\), the operation
\(f_m\hcup{k}h_n\) is an \((m+n-k)\)-cochain. We set
\(f_m\hcup{k}h_n=0\) for \(k<0\) or \(k>\min(m,n)\).

A convenient way to write the simplicial formula is the following. Set
\(q=m+n-k\). For a strictly increasing sequence
\[
0\le i_0<i_1<\cdots<i_k\le q ,
\]
form two alternating vertex lists
\[
F(i_\bullet)=
(0,\ldots,i_0,\ i_1,\ldots,i_2,\ i_3,\ldots,i_4,\ldots),
\]
\[
H(i_\bullet)=
(i_0,\ldots,i_1,\ i_2,\ldots,i_3,\ i_4,\ldots,i_5,\ldots).
\]
The final interval \((i_k,\ldots,q)\) is placed in \(F\) if \(k\) is odd
and in \(H\) if \(k\) is even. Only those \(i_\bullet\) for which
\(F(i_\bullet)\) is an \(m\)-simplex and \(H(i_\bullet)\) is an
\(n\)-simplex are retained. Then \cite{S4790} (Section 2)
\begin{align}
\label{hcupdef-expanded}
&\ \ \ \
\langle f_m\hcup{k}h_n,(0,1,\ldots,q)\rangle
\nonumber\\
& = \sum_{i_\bullet}
(-1)^{\epsilon(i_\bullet)}
\langle f_m,F(i_\bullet)\rangle
\langle h_n,H(i_\bullet)\rangle ,
\end{align}
where \((-1)^{\epsilon(i_\bullet)}\) is the orientation sign of the
corresponding Steenrod shuffle. 
It is the sign of permutation
to connect two ordered array's 
\begin{align}
& (0,\ldots,i_0,\ i_1,\ldots,i_2,\ i_3,\ldots,i_4,\ldots)
\nonumber\\
& (i_0,\ldots,i_1,\ i_2,\ldots,i_3,\ i_4,\ldots,i_5,\ldots)
\nonumber\\
\to &
(0,\ldots,i_0,\ i_0,\ldots,i_1,\ i_1,\ldots,i_2,\ldots)
\end{align}
This sign convention is chosen so that
the following coboundary identity holds \cite{S4790} (Theorem 5.1):
\begin{align}
\label{cupkrel}
&
\dd(f_m\hcup{k}h_n)
=
\dd f_m\hcup{k}h_n
+
(-1)^m f_m\hcup{k}\dd h_n
\\
&\quad+
(-1)^{m+n-k}f_m\hcup{k-1}h_n
+
(-1)^{mn+m+n}h_n\hcup{k-1}f_m .
\nonumber 
\end{align}

As checks:
for \(k=0\), the terms with \(\hcup{-1}\) vanish, and
\eqref{cupkrel} reduces to the ordinary Leibniz rule
\eqref{cup-Leibniz-expanded}. For \(k=1\), if \(f_m\) and \(h_n\) are
cocycles, then
\[
\dd(f_m\hcup{1}h_n)
=
(-1)^{m+n-1}f_mh_n
+
(-1)^{mn+m+n}h_nf_m .
\]
Equivalently,
\[
f_mh_n-(-1)^{mn}h_nf_m
=
(-1)^{m+n-1}\dd(f_m\hcup{1}h_n),
\]
the cochain-level homotopy witnessing graded commutativity of the cup product
on cohomology.

Taking \(f_m=h_n=c_l\) in \eqref{cupkrel} gives
\begin{align}
\label{cupkrel-self-expanded}
&\dd(c_l\hcup{k}c_l)
=
\dd c_l\hcup{k}c_l
+
(-1)^l c_l\hcup{k}\dd c_l
\\
&\quad+
\left[(-1)^{2l-k}+(-1)^{l^2+2l}\right]
c_l\hcup{k-1}c_l
\nonumber\\
&=
\dd c_l\hcup{k}c_l
+
(-1)^l c_l\hcup{k}\dd c_l
+
\left[(-1)^k+(-1)^l\right]
c_l\hcup{k-1}c_l .
\nonumber 
\end{align}

We will also need the following derived identity. First,
\begin{align}
\dd(c_l\hcup{k-1}c_l)
&=
\dd c_l\hcup{k-1}c_l
+
(-1)^l c_l\hcup{k-1}\dd c_l
\nonumber\\
&\quad+
\left[(-1)^{k-1}+(-1)^l\right]
c_l\hcup{k-2}c_l .
\label{derive-1}
\end{align}
Second, applying \eqref{cupkrel} to \(c_l\hcup{k}\dd c_l\), with
\(\deg(\dd c_l)=l+1\), gives
\begin{align}
&\dd(c_l\hcup{k}\dd c_l)
=
\dd c_l\hcup{k}\dd c_l
+
(-1)^{2l+1-k}c_l\hcup{k-1}\dd c_l
\nonumber\\
&\quad+
(-1)^{l(l+1)+2l+1}\dd c_l\hcup{k-1}c_l
\nonumber\\
&=
\dd c_l\hcup{k}\dd c_l
+
(-1)^{k+1}c_l\hcup{k-1}\dd c_l
-
\dd c_l\hcup{k-1}c_l .
\label{derive-2}
\end{align}
Adding \eqref{derive-1} and \eqref{derive-2}, the
\(\dd c_l\hcup{k-1}c_l\) terms cancel. Thus
\begin{align}
\label{cupkrel-derived-expanded}
&\dd\bigl(c_l\hcup{k-1}c_l+c_l\hcup{k}\dd c_l\bigr)
\\
&\quad=
\dd c_l\hcup{k}\dd c_l
+
\left[(-1)^l-(-1)^k\right]
\bigl(c_l\hcup{k-2}c_l+c_l\hcup{k-1}\dd c_l\bigr)
\nonumber\\
&\quad=
\dd c_l\hcup{k}\dd c_l
-
\left[(-1)^k-(-1)^l\right]
\bigl(c_l\hcup{k-2}c_l+c_l\hcup{k-1}\dd c_l\bigr).
\nonumber 
\end{align}

\subsection{Steenrod squares}

For a \(\Z_2\)-valued \(n\)-cocycle \(z_n\), the Steenrod square is
represented on cochains by
\[
\Sq^k(z_n)=z_n\hcup{n-k}z_n .
\]
This is a \(\Z_2\)-valued cocycle, and its cohomology class depends only
on the cohomology class of \(z_n\).

The Steenrod squares also obey the Adem relations. For \(0<a<2b\),
\begin{align}
\Sq^a\Sq^b
=
\sum_{j=0}^{\lfloor a/2\rfloor}
\binom{b-j-1}{a-2j}
\Sq^{a+b-j}\Sq^j ,
\end{align}
where the binomial coefficients are understood modulo \(2\).

For low total degrees, the relations are as follows.

\paragraph{Total degree \(a+b=2\).}
The only admissible pair with \(0<a<2b\) is
\[
(a,b)=(1,1).
\]
Thus
\begin{align}
\Sq^1\Sq^1
&=
\binom{0}{1}\Sq^2
=
0.
\end{align}
Hence
\begin{align}
\Sq^1\Sq^1=0.
\end{align}

\paragraph{Total degree \(a+b=3\).}
The pair
\[
(a,b)=(1,2)
\]
satisfies \(0<a<2b\). Therefore
\begin{align}
\Sq^1\Sq^2
&=
\binom{1}{1}\Sq^3
=
\Sq^3.
\end{align}
Thus
\begin{align}
\Sq^1\Sq^2=\Sq^3.
\end{align}
The other product of total degree \(3\),
\[
\Sq^2\Sq^1,
\]
does not satisfy \(2<2\), so there is no Adem reduction for it.

\paragraph{Total degree \(a+b=4\).}
There are two Adem relations.

First, for
\[
(a,b)=(1,3),
\]
we get
\begin{align}
\Sq^1\Sq^3
&=
\binom{2}{1}\Sq^4
=
0
\qquad \mod 2.
\end{align}
Hence
\begin{align}
\Sq^1\Sq^3=0.
\end{align}

Second, for
\[
(a,b)=(2,2),
\]
we get
\begin{align}
\Sq^2\Sq^2
&=
\binom{1}{2}\Sq^4
+
\binom{0}{0}\Sq^3\Sq^1
\nonumber\\
&=
0+\Sq^3\Sq^1.
\end{align}
Hence
\begin{align}
\Sq^2\Sq^2=\Sq^3\Sq^1.
\end{align}
Equivalently, since we work modulo \(2\),
\begin{align}
\label{Sq2Sq2Sq3Sq1}
\Sq^2\Sq^2+\Sq^3\Sq^1=0.
\end{align}
The remaining product of total degree \(4\),
\[
\Sq^3\Sq^1,
\]
is already Adem-admissible and is not reduced further by the Adem relations.

\subsection{Generalized Steenrod squares}
For an arbitrary \(n\)-cochain \(c_n\), we define
\begin{align}
\label{Sqdef}
\gSq^k c_n
\equiv
c_n\hcup{n-k}c_n
+
c_n\hcup{n-k+1}\dd c_n .
\end{align}
The sign of the second term is chosen to be compatible with
\eqref{cupkrel}. If \(c_n\) is closed, this reduces to
\[
\gSq^k c_n=c_n\hcup{n-k}c_n .
\]
For a \(\Z_2\)-valued cocycle, \(\gSq^k\) therefore reduces to the usual
Steenrod square \(\Sq^k\).

We now compute its coboundary. In
\eqref{cupkrel-derived-expanded}, set \(l=n\) and
\[
r=n-k+1 .
\]
Then
\[
r-1=n-k,\qquad r-2=n-k-1 .
\]
Thus
\begin{align}
&\dd\gSq^k c_n
=
\dd\bigl(c_n\hcup{n-k}c_n+c_n\hcup{n-k+1}\dd c_n\bigr)
\nonumber\\
= &
\dd c_n\hcup{n-k+1}\dd c_n
\\
&
\ \
-
\left[(-1)^{n-k+1}-(-1)^n\right]
\bigl(
c_n\hcup{n-k-1}c_n
+
c_n\hcup{n-k}\dd c_n
\bigr).
\nonumber 
\end{align}
The first term is precisely \(\gSq^k(\dd c_n)\), because
\(\dd c_n\) has degree \(n+1\). The expression in parentheses is
\(\gSq^{k+1}c_n\). Finally,
\[
-\left[(-1)^{n-k+1}-(-1)^n\right]
=
(-1)^n\left[1+(-1)^k\right].
\]
Therefore
\begin{align}
\label{Sqd1}
\dd\gSq^k c_n
=
\gSq^k(\dd c_n)
+
(-1)^n\left[1+(-1)^k\right]\gSq^{k+1}c_n .
\end{align}
Equivalently,
\[
\dd\gSq^k c_n
=
\gSq^k(\dd c_n)
+
(-1)^n
\begin{cases}
0, & k\ \text{odd},\\
2\,\gSq^{k+1}c_n, & k\ \text{even}.
\end{cases}
\]
After reducing coefficients mod \(2\), this becomes
\begin{align}
\label{Sqd}
\dd\gSq^k c_n
\se{2}
\gSq^k(\dd c_n).
\end{align}
Hence if \(c_n\) is a \(\Z_2\)-valued cocycle, then
\(\gSq^k c_n\) is also a \(\Z_2\)-valued cocycle.

We also record the additivity property carefully. Let
\(c_n,c_n'\) be \(\Z_2\)-valued \(n\)-cochains and put
\(s=n-k\). Expanding the definition gives
\begin{align}
\label{Sq-cross}
&\gSq^k(c_n+c_n')-\gSq^k c_n-\gSq^k c_n'
\\
&\quad =
c_n\hcup{s}c_n'
+
c_n'\hcup{s}c_n
+
c_n\hcup{s+1}\dd c_n'
+
c_n'\hcup{s+1}\dd c_n\\
&\quad =\dd(c_n\hcup{s+1}c'_n+\dd c_n\hcup{s+2}c'_n)+\dd c_n\hcup{s+2}\dd c'_n.
\nonumber 
\end{align}
 Hence, if either $\dd c_n=0$ or $\dd c_n'=0$,
\begin{align}
\label{Sqplus}
\gSq^k(c_n+c_n')
\se{2,\dd}
\gSq^k c_n+\gSq^k c_n' .
\end{align}

As a consequence, if \(c_n\) is a \(\Z_2\)-valued cocycle and
\(c_n\mapsto c_n+\dd f_{n-1}\), then
\[
\gSq^k(c_n+\dd f_{n-1})
\se{2,\dd}
\gSq^k c_n .
\]
Indeed, the extra term \(\gSq^k(\dd f_{n-1})\) is exact mod \(2\), since
\eqref{Sqd} applied to \(f_{n-1}\) gives
\[
\gSq^k(\dd f_{n-1})\se{2}\dd\gSq^k f_{n-1}.
\]
Therefore \(\Sq^k\) descends to a well-defined operation on
\(\Z_2\)-cohomology classes.

Finally, the same sign convention gives the usual cochain representative
of the Pontryagin square. If \(B_n\) is an integer lift of a
\(\Z_N\)-valued cocycle, so that
\[
\dd B_n=Nb_{n+1}, \ \ \ \ b_{n+1} \text{ is $\Z$-valued},
\]
then the Pontryagin-square representative compatible with the above
\(\hcup{1}\) convention is
\begin{align}
\label{Pontryagin-square-plus}
\mathfrak P(B_n)
=
B_nB_n+B_n\hcup{1}\dd B_n
=
\gSq^n B_n .
\end{align}
If another convention writes instead
\[
B_nB_n-B_n\hcup{1}\dd B_n,
\]
then on \(\Z_N\)-cocycles the two representatives differ by
\[
2B_n\hcup{1}\dd B_n
=
2N\,B_n\hcup{1}b_{n+1}
\se{2N}0 .
\]
Thus the sign difference is invisible after reduction mod \(2N\) on
lifted \(\Z_N\)-cocycles. However, for the integral cochain identity
\eqref{Sqd1}, the plus sign in \eqref{Sqdef} is the convention
compatible with \eqref{cupkrel}.

We also note that, for a $\R$-valued cochain $m_d$ and using
\eqn{cupkrel} and \eqref{Bsdef},
\begin{align}
\label{Sq1Bs}
& \gSq^1(m_{d}) = m_{d}\hcup{d-1} m_{d} + m_{d}\hcup{d} \dd m_{d}
\nonumber\\
&=\frac12 (-)^{d} 
[\dd (m_{d}\hcup{d} m_{d}) 
-\dd m_{d} \hcup{d} m_{d}] 
+\frac12  m_{d} \hcup{d} \dd m_{d} 
\nonumber\\
&=
(-)^{d} \Bs_2 (m_{d}\hcup{d} m_{d}) -(-)^d \Bs_2 m_{d} \hcup{d} m_{d}
+  m_{d} \hcup{d} \Bs_2 m_{d}
\nonumber\\
&=
(-)^{d} \Bs_2  \gSq^0 m_{d} 
-2 (-)^d \Bs_2 m_{d} \hcup{d+1} \Bs_2 m_{d}
\nonumber\\
&=
(-)^{d} \Bs_2 \gSq^0 m_{d} 
-2 (-)^d \gSq^0 \Bs_2 m_{d} 
\end{align}
This way, we obtain a relation between Steenrod square and Bockstein
homomorphism, when $m_d$ is a $\Z_2$-valued cochain
\begin{align}
\label{Sq1Bs2}
  \gSq^1(m_{d}) \se{2} \Bs_2 m_{d} ,
\end{align}
where we have used $\gSq^0 m_{d} \se{2} m_d$ for $\Z_2$-valued cochain.

\section{A simplicial model for the target space \(\cK\)}
\label{triK}

In this appendix, we construct an explicit simplicial model for a connected target space \(\cK\) with prescribed homotopy groups
\[
\pi_1(\cK)=G,\qquad
\pi_2(\cK)=A_2,\qquad
\pi_3(\cK)=A_3,\qquad \cdots ,
\]
where \(G\) is a discrete group and the \(A_j\) (\(j\ge 2\)) are Abelian groups. More precisely, we construct a reduced simplicial set \(K_\bullet\), whose geometric realization \(|K_\bullet|\) possesses the desired homotopy type. We will often simply write \(K\) for this simplicial model.

Strictly speaking, \(K_\bullet\) is a simplicial set rather than an ordinary abstract simplicial complex, as it features a single vertex but permits multiple distinct loop edges labeled by elements of \(G\). Equivalently, it can be viewed as a generalized simplicial complex that allows multiple simplices to share the same boundary.

We denote the resulting simplicial model by
\begin{widetext}
\begin{equation}\label{eq:Bdef}
K_\bullet =
\begin{cases}
K_\cdot(), & G=\{1\},\ A_2=A_3=\cdots=0,\\[2mm]
K_\cdot(G), & A_2=A_3=\cdots=0,\\[2mm]
K_{\alpha_2,k_3;\alpha_3,k_4;\cdots;\alpha_n,k_{n+1}}
(G,A_2,\cdots,A_n),
& A_{n+1}=A_{n+2}=\cdots=0 .
\end{cases}
\end{equation}
\end{widetext}
Here, \(K_\cdot()\) is the trivial one-point simplicial set, while \(K_\cdot(G)\) is a simplicial model of the classifying space \(BG\). The general model is built via an inductive Postnikov construction based on the following data:

\begin{enumerate}
\item A sequence of groups \(G, A_2, \ldots, A_n\) defining the homotopy groups of the target:
\[
\pi_1(K_\bullet)=G,\qquad
\pi_j(K_\bullet)=A_j,\quad 2\le j\le n .
\]

\item For each \(j=2,\ldots,n\), a group homomorphism
\[
\alpha_j:G\to \Aut(A_j)
\]
specifying the action of the fundamental group \(G=\pi_1(K_\bullet)\) on the higher homotopy group \(A_j=\pi_j(K_\bullet)\). We denote the group \(A_j\) equipped with this \(G\)-action by \(A_j^{\alpha_j}\).

\item A sequence of Postnikov cocycles \(k_3, k_4, \ldots, k_{n+1}\). Inductively, suppose the \((j-1)\)-stage
\[
K^{(j-1)}
=
K_{\alpha_2,k_3;\cdots;\alpha_{j-1},k_j}
(G,A_2,\ldots,A_{j-1})
\]
has been constructed. The next Postnikov class is represented by a cocycle
\[
k_{j+1}\in
Z^{j+1}\bigl(K^{(j-1)},A_j^{\alpha_j}\bigr),
\]
or equivalently, by a cohomology class \([k_{j+1}] \in H^{j+1}\bigl(K^{(j-1)},A_j^{\alpha_j}\bigr)\). This cocycle determines the next fibration, \(K(A_j,j)\longrightarrow K^{(j)}\longrightarrow K^{(j-1)}\).
\end{enumerate}

With fixed identifications of the homotopy groups, these are precisely the
standard Postnikov data of a connected homotopy \(n\)-type. 

\begin{figure}[t!]
\centering
\begin{tikzpicture}[
    scale=0.5,
    edge/.style={
        thick, 
        postaction={decorate},
        decoration={markings, mark=at position 0.55 with {\arrow{Stealth[length=3mm, width=2mm]}}}
    }
]

    \coordinate (V0) at (0,0);
    \coordinate (V1) at (4,0);
    \coordinate (V2) at (2,3.464); 

    \fill[blue!10] (V0) -- (V1) -- (V2) -- cycle;

    \draw[edge] (V0) -- (V1) node[midway, below=3pt] {$x_1^{01}$};
    \draw[edge] (V1) -- (V2) node[midway, right=4pt] {$x_1^{12}$};
    \draw[edge] (V0) -- (V2) node[midway, left=4pt] {$x_1^{02}$};

    \filldraw (V0) circle (1.5pt) node[below left] {pt};
    \filldraw (V1) circle (1.5pt) node[below right] {pt};
    \filldraw (V2) circle (1.5pt) node[above] {pt};

    \node at (2, 1.155) {$x_2^{012}$};

\end{tikzpicture}
\caption{A triangle, or \(2\)-simplex, in \([K_\bullet]_2\). The labels satisfy
\(x_1^{pq}\in G\), with \(x_1^{12}x_1^{01}=x_1^{02}\), and
\(x_2^{012}\in A_2\). All vertices \(pt_0,pt_1,pt_2\) are identified with the
single vertex \(\mathrm{pt}\).}
\label{fig:tetr2}
\end{figure}
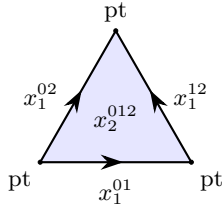

We now explicitly describe the simplices of \(K_\bullet\). The \(0\)-simplices consist of a single vertex, \([K_\bullet]_0=\{\mathrm{pt}\}\). The \(1\)-simplices correspond to \(G\), \([K_\bullet]_1=G\), where each group element labels a loop based at \(\mathrm{pt}\).

A \(2\)-simplex is decorated by group elements on its three edges and an \(A_2\)-label on its face (see Fig.~\ref{fig:tetr2}). We write
\[
s[012]=(x_1^{01},x_1^{02},x_1^{12};x_2^{012}),
\]
where \(x_1^{pq}\in G\) and \(x_2^{012}\in A_2\). The edge labels must satisfy the flatness condition:
\[
x_1^{12}(x_1^{02})^{-1}x_1^{01}=1
\quad \Longleftrightarrow \quad
x_1^{02}=x_1^{12}x_1^{01}.
\]
This ensures that the boundary of the triangle is null-homotopic in the \(BG\) direction.

More generally, a \(d\)-simplex in \(K_\bullet\) is a decorated simplex
\[
s[0\cdots d]
=
(x_1^{pq};x_2^{pqr};\cdots;x_q^{i_0\cdots i_q};\cdots),
\]
where \(x_1^{pq}\in G\) for \(0\le p<q\le d\) and, for \(2\le q\le \min(d,n)\),
\[
x_q^{i_0\cdots i_q}\in A_q,
\qquad
0\le i_0<\cdots<i_q\le d.
\]
Thus, \(x_1\) decorates edges, \(x_2\) decorates triangles, \(x_3\) decorates tetrahedra, and so on. For \(q>n\), we set \(A_q=0\). 

The face maps \(d_m\) are defined by deleting the \(m\)-th vertex: \(d_m s[0\cdots d]=s[0\cdots \widehat m\cdots d]\), and restricting the labels accordingly. Degeneracy maps repeat a vertex; labels on collapsed edges default to the identity in \(G\), while higher labels on degenerate faces vanish.

Using multiplicative notation for \(G\) and additive notation for the \(A_j\) groups, we define canonical cochains. The canonical \(G\)-valued \(1\)-cochain \(x_1\) projects onto the edge labels, \(x_1(s[01])=x_1^{01}\). Similarly, the canonical \(A_j\)-valued \(j\)-cochain \(x_j\) projects onto the \(j\)-face labels, \(x_j(s[i_0\cdots i_j])=x_j^{i_0\cdots i_j}\).

The base constraint is the flatness of \(x_1\): \(\dd x_1=1\). Imposing this yields the first Postnikov stage \(K_\cdot(G)\), a simplicial model of \(BG\) with \(\pi_1(K_\cdot(G))=G\) and vanishing higher homotopy groups. 

For higher labels (\(j\ge 2\)), we utilize the twisted coboundary operator \(\dd_{\alpha_j}\). On a \((j+1)\)-simplex, it evaluates as:
\begin{align}
(\dd_{\alpha_j}x_j)^{i_0\cdots i_{j+1}}
&=
\alpha_j(x_1^{i_0i_1})
x_j^{i_1\cdots i_{j+1}}
\nonumber\\
&\quad
+\sum_{m=1}^{j+1}
(-1)^m
x_j^{i_0\cdots \widehat{i_m}\cdots i_{j+1}} .
\end{align}
For instance, on a tetrahedron \(s[0123]\):
\begin{align}
\dd_{\alpha_2}x_2(s[0123])
&=
\alpha_2(x_1^{01})x_2^{123}
-x_2^{023}
+x_2^{013}
-x_2^{012}.
\end{align}

The next Postnikov stage is defined by the constraint \(\dd_{\alpha_2}x_2=k_3(x_1)\). Here, \(k_3\in Z^3(K_\cdot(G),A_2^{\alpha_2})\) represents the first Postnikov class. Because \(K_\cdot(G)\) models \(BG\), its cohomology maps to the group cohomology \(H^3(G,A_2^{\alpha_2})\). Therefore, \(k_3\) can be chosen as a normalized group cocycle depending only on the edge labels \(x_1\). Imposing this local constraint across all simplices yields the second stage \(K_{\alpha_2,k_3}(G,A_2)\).

This reasoning extends to all higher stages. Once \(K^{(j-1)}\) is constructed, the subsequent Postnikov cocycle \(k_{j+1}\) is chosen to depend only on the lower canonical cochains \(x_1,x_2,\ldots,x_{j-1}\). The \(j\)-th canonical cochain \(x_j\) must then satisfy \(\dd_{\alpha_j}x_j = k_{j+1}(x_1,\ldots,x_{j-1})\), defining the \(j\)-th Postnikov extension.

In summary, the \(d\)-simplices of the complete model \(K_\bullet\) form the set of all valid label assignments on the standard \(d\)-simplex \(\Delta^d\):
\begin{widetext}
\begin{align}
[K_\bullet]_d
=
\Bigg\{ &
(x_1, x_2, \ldots, x_{m}) \in C^1(\Delta^d, G) \times \prod_{q=2}^{m} C^q(\Delta^d, A_q)
\ \Bigg|
\nonumber\\
&\ \ \ \
\dd x_1=1,\ 
\dd_{\alpha_j}x_j
=
k_{j+1}(x_1,\ldots,x_{j-1}),
\ 
2\le j\le \min(n,d-1)
\Bigg\},
\end{align}
\end{widetext}
where \(m = \min(d,n)\). Here, \(C^q(\Delta^d, A_q)\) denotes the set of all possible \(A_q\)-valued assignments to the \(q\)-dimensional faces of the \(d\)-simplex. 

The consistency of these local equations is guaranteed by the cocycle condition \(\dd_{\alpha_j}k_{j+1}=0\) from the previous stage. Shifting \(k_{j+1}\) by a coboundary simply corresponds to redefining the cochain \(x_j\). By construction, the resulting simplicial set possesses the required homotopy groups:
\begin{align}
\pi_1(K_\bullet)&=G,\qquad
\pi_j(K_\bullet)=A_j,\quad 2\le j\le n,
\nonumber\\
\pi_m(K_\bullet)&=0,\quad m>n.
\end{align}
Thus, \(K_\bullet\) provides a rigorous simplicial model for the target connected homotopy \(n\)-type.

\section{The group $H^{d+2}\big(K(\Z_n,d-1),\RZ\big)$}
\label{Hd2}

We have
\begin{equation}
\boxed{
H^{d+2}\big(K(\Z_n,d-1),\RZ\big)
\cong
\begin{cases}
0, & d=2,\\[4pt]
\mathbb Z_{\gcd(n,2)}, & d\ge 3,
\end{cases}\;}
\end{equation}
\ie it is $\mathbb Z_2$ for every even $n$ and $0$ for odd $n$, in all
dimensions $d\ge 3$; it vanishes identically at $d=2$.

To understand the above result, we note that there is $\Z_n$-valued cocycle
$x_{d-1}$ (the canonical cocycle) on $K(\mathbb Z_n,d-1)$ satisfying
\begin{align}
 \dd x_{d-1} \se{n} 0 .
\end{align}
We work with the integral lift of  $x_{d-1}$, representing $x_{d-1}\in
C^{d-1}(K(\Z_n,d-1),\Z_n)$ by an $\Z$-valued cochain with values in
$\{0,\dots,n-1\}$.

We like to use $x_{d-1}$ to construct the generator of
$H^{d+2}(K(\Z_n,d-1);\RZ)$.  For example $\frac 1n x_{d-1}$ satisfies
\begin{align}
 \dd \frac 1n x_{d-1} \se{1}0,
\end{align}
Thus $\frac 1n x_{d-1}$ is a  $\RZ$-valued cocycle and is a generator of
$H^{d-1}(K(\Z_n,d-1);\RZ)$.

To construct a  $(d+2)$-cocycle, we apply $\gSq^3$ to $x_{d-1}$ and consider  
\begin{align}
\Om_{d+2}= \frac 1{n^2} \gSq^3 x_{d-1}.
\end{align}
To show $\frac {1}{n^2} \gSq^3 x_{d-1}$ is a non-trivial generator in
$H^{d-1}(K(\Z_n,d-1);\RZ)$, let 
\begin{align}
z_d \equiv
\tfrac1n\,\dd x_{d-1} = \Bs_n x_{d-1} \in C^{d}(K(\Z_n,d-1),\Z) .
\end{align}
Because $\R$ gives no positive-degree cohomology of a finite group space, the
Bockstein of $0\to\Z\to\R\to\RZ\to0$ is an isomorphism
\begin{align}
\label{RZtoZ}
\dd:H^{d+2}(K(\Z_n,d-1),\RZ)\;\xrightarrow{\cong}\;H^{d+3}(K(\Z_n,d-1),\Z).
\end{align}
Since (see \eqref{Sqdef})
\begin{equation}
\gSq^3(\dd x_{d-1})=(\dd x_{d-1})\hcup{d-3}(\dd x_{d-1})
= n^2\,z_d\hcup{d-3}z_d ,
\end{equation}
therefore
\begin{equation}
\dd \Om_{d+2} =
\tfrac{1}{n^2}\,\dd\gSq^3 x_{d-1}=z_d\hcup{d-3}z_d
\end{equation}
This is integer-valued, so $\Omega_{d+2}=\tfrac1{n^2}\gSq^3x_{d-1}$ is a
genuine $\RZ$-cocycle, and its integral Bockstein is $z_d\hcup{d-3}z_d$.  For
$d=2$, $\hcup{d-3}=\hcup{-1}=0$, so the integral Bockstein class is $0$, and
$\tfrac1{n^2}\gSq^3x_{d-1}$ is a $\RZ$-coboundary.

Apply \eqref{cupkrel} with $f_m=h_n=z_d$ (both degree $d$, closed) and $k=d-2$,
so that the ``$\hcup{k-1}$'' terms produce $\hcup{d-3}$:
\begin{align}
\dd\big(z_d\hcup{d-2}z_d\big)
&=\underbrace{(-1)^{d+2}}_{(-1)^d}\,z_d\hcup{d-3}z_d
+\underbrace{(-1)^{d^2+2d}}_{(-1)^d}\,z_d\hcup{d-3}z_d
\nonumber\\
&=2(-1)^d\,z_d\hcup{d-3}z_d .
\end{align}
Hence
\begin{equation}
2\,\big(z_d\hcup{d-3}z_d\big)=(-1)^d\,\dd\big(z_d\hcup{d-2}z_d\big)\quad\text{is exact.}
\end{equation}
So $[z_d\hcup{d-3}z_d]$ has order dividing $2$.

Since all reduced cohomology of $K(\Z_n,d-1)$ is $n$-torsion, hence contains
no $2$-torsion if $n$ is odd. In this case the class vanishes -- consistent
with $H^{d+2}(K(\Z_n,d-1),\RZ)=0$ for odd $n$.

It remains to show $[z_d\hcup{d-3}z_d]\neq 0$ for even $n$ and $d\ge3$.
First, we reduce  $n = $ even to $n=2$.  For even $n$ use the injection
$s:\Z_2\hookrightarrow\Z_n$, $1\mapsto m\equiv n/2$, inducing
$s:K(\Z_2,d-1)\to K(\Z_n,d-1)$. The fundamental cocycle pulls back to
$s^*x_{d-1}^{(n)}=m\,x_{d-1}^{(2)}$ (integer cochains). Since higher cup
products, $\dd$, and hence $\gSq$ are natural, and $\gSq$ is quadratic,
\begin{equation}
s^*\Big(\tfrac1{n^2}\gSq^3 x^{(n)}\Big)
=\tfrac{m^2}{n^2}\gSq^3 x^{(2)}
=\tfrac{1}{4}\gSq^3 x^{(2)}=\Omega^{(2)}_{d+2}.
\end{equation}
Thus $s^*\Omega^{(n)}=\Omega^{(2)}$: it suffices to prove non-triviality at $n=2$.

For $n=2$ ($m=1$) one has $\rho_2 z_d = \Sq^1 u$, where
$\rho_2$ is the mod 2 reduction $\Z\twoheadrightarrow  \Z_2$ and
$u=\rho_2 x_{d-1} \in H^{d-1}(K(\Z_2,d-1),\Z_2)$. 
\begin{equation}
z_d\hcup{d-3}z_d \se{2} \Sq^3(\rho_2 z_d)\;=\;\Sq^3 \Sq^1\,u .
\end{equation}
$\Sq^3\Sq^1$ is admissible with excess $e(3,1)=2$. By Serre's theorem,
$\Sq^I x_{d-1}$ is nonzero on the fundamental class $x_{d-1}$ of $K(\Z_2,d-1)$ whenever
$e(I)\le d-1$, \ie for $d\ge3$. Hence $\Sq^3\Sq^1u\neq0$, so the integral class
$z_d\hcup{d-3}z_d$ is nonzero for $d\ge3$.
So 
$\Om_{d+2}= \frac 1{n^2} \gSq^3 x_{d-1}$
 generates $H^{d+2}(K(\Z_n,d-1),\RZ)\cong \Z_{\gcd(n,2)}$
for $d\geq 3$. 

When $d=3$, $\Om_{5} = \frac 1{n^2} \gSq^3 x_{2}$ reduces to
\begin{align}
\label{Om5}
 \Om_{5} = \frac 1{n^2} \gSq^3 x_2 = \frac 1n x_2 \Bs_n x_2.
\end{align}

We can also show that 
\begin{align}
 \Om'_{d+2}= \frac 12 \gSq^2 \Bs_n x_{d-1}
\end{align}
generates $H^{d+2}(K(\Z_n,d-1),\RZ)$.
Applying \eqref{Sqd1}, we obtain
\begin{align}
 \dd \Om'_{d+2} &=  (-)^{d}\gSq^3 \Bs_n x_{d-1} 
\nonumber\\
&
= (-)^{d} \Bs_n x_{d-1} \hcup{d-3} \Bs_n x_{d-1}
= (-)^{d} \dd \Om_{d+2}.
\end{align}
So $\dd \Om'_{d+2}$ and $\dd \Om_{d+2}$ are the equivalent generators in
$H^{d+3}(K(\Z_n,d-1),\Z)$.  From Bockstein isomorphism \eqref{RZtoZ}, $
\Om'_{d+2}$ and $ \Om_{d+2}$ are the equivalent generators in
$H^{d+2}(K(\Z_n,d-1),\RZ)$.

\section{An introduction to the Serre spectral sequence}
\label{SSintro}

We briefly review the Serre spectral sequence in a form adapted to the
calculations used in this paper. Consider a fibration
\[
F\longrightarrow E\longrightarrow B .
\]
In our applications, \(E\) is a higher-group classifying space whose
cohomology we want to compute, \(B\) is a lower Postnikov stage, and
\(F\) is an Eilenberg-MacLane space describing the next higher gauge
field.

The key idea is that a cochain on \(E\) can be organized according to
how many degrees come from the base \(B\) and how many degrees come from
the fiber \(F\). Thus a total degree-\(n\) cochain may contain pieces of
bidegree
\[
(p,q),\qquad p+q=n,
\]
where \(p\) is the base degree and \(q\) is the fiber degree.

The Serre spectral sequence formalizes this by introducing a decreasing
filtration of the cochain complex:
\[
C^n(E;A)=\cF^0C^n
\supset
\cF^1C^n
\supset
\cF^2C^n
\supset
\cdots ,
\]
where, roughly,
\[
\cF^pC^n
=
\{\text{cochains of total degree \(n\) whose base degree is at least \(p\)}\}.
\]
Thus:
\[
\cF^0C^n
\]
contains all degree-\(n\) cochains,
\[
\cF^1C^n
\]
contains those with no purely-fiber component, and so on.

The ordinary coboundary preserves this filtration:
\[
\dd(\cF^pC^n)\subset \cF^pC^{n+1}.
\]
This simple fact is the starting point of the spectral sequence.

The spectral sequence consists of pages
\[
E_0,\ E_1,\ E_2,\ \ldots,\ E_r,\ \ldots,\ E_\infty .
\]
Each page is bigraded:
\[
E_r=\bigoplus_{p,q}E_r^{p,q}.
\]
The group \(E_r^{p,q}\) is built from total degree-\((p+q)\) cochains
whose leading nonzero part has base degree \(p\) and fiber degree \(q\),
but with equivalence relations that become stronger as \(r\) increases.

\subsection{The \(E_0\)-page}

The \(E_0\)-page is the associated graded object of the filtered cochain
complex:
\begin{align}
E_0^{p,q}
=
\frac{\cF^pC^{p+q}}{\cF^{p+1}C^{p+q}}.
\end{align}
Thus \(E_0^{p,q}\) keeps only the component of a total degree-\((p+q)\)
cochain with base degree exactly \(p\). Components with base degree
larger than \(p\) are invisible on \(E_0^{p,q}\).

For example, if a cochain has the form
\[
\Om=b_pf_q+b_{p+1}f_{q-1}+b_{p+2}f_{q-2}+\cdots ,
\]
where \(b_i\) has base degree \(i\) and \(f_j\) has fiber degree \(j\),
then on \(E_0^{p,q}\) only the leading term
\[
b_pf_q
\]
is seen.

\subsection{Cycles and boundaries on the \(r\)-th page}

For \(r\ge 0\), define
\begin{align}
Z_r^{p,q}
=
\left\{
\Om\in \cF^pC^{p+q}
\ \middle|\
\dd\Om\in \cF^{p+r}C^{p+q+1}
\right\}.
\end{align}
Thus \(\Om\in Z_r^{p,q}\) means that \(\dd\Om\) has no terms of base
degree
\[
p,\ p+1,\ \ldots,\ p+r-1.
\]
Equivalently, \(\Om\) is closed up to filtration degree \(p+r-1\).

The subgroup of terms that are already trivial on the \(r\)-th page is
\[
B_r^{p,q}
=
Z_{r-1}^{p+1,q-1}
+
\dd Z_{r-1}^{p-r+1,q+r-2}.
\]
The first term,
\[
Z_{r-1}^{p+1,q-1},
\]
consists of cochains whose leading base degree is already larger than
\(p\). Such cochains are invisible in bidegree \((p,q)\). The second
term,
\[
\dd Z_{r-1}^{p-r+1,q+r-2},
\]
consists of coboundaries coming from the previous stage.

The \(r\)-th page is then
\begin{align}
E_r^{p,q}
=
\frac{Z_r^{p,q}}{B_r^{p,q}}.
\end{align}
Thus a class
\[
[\Om]\in E_r^{p,q}
\]
is represented by a cochain \(\Om\in Z_r^{p,q}\), modulo the equivalence
relation generated by \(B_r^{p,q}\). It is important that
\[
[\Om]\in E_r^{p,q}
\]
is not yet an ordinary cohomology class of \(E\). It is a class on the
\(r\)-th page of the spectral sequence.

\subsection{The spectral-sequence differential}

Let
\[
\Om\in Z_r^{p,q}.
\]
By definition,
\[
\dd\Om\in \cF^{p+r}C^{p+q+1}.
\]
Furthermore,
\[
\dd^2\Om=0,
\]
so \(\dd\Om\) automatically represents an \(r\)-cycle in bidegree
\[
(p+r,q-r+1).
\]
Therefore \(\dd\Om\) defines a class in
\[
E_r^{p+r,q-r+1}.
\]
We define
\begin{align}
\dd_r[\Om]=[\dd\Om].
\end{align}
Thus
\begin{align}
\dd_r:E_r^{p,q}\longrightarrow E_r^{p+r,q-r+1}.
\end{align}
The differential \(\dd_r\) raises the base degree by \(r\), lowers the
fiber degree by \(r-1\), and raises the total degree by one:
\[
(p+r)+(q-r+1)=p+q+1.
\]

Although \(\dd_r\) is induced by the ordinary coboundary \(\dd\), it is not
the same object. The ordinary coboundary acts on cochains. The
differential \(\dd_r\) acts on equivalence classes on the \(r\)-th page.

The next page is obtained by taking cohomology with respect to \(\dd_r\):
\begin{align}
E_{r+1}^{p,q}
=
\frac{
\ker\bigl(\dd_r:E_r^{p,q}\to E_r^{p+r,q-r+1}\bigr)
}{
\operatorname{im}\bigl(\dd_r:E_r^{p-r,q+r-1}\to E_r^{p,q}\bigr)
}.
\end{align}
Thus passing from \(E_r\) to \(E_{r+1}\) does two things:
\begin{enumerate}
\item it keeps only the classes whose outgoing \(\dd_r\)-differential
vanishes;
\item it quotients out classes that are incoming \(\dd_r\)-differentials.
\end{enumerate}

In physics language:
\begin{widetext}
\begin{align}
\text{nonzero outgoing differential}
&\quad\Longleftrightarrow\quad
\text{the candidate topological term has an obstruction},
\nonumber\\[1mm]
\text{incoming differential}
&\quad\Longleftrightarrow\quad
\text{the candidate term is equivalent to a trivial term}.
\end{align}
\end{widetext}
This is analogous to ordinary cohomology: cocycles are killed by
\(\dd\), while coboundaries are hit by \(\dd\). A spectral sequence
repeats this idea page by page, with the successive differentials
\(\dd_r\).

\subsection{The first few pages}

We now spell out the first few pages in a way that is useful for
computations.

\paragraph{\(E_0\)-page.}

Since
\[
Z_0^{p,q}
=
\left\{
\Om\in \cF^pC^{p+q}
\ \middle|\
\dd\Om\in \cF^pC^{p+q+1}
\right\}
=
\cF^pC^{p+q},
\]
and
\[
B_0^{p,q}=\cF^{p+1}C^{p+q},
\]
we obtain
\begin{align}
E_0^{p,q}
=
\frac{\cF^pC^{p+q}}{\cF^{p+1}C^{p+q}}.
\end{align}
Therefore \(E_0^{p,q}\) consists of cochains of bidegree exactly
\((p,q)\).

\paragraph{\(E_1\)-page.}

The condition
\[
\Om\in Z_1^{p,q}
\]
means
\[
\dd\Om\in \cF^{p+1}C^{p+q+1}.
\]
Equivalently, the component of \(\dd\Om\) with base degree \(p\) must
vanish. This is precisely the fiber coboundary condition.

Thus a representative of \(E_1^{p,q}\) may be written schematically as
\[
b_p[f_q],
\]
where \(b_p\) is a degree-\(p\) cochain on the base and
\[
[f_q]\in H^q(F;A)
\]
is a fiber cohomology class. Hence
\begin{align}
E_1^{p,q}
\cong
C^p\bigl(B;H^q(F;A)\bigr).
\end{align}
At this stage we have taken cohomology only in the fiber direction.

\paragraph{\(E_2\)-page.}

The condition
\[
\Om\in Z_2^{p,q}
\]
means that \(\dd\Om\) has no components of base degree \(p\) or
\(p+1\). The first condition says that the fiber part is closed. The
second says that the base coboundary of the resulting fiber cohomology
class vanishes.

Thus the \(E_2\)-page is
\begin{align}
E_2^{p,q}
=
H^p\bigl(B;H^q(F;A)\bigr),
\end{align}
assuming that the action of \(\pi_1(B)\) on \(H^q(F;A)\) is trivial.
A typical class is written
\[
[b_p][f_q],
\]
where
\[
[b_p]\in H^p(B;A),
\qquad
[f_q]\in H^q(F;A).
\]

In many computations, the \(E_2\)-page is the first page one writes
down. The later differentials
\[
\dd_r:E_r^{p,q}\to E_r^{p+r,q-r+1}
\]
then encode the twisting of the fibration.

As \(r\) increases, the cycle condition becomes stronger:
\[
Z_0^{p,q}\supset Z_1^{p,q}\supset Z_2^{p,q}\supset\cdots .
\]
At the same time, more equivalences are imposed by the boundary groups
\(B_r^{p,q}\). Therefore classes may disappear either because they have
a nonzero outgoing differential or because they are hit by an incoming
differential.

\subsection{The \(E_\infty\)-page and the extension problem}

After all differentials are taken into account, the spectral sequence
stabilizes to the \(E_\infty\)-page. However, \(E_\infty\) is not
usually the cohomology group itself. Instead, it gives the associated
graded pieces of a filtration of the desired cohomology group.

For fixed total degree \(n\), there is a filtration
\begin{widetext}
\[
0=F^{n+1}H^n(E;A)
\subset
F^nH^n(E;A)
\subset
\cdots
\subset
F^1H^n(E;A)
\subset
F^0H^n(E;A)
=
H^n(E;A),
\]
\end{widetext}
such that
\begin{align}
E_\infty^{p,n-p}
\cong
F^pH^n(E;A)/F^{p+1}H^n(E;A).
\end{align}
Thus \(E_\infty\) tells us the graded pieces of \(H^n(E;A)\). One may
still need to solve an extension problem to reconstruct the actual
group.

For example, suppose the only nonzero \(E_\infty\)-terms in total degree
\(n\) are
\[
E_\infty^{0,n}\cong \Z_2,
\qquad
E_\infty^{n,0}\cong \Z_2.
\]
Then \(H^n(E;A)\) fits into a short exact sequence
\[
0\longrightarrow \Z_2
\longrightarrow H^n(E;A)
\longrightarrow \Z_2
\longrightarrow 0.
\]
There are two possibilities:
\[
H^n(E;A)\cong \Z_2\oplus\Z_2,
\]
or
\[
H^n(E;A)\cong \Z_4.
\]
The \(E_\infty\)-page alone does not distinguish these two cases. To
solve the extension problem, one constructs explicit cocycles and
checks their orders.

\section{Serre spectral sequence computation for
\(H^5(K_{x_2^2}(\Z_2,2,\Z_2,3);\RZ)\)}
\label{SSK1}

Let
\[
K_1=K_{x_2^2}(\Z_2,2,\Z_2,3)
\]
be the two-stage Postnikov system
\[
K(\Z_2,3)
\longrightarrow
K_1
\longrightarrow
K(\Z_2,2),
\]
with Postnikov class
\[
k_4=x_2^2\in H^4(K(\Z_2,2),\Z_2).
\]
Here
\[
x_2\in H^2(K(\Z_2,2),\Z_2)
\]
is the canonical generator.

Equivalently, on \(K_1\) there are canonical \(\Z_2\)-valued cochains
\[
x_2\in C^2(K_1,\Z_2),
\qquad
x_3\in C^3(K_1,\Z_2),
\]
satisfying
\begin{align}
\label{K1-cocycle-eq}
\dd x_2&\se{2}0,\\
\dd x_3&\se{2}x_2^2.
\nonumber
\end{align}
The second equation is the cochain-level manifestation of the
Postnikov class \(k_4=x_2^2\).

We want to compute
\[
H^5(K_1;\RZ).
\]
We use the Serre spectral sequence for the fibration
\[
K(\Z_2,3)\longrightarrow K_1\longrightarrow K(\Z_2,2).
\]
Since \(K(\Z_2,2)\) is simply connected, there is no nontrivial
\(\pi_1\)-action on the cohomology of the fiber. Hence
\begin{align}
E_2^{p,q}
&=
H^p\bigl(K(\Z_2,2);H^q(K(\Z_2,3);\RZ)\bigr)
\nonumber\\
&\ \ \ \
\Longrightarrow
H^{p+q}(K_1;\RZ).
\end{align}

\subsection{The relevant \(E_2\)-page}

We only need classes that can affect total degree \(5\), together with
their nearby differential targets.

On the base \(K(\Z_2,2)\), define
\begin{align}
Q_4&:=\frac14\,\mathfrak P(x_2)
\in H^4(K(\Z_2,2);\RZ),\\
S_5&:=\frac12\,x_2\Sq^1x_2
\in H^5(K(\Z_2,2);\RZ),\\
R_6&:=\frac12\,x_2^3
\in H^6(K(\Z_2,2);\RZ).
\end{align}
Here
\[
H^4(K(\Z_2,2);\RZ)\cong \Z_4
\]
is generated by \(Q_4\). The class \(S_5\) is the string
self-statistics class, while \(R_6\) will appear as the \(\dd_4\)-image
of the particle-string mutual term.

On the fiber \(K(\Z_2,3)\), let
\[
x_3\in H^3(K(\Z_2,3),\Z_2)
\]
be the canonical generator, and define
\begin{align}
U_3&:=\frac12\,x_3
\in H^3(K(\Z_2,3);\RZ),\\
P_5&:=\frac12\,\Sq^2x_3
=
\frac12\,(x_3\hcup{1}x_3)
\in H^5(K(\Z_2,3);\RZ).
\end{align}
The class \(P_5\) is the fiber class corresponding to fermionic
particle self-statistics.

Before imposing the Postnikov constraint \(\dd x_3=x_2^2\), the
total-degree-\(5\) candidates are
\begin{align}
S_5&=\frac12\,x_2\Sq^1x_2,\\
M_5&=\frac12\,x_2x_3,\\
P_5&=\frac12\,\Sq^2x_3.
\end{align}
They correspond respectively to
\[
\begin{array}{ccl}
S_5 &:& \text{string self-statistics},\\[1mm]
M_5 &:& \text{particle-string mutual statistics},\\[1mm]
P_5 &:& \text{particle self-statistics}.
\end{array}
\]

The relevant part of the \(E_2\)-page is:
\[
\begin{array}{c|ccccccc}
q\backslash p
&0&1&2&3&4&5&6\\
\hline
5
&\Z_2\langle P_5\rangle
&\ast&\ast&\ast&\ast&\ast&\ast
\\[1mm]
4
&0&\ast&\ast&\ast&\ast&\ast&\ast
\\[1mm]
3
&\Z_2\langle U_3\rangle
&\ast
&\Z_2\langle M_5\rangle
&\ast&\ast&\ast&\ast
\\[1mm]
2
&0&\ast&\ast&\ast&\ast&\ast&\ast
\\[1mm]
1
&0&\ast&\ast&\ast&\ast&\ast&\ast
\\[1mm]
0
&\RZ
&0&0&0
&\Z_4\langle Q_4\rangle
&\Z_2\langle S_5\rangle
&\Z_2\langle R_6\rangle
\\
\hline
\end{array}
\]
The diagonal \(p+q=5\) contains precisely
\[
P_5,\qquad M_5,\qquad S_5.
\]

\subsection{The Postnikov differential \(\dd_4\)}

In this fibration, \(x_2\) is a base cochain and \(x_3\) is a fiber
cochain. Hence we assign bidegrees
\[
x_2:\ (2,0),
\qquad
x_3:\ (0,3).
\]
The Postnikov relation
\[
\dd x_3\se{2}x_2^2
\]
takes a cochain of bidegree \((0,3)\) to a cochain of bidegree
\((4,0)\). Therefore it raises base degree by \(4\) and lowers fiber
degree by \(3\). This is precisely the bidegree of a \(\dd_4\)-differential:
\[
\dd_4:E_4^{p,q}\longrightarrow E_4^{p+4,q-3}.
\]
There are no nonzero \(\dd_2\) or \(\dd_3\) differentials generated by this
Postnikov class, so
\[
E_2=E_3=E_4
\]
for the entries relevant to this computation.

\paragraph{The differential of \(U_3\).}

Consider
\[
U_3=\frac12\,x_3.
\]
On the fiber, \(x_2=0\), so \(U_3\) is closed. On the total space
\(K_1\), however,
\begin{align}
\dd U_3
=
\dd\left(\frac12x_3\right)
=
\frac12\,\dd x_3
\se{1}
\frac12\,x_2^2.
\end{align}
The right-hand side has bidegree \((4,0)\). Therefore
\[
\dd_4(U_3)=\frac12\,x_2^2.
\]
Since
\[
2Q_4
=
\frac12\,\mathfrak P(x_2)
\se{1,\dd}
\frac12\,x_2^2,
\]
we write
\begin{align}
\label{d4-U3}
\dd_4(U_3)=2Q_4.
\end{align}
Thus \(U_3\) does not survive to \(E_5\).

Physically, the pure particle \(U_3\) term cannot be extended to a
closed topological term on the total space once the Postnikov constraint
\[
\dd x_3=x_2^2
\]
is imposed.

\paragraph{The differential of \(M_5\).}

Now consider the mutual particle-string candidate
\[
M_5=\frac12\,x_2x_3.
\]
Using \(\dd x_2=0\), we find
\begin{align}
\dd M_5
&=
\dd\left(\frac12\,x_2x_3\right)\\
&=
\frac12\,x_2\,\dd x_3
\nonumber\\
&\se{1}
\frac12\,x_2^3.
\nonumber
\end{align}
The result has bidegree \((6,0)\), so it defines a \(\dd_4\)-differential
\[
\dd_4:E_4^{2,3}\longrightarrow E_4^{6,0}.
\]
Therefore
\begin{align}
\label{d4-M5}
\dd_4(M_5)
=
\frac12\,x_2^3
=
R_6.
\end{align}
Thus the mutual particle-string term does not survive to \(E_\infty\).

The relevant \(\dd_4\)-differentials are therefore
\[
\begin{array}{rcl}
\dd_4:E_4^{0,3}&\longrightarrow&E_4^{4,0},\\[1mm]
U_3&\longmapsto&2Q_4,
\end{array}
\qquad
\begin{array}{rcl}
\dd_4:E_4^{2,3}&\longrightarrow&E_4^{6,0},\\[1mm]
M_5&\longmapsto&R_6.
\end{array}
\]

Graphically:
\begin{widetext}
\[
\xymatrix@C=1.8em@R=1.2em{
q\backslash p
&0&1&2&3&4&5&6
\\
5
&\Z_2\langle P_5\rangle
&\ast&\ast&\ast&\ast&\ast&\ast
\\
4
&0&\ast&\ast&\ast&\ast&\ast&\ast
\\
3
&\Z_2\langle U_3\rangle \ar[dddrrrr]^{\dd_4}
&\ast
&\Z_2\langle M_5\rangle \ar[dddrrrr]^{\dd_4}
&\ast&\ast&\ast&\ast
\\
2
&0&\ast&\ast&\ast&\ast&\ast&\ast
\\
1
&0&\ast&\ast&\ast&\ast&\ast&\ast
\\
0
&\RZ
&0&0&0
&\Z_4\langle Q_4\rangle
&\Z_2\langle S_5\rangle
&\Z_2\langle R_6\rangle .
}
\]
\end{widetext}

\subsection{Tracking the total-degree-\(5\) diagonal}

We now examine the three total-degree-\(5\) candidates:
\[
P_5,\qquad M_5,\qquad S_5.
\]

\paragraph{The class \(S_5\).}

The base class
\[
S_5=\frac12\,x_2\Sq^1x_2
\]
lies in
\[
E_4^{5,0}.
\]
It has no possible outgoing differential, because a differential
\[
\dd_r:E_r^{5,0}\to E_r^{5+r,1-r}
\]
would land in negative fiber degree for \(r\ge 2\). It also has no
incoming differential relevant to this computation. Hence \(S_5\)
survives to \(E_\infty\).

\paragraph{The class \(M_5\).}

The mutual class
\[
M_5=\frac12\,x_2x_3
\]
does not survive, since
\[
\dd_4(M_5)=R_6\neq0.
\]
Thus an independent particle-string mutual-statistics term is obstructed
by the nontrivial Postnikov constraint
\[
\dd x_3=x_2^2.
\]

\paragraph{The class \(P_5\).}

The fiber class
\[
P_5=\frac12\,\Sq^2x_3
\]
lies in
\[
E_4^{0,5}.
\]
Its possible \(\dd_4\)-target is
\[
E_4^{4,2},
\]
but the relevant fiber cohomology in degree \(2\) vanishes. Therefore
\(\dd_4(P_5)=0\).

There is also no possible \(\dd_5\)-target relevant here. The next possible
obstruction is
\[
\dd_6:E_6^{0,5}\longrightarrow E_6^{6,0}.
\]
At the cochain level, this obstruction is governed by the Adem
combination
\[
\Sq^2\Sq^2x_2+\Sq^3\Sq^1x_2.
\]
The Adem relation
\[
\Sq^2\Sq^2+\Sq^3\Sq^1=0
\]
implies that this obstruction vanishes in cohomology. Therefore
\(P_5\) survives to \(E_\infty\). However, lifting \(P_5\) to an actual
closed cochain on \(K_1\) requires adding a secondary operation. This
secondary operation will also solve the extension problem below.

Thus in total degree \(5\), the \(E_\infty\)-page has two surviving
\(\Z_2\) pieces:
\begin{align}
E_\infty^{5,0}&\cong \Z_2\langle S_5\rangle,\\
E_\infty^{0,5}&\cong \Z_2\langle P_5\rangle.
\end{align}
Therefore there is a filtration
\[
0\subset F^5\subset H^5(K_1;\RZ)
\]
with
\[
F^5\cong \Z_2\langle S_5\rangle,
\qquad
H^5(K_1;\RZ)/F^5
\cong
\Z_2\langle P_5\rangle.
\]
Equivalently,
\begin{align}
0\rightarrow
\Z_2\langle S_5\rangle
\rightarrow
H^5(K_1;\RZ)
\rightarrow
\Z_2\langle P_5\rangle
\rightarrow
0.
\end{align}
Thus the associated graded group is
\[
\Z_2\oplus\Z_2.
\]
The actual group is either
\[
\Z_2\oplus\Z_2
\]
or
\[
\Z_4.
\]
We now solve this extension problem explicitly.

\subsection{Solving the extension problem by an explicit cocycle}

The naive particle self-statistics term is
\[
\frac12\,\Sq^2x_3.
\]
Because \(x_3\) is not closed on \(K_1\), one must use the generalized
Steenrod square
\[
\gSq^2x_3
=
x_3\hcup{1}x_3+x_3\hcup{2}\dd x_3.
\]
This expression reduces to the ordinary representative
\[
\Sq^2x_3=x_3\hcup{1}x_3
\]
on the fiber, where \(\dd x_3=0\). Hence
\[
\frac12\,\gSq^2x_3
\]
is the natural lift of \(P_5\) from the fiber to the total space.

However, it is not closed by itself. Using the generalized Steenrod
identity
\[
\dd\gSq^k c\se{2}\gSq^k(\dd c)
\]
for \(\Z_2\)-valued cochains, we get
\begin{align}
\dd\,\gSq^2x_3
\se{2}
\gSq^2(\dd x_3)
\se{2}
\gSq^2(x_2^2).
\end{align}
Since \(x_2\) is closed mod \(2\), \(x_2^2=\Sq^2x_2\), and therefore
\[
\gSq^2(x_2^2)
\se{2}
\Sq^2\Sq^2x_2.
\]
The Adem relation says
\[
\Sq^2\Sq^2+\Sq^3\Sq^1=0.
\]
Thus
\[
\Sq^2\Sq^2x_2
\se{2,\dd}
\Sq^3\Sq^1x_2.
\]
More precisely, choose a \(\Z_2\)-valued \(5\)-cochain
\[
\Del_5(x_2)
\]
such that
\begin{align}
\label{Delta5-def}
\dd\Del_5(x_2)
\se{2}
\Sq^2\Sq^2x_2+\Sq^3\Sq^1x_2.
\end{align}
This cochain exists because the right-hand side is the cochain-level
representative of the Adem relation.

We claim that a closed lift of \(P_5\) is
\begin{align}
\label{Omega5-def}
\Om_5(x_2,x_3)
=
\frac12\,\gSq^2x_3
+
\frac18\,\gSq^3x_2
+
\frac12\,\Del_5(x_2).
\end{align}

Let us check closedness carefully. Since $x_2$ is $\Z_2$-valued cocycle,
we have
\[
\dd x_3\se{2}\gSq^2x_2 \se{2} \Sq^2x_2
\]
and
\begin{align}
\dd\left(\frac12\,\gSq^2x_3\right)
&\se{1}
\frac12\,\gSq^2(\dd x_3)
\se{1}
\frac12\,\Sq^2\Sq^2x_2.
\end{align}
Again because \(x_2\) is a \(\Z_2\)-valued cocycle, an integer lift
satisfies
\[
\dd x_2=2\Bs_2x_2.
\]
For \(x_2\) of degree \(2\),
\[
\gSq^3x_2=x_2\,\dd x_2.
\]
Hence
\begin{align}
\dd\left(\frac18\,\gSq^3x_2\right)
&=
\frac18\,\dd(x_2\,\dd x_2)
=
\frac18\,(\dd x_2)^2
\nonumber\\
&=
\frac18\,(2\Bs_2x_2)^2
=
\frac12\,(\Bs_2x_2)^2.
\end{align}
Modulo \(2\),
\[
\Bs_2x_2\se{2}\Sq^1x_2.
\]
Therefore
\begin{align}
\dd\left(\frac18\,\gSq^3x_2\right)
\se{1,\dd}
\frac12(\Sq^1x_2)^2
\se{1,\dd}
\frac12\,\Sq^3\Sq^1x_2
\end{align}
Finally,
\[
\dd\left(\frac12\,\Del_5(x_2)\right)
=
\frac12\,\dd\Del_5(x_2).
\]
Using \eqref{Delta5-def}, the three pieces add to
\begin{align}
\dd\Om_5
&\se{1}
\frac12\,\Sq^2\Sq^2x_2
+
\frac12\,\Sq^3\Sq^1x_2
\nonumber\\
&\ \ \ \
+
\frac12\left(
\Sq^2\Sq^2x_2+\Sq^3\Sq^1x_2
\right)
\nonumber\\
&\se{1}0.
\end{align}
Thus
\[
\dd\Om_5\se{1}0,
\]
and \(\Om_5\) defines a class in
\[
H^5(K_1;\RZ).
\]
By construction, its leading fiber component is
\[
\frac12\,\Sq^2x_3=P_5.
\]
Hence \(\Om_5\) is a lift of the surviving \(E_\infty^{0,5}\) class.

\subsection{The order of \(\Om_5\)}

We now compute twice this class. Since
\[
\gSq^2x_3
\qquad\text{and}\qquad
\Del_5(x_2)
\]
are integer-valued cochains, multiplying their coefficients by \(2\)
makes those terms vanish in \(\RZ\). Therefore
\begin{align}
2\Om_5(x_2,x_3)
\se{1}
\frac14\,\gSq^3x_2.
\end{align}
Using
\[
\gSq^3x_2=x_2\,\dd x_2
=
2x_2\Bs_2x_2,
\]
we obtain
\begin{align}
2\Om_5(x_2,x_3)
&\se{1}
\frac14\cdot 2x_2\Bs_2x_2
\nonumber\\
&=
\frac12\,x_2\Bs_2x_2.
\end{align}
Again using
\[
\Bs_2x_2\se{2,\dd}\Sq^1x_2,
\]
we find
\begin{align}
2\Om_5(x_2,x_3)
\se{1,\dd}
\frac12\,x_2\Sq^1x_2
=
S_5.
\end{align}
Thus the lift \(\Om_5\) satisfies
\begin{align}
\label{2Omega-S5}
2\Om_5=S_5.
\end{align}
Since \(S_5\neq0\), the class \(\Om_5\) has order \(4\).

This solves the extension problem. The short exact sequence
\[
0\longrightarrow
\Z_2\langle S_5\rangle
\longrightarrow
H^5(K_1;\RZ)
\longrightarrow
\Z_2\langle P_5\rangle
\longrightarrow
0
\]
is not split. Instead, it is the nonsplit extension
\[
0\longrightarrow
\Z_2
\longrightarrow
\Z_4
\longrightarrow
\Z_2
\longrightarrow
0.
\]
Therefore
\begin{align}
H^5(K_1;\RZ)\cong \Z_4.
\end{align}

A generator may be chosen as
\begin{align}
\label{K1-H5-generator}
\Om_5(x_2,x_3)
=
\frac12\,\gSq^2x_3
+
\frac18\,\gSq^3x_2
+
\frac12\,\Del_5(x_2),
\end{align}
where
\[
\dd\Del_5(x_2)
\se{2}
\gSq^2\gSq^2x_2+\gSq^3\gSq^1x_2.
\]
It obeys
\begin{align}
2\Om_5
\se{1,\dd}
\frac12\,x_2\Sq^1x_2
=
S_5.
\end{align}

Thus the spectral sequence first gives the associated graded group
\[
\operatorname{gr}H^5(K_1;\RZ)\cong \Z_2\oplus\Z_2,
\]
with graded pieces generated by \(P_5\) and \(S_5\). The explicit
secondary-operation cocycle \(\Om_5\) shows that these two \(\Z_2\)
pieces combine into a single \(\Z_4\):
\[
H^5(K_1;\RZ)\cong \Z_4.
\]

\bibliography{./all,./publst,./local}

\begin{thebibliography}{59}%
\makeatletter
\providecommand \@ifxundefined [1]{%
 \@ifx{#1\undefined}
}%
\providecommand \@ifnum [1]{%
 \ifnum #1\expandafter \@firstoftwo
 \else \expandafter \@secondoftwo
 \fi
}%
\providecommand \@ifx [1]{%
 \ifx #1\expandafter \@firstoftwo
 \else \expandafter \@secondoftwo
 \fi
}%
\providecommand \natexlab [1]{#1}%
\providecommand \enquote  [1]{``#1''}%
\providecommand \bibnamefont  [1]{#1}%
\providecommand \bibfnamefont [1]{#1}%
\providecommand \citenamefont [1]{#1}%
\providecommand \href@noop [0]{\@secondoftwo}%
\providecommand \href [0]{\begingroup \@sanitize@url \@href}%
\providecommand \@href[1]{\@@startlink{#1}\@@href}%
\providecommand \@@href[1]{\endgroup#1\@@endlink}%
\providecommand \@sanitize@url [0]{\catcode `\\12\catcode `\$12\catcode
  `\&12\catcode `\#12\catcode `\^12\catcode `\_12\catcode `\%12\relax}%
\providecommand \@@startlink[1]{}%
\providecommand \@@endlink[0]{}%
\providecommand \url  [0]{\begingroup\@sanitize@url \@url }%
\providecommand \@url [1]{\endgroup\@href {#1}{\urlprefix }}%
\providecommand \urlprefix  [0]{URL }%
\providecommand \Eprint [0]{\href }%
\providecommand \doibase [0]{https://doi.org/}%
\providecommand \selectlanguage [0]{\@gobble}%
\providecommand \bibinfo  [0]{\@secondoftwo}%
\providecommand \bibfield  [0]{\@secondoftwo}%
\providecommand \translation [1]{[#1]}%
\providecommand \BibitemOpen [0]{}%
\providecommand \bibitemStop [0]{}%
\providecommand \bibitemNoStop [0]{.\EOS\space}%
\providecommand \EOS [0]{\spacefactor3000\relax}%
\providecommand \BibitemShut  [1]{\csname bibitem#1\endcsname}%
\let\auto@bib@innerbib\@empty
\bibitem [{\citenamefont {Leinaas}\ and\ \citenamefont
  {Myrheim}(1977)}]{LM7701}%
  \BibitemOpen
  \bibfield  {author} {\bibinfo {author} {\bibfnamefont {J.~M.}\ \bibnamefont
  {Leinaas}}\ and\ \bibinfo {author} {\bibfnamefont {J.}~\bibnamefont
  {Myrheim}},\ }\bibfield  {title} {\bibinfo {title} {On the theory of
  identical particles},\ }\href {https://doi.org/10.1007/bf02727953} {\bibfield
   {journal} {\bibinfo  {journal} {Nuovo Cim B}\ }\textbf {\bibinfo {volume}
  {37}},\ \bibinfo {pages} {1} (\bibinfo {year} {1977})}\BibitemShut {NoStop}%
\bibitem [{\citenamefont {Wilczek}(1982)}]{W8257}%
  \BibitemOpen
  \bibfield  {author} {\bibinfo {author} {\bibfnamefont {F.}~\bibnamefont
  {Wilczek}},\ }\bibfield  {title} {\bibinfo {title} {Quantum mechanics of
  fractional-spin particles},\ }\href
  {https://doi.org/10.1103/physrevlett.49.957} {\bibfield  {journal} {\bibinfo
  {journal} {Physical Review Letter}\ }\textbf {\bibinfo {volume} {49}},\
  \bibinfo {pages} {957} (\bibinfo {year} {1982})}\BibitemShut {NoStop}%
\bibitem [{\citenamefont {Halperin}(1984)}]{H8483}%
  \BibitemOpen
  \bibfield  {author} {\bibinfo {author} {\bibfnamefont {B.~I.}\ \bibnamefont
  {Halperin}},\ }\bibfield  {title} {\bibinfo {title} {Statistics of
  quasiparticles and the hierarchy of fractional quantized {{H}all} states},\
  }\href@noop {} {\bibfield  {journal} {\bibinfo  {journal} {Physical Review
  Letter}\ }\textbf {\bibinfo {volume} {52}},\ \bibinfo {pages} {1583}
  (\bibinfo {year} {1984})}\BibitemShut {NoStop}%
\bibitem [{\citenamefont {Arovas}\ \emph {et~al.}(1984)\citenamefont {Arovas},
  \citenamefont {Schrieffer},\ and\ \citenamefont {Wilczek}}]{ASW8422}%
  \BibitemOpen
  \bibfield  {author} {\bibinfo {author} {\bibfnamefont {D.}~\bibnamefont
  {Arovas}}, \bibinfo {author} {\bibfnamefont {J.~R.}\ \bibnamefont
  {Schrieffer}},\ and\ \bibinfo {author} {\bibfnamefont {F.}~\bibnamefont
  {Wilczek}},\ }\bibfield  {title} {\bibinfo {title} {Fractional statistics and
  the quantum {H}all effect},\ }\href
  {https://doi.org/10.1103/physrevlett.53.722} {\bibfield  {journal} {\bibinfo
  {journal} {Physical Review Letter}\ }\textbf {\bibinfo {volume} {53}},\
  \bibinfo {pages} {722} (\bibinfo {year} {1984})}\BibitemShut {NoStop}%
\bibitem [{\citenamefont {Wu}(1984)}]{W8411}%
  \BibitemOpen
  \bibfield  {author} {\bibinfo {author} {\bibfnamefont {Y.-S.}\ \bibnamefont
  {Wu}},\ }\bibfield  {title} {\bibinfo {title} {Multiparticle quantum
  mechanics obeying fractional statistics},\ }\href
  {https://doi.org/10.1103/physrevlett.53.111} {\bibfield  {journal} {\bibinfo
  {journal} {Physical Review Letter}\ }\textbf {\bibinfo {volume} {53}},\
  \bibinfo {pages} {111} (\bibinfo {year} {1984})}\BibitemShut {NoStop}%
\bibitem [{\citenamefont {Witten}(1989)}]{W8951}%
  \BibitemOpen
  \bibfield  {author} {\bibinfo {author} {\bibfnamefont {E.}~\bibnamefont
  {Witten}},\ }\bibfield  {title} {\bibinfo {title} {Quantum field theory and
  the {Jones} polynomial},\ }\href {https://doi.org/10.1007/bf01217730}
  {\bibfield  {journal} {\bibinfo  {journal} {Commun. Math. Phys.}\ }\textbf
  {\bibinfo {volume} {121}},\ \bibinfo {pages} {351} (\bibinfo {year}
  {1989})}\BibitemShut {NoStop}%
\bibitem [{\citenamefont {Moore}\ and\ \citenamefont {Read}(1991)}]{MR9162}%
  \BibitemOpen
  \bibfield  {author} {\bibinfo {author} {\bibfnamefont {G.}~\bibnamefont
  {Moore}}\ and\ \bibinfo {author} {\bibfnamefont {N.}~\bibnamefont {Read}},\
  }\bibfield  {title} {\bibinfo {title} {Nonabelions in the fractional quantum
  hall effect},\ }\href {https://doi.org/10.1016/0550-3213(91)90407-o}
  {\bibfield  {journal} {\bibinfo  {journal} {Nucl. Phys. B}\ }\textbf
  {\bibinfo {volume} {360}},\ \bibinfo {pages} {362} (\bibinfo {year}
  {1991})}\BibitemShut {NoStop}%
\bibitem [{\citenamefont {Wen}(1991)}]{W9102}%
  \BibitemOpen
  \bibfield  {author} {\bibinfo {author} {\bibfnamefont {X.-G.}\ \bibnamefont
  {Wen}},\ }\bibfield  {title} {\bibinfo {title} {Non-{A}belian statistics in
  the {FQH} states},\ }\href {https://doi.org/10.1103/PhysRevLett.66.802}
  {\bibfield  {journal} {\bibinfo  {journal} {Phys. Rev. Lett.}\ }\textbf
  {\bibinfo {volume} {66}},\ \bibinfo {pages} {802} (\bibinfo {year}
  {1991})}\BibitemShut {NoStop}%
\bibitem [{\citenamefont {Wang}\ and\ \citenamefont {Levin}(2014)}]{WL1437}%
  \BibitemOpen
  \bibfield  {author} {\bibinfo {author} {\bibfnamefont {C.}~\bibnamefont
  {Wang}}\ and\ \bibinfo {author} {\bibfnamefont {M.}~\bibnamefont {Levin}},\
  }\bibfield  {title} {\bibinfo {title} {Braiding statistics of loop
  excitations in three dimensions},\ }\href
  {https://doi.org/10.1103/physrevlett.113.080403} {\bibfield  {journal}
  {\bibinfo  {journal} {Physical Review Letter}\ }\textbf {\bibinfo {volume}
  {113}},\ \bibinfo {pages} {080403} (\bibinfo {year} {2014})},\ \Eprint
  {https://arxiv.org/abs/1403.7437} {arXiv:1403.7437} \BibitemShut {NoStop}%
\bibitem [{\citenamefont {{Jiang}}\ \emph {et~al.}(2014)\citenamefont
  {{Jiang}}, \citenamefont {{Mesaros}},\ and\ \citenamefont {{Ran}}}]{JMR1462}%
  \BibitemOpen
  \bibfield  {author} {\bibinfo {author} {\bibfnamefont {S.}~\bibnamefont
  {{Jiang}}}, \bibinfo {author} {\bibfnamefont {A.}~\bibnamefont {{Mesaros}}},\
  and\ \bibinfo {author} {\bibfnamefont {Y.}~\bibnamefont {{Ran}}},\ }\bibfield
   {title} {\bibinfo {title} {{Generalized Modular Transformations in
  $(3+1)\mathrm{D}$ Topologically Ordered Phases and Triple Linking Invariant
  of Loop Braiding}},\ }\href {https://doi.org/10.1103/physrevx.4.031048}
  {\bibfield  {journal} {\bibinfo  {journal} {Phys. Rev. X}\ }\textbf {\bibinfo
  {volume} {4}},\ \bibinfo {pages} {031048} (\bibinfo {year} {2014})},\ \Eprint
  {https://arxiv.org/abs/1404.1062} {arXiv:1404.1062} \BibitemShut {NoStop}%
\bibitem [{\citenamefont {Wang}\ and\ \citenamefont {Wen}(2015)}]{WW1454}%
  \BibitemOpen
  \bibfield  {author} {\bibinfo {author} {\bibfnamefont {J.~C.}\ \bibnamefont
  {Wang}}\ and\ \bibinfo {author} {\bibfnamefont {X.-G.}\ \bibnamefont {Wen}},\
  }\bibfield  {title} {\bibinfo {title} {Non-{A}belian string and particle
  braiding in topological order: {Modular SL(3,Z) representation}
  and(3+1)-dimensional twisted gauge theory},\ }\href
  {https://doi.org/10.1103/physrevb.91.035134} {\bibfield  {journal} {\bibinfo
  {journal} {Phys. Rev. B}\ }\textbf {\bibinfo {volume} {91}},\ \bibinfo
  {pages} {035134} (\bibinfo {year} {2015})},\ \Eprint
  {https://arxiv.org/abs/1404.7854} {arXiv:1404.7854} \BibitemShut {NoStop}%
\bibitem [{\citenamefont {Thorngren}(2015)}]{T14044385}%
  \BibitemOpen
  \bibfield  {author} {\bibinfo {author} {\bibfnamefont {R.}~\bibnamefont
  {Thorngren}},\ }\bibfield  {title} {\bibinfo {title} {Framed {Wilson}
  operators, fermionic strings, and gravitational anomaly in 4d},\ }\href
  {https://doi.org/10.1007/jhep02(2015)152} {\bibfield  {journal} {\bibinfo
  {journal} {J. High Energ. Phys.}\ }\textbf {\bibinfo {volume}
  {2015}}\bibfield  {number} {\bibinfo  {number} { (2)},\ \bibinfo {pages}
  {152}},\ }\Eprint {https://arxiv.org/abs/1404.4385} {arXiv:1404.4385}
  \BibitemShut {NoStop}%
\bibitem [{\citenamefont {Fidkowski}\ \emph {et~al.}(2022)\citenamefont
  {Fidkowski}, \citenamefont {Haah},\ and\ \citenamefont
  {Hastings}}]{FH211014654}%
  \BibitemOpen
  \bibfield  {author} {\bibinfo {author} {\bibfnamefont {L.}~\bibnamefont
  {Fidkowski}}, \bibinfo {author} {\bibfnamefont {J.}~\bibnamefont {Haah}},\
  and\ \bibinfo {author} {\bibfnamefont {M.~B.}\ \bibnamefont {Hastings}},\
  }\bibfield  {title} {\bibinfo {title} {Gravitational anomaly of 3+1
  dimensional $z_2$ toric code with fermionic charges and fermionic loop
  self-statistics},\ }\bibfield  {journal} {\bibinfo  {journal} {Physical
  Review B}\ }\textbf {\bibinfo {volume} {106}},\ \href
  {https://doi.org/10.1103/physrevb.106.165135} {10.1103/physrevb.106.165135}
  (\bibinfo {year} {2022}),\ \Eprint {https://arxiv.org/abs/2110.14654}
  {arXiv:2110.14654} \BibitemShut {NoStop}%
\bibitem [{\citenamefont {{Wan}}\ \emph {et~al.}(2021)\citenamefont {{Wan}},
  \citenamefont {{Wang}},\ and\ \citenamefont {{Wen}}}]{WW211212148}%
  \BibitemOpen
  \bibfield  {author} {\bibinfo {author} {\bibfnamefont {Z.}~\bibnamefont
  {{Wan}}}, \bibinfo {author} {\bibfnamefont {J.}~\bibnamefont {{Wang}}},\ and\
  \bibinfo {author} {\bibfnamefont {X.-G.}\ \bibnamefont {{Wen}}},\ }\bibfield
  {title} {\bibinfo {title} {{3+1d Boundaries with Gravitational Anomaly of
  4+1d Invertible Topological Order for Branch-Independent Bosonic Systems}},\
  }\href@noop {} {\  (\bibinfo {year} {2021})},\ \Eprint
  {https://arxiv.org/abs/2112.12148} {arXiv:2112.12148} \BibitemShut {NoStop}%
\bibitem [{\citenamefont {Kobayashi}\ \emph {et~al.}(2026)\citenamefont
  {Kobayashi}, \citenamefont {Li}, \citenamefont {Xue}, \citenamefont {Hsin},\
  and\ \citenamefont {Chen}}]{KC241201886}%
  \BibitemOpen
  \bibfield  {author} {\bibinfo {author} {\bibfnamefont {R.}~\bibnamefont
  {Kobayashi}}, \bibinfo {author} {\bibfnamefont {Y.}~\bibnamefont {Li}},
  \bibinfo {author} {\bibfnamefont {H.}~\bibnamefont {Xue}}, \bibinfo {author}
  {\bibfnamefont {P.-S.}\ \bibnamefont {Hsin}},\ and\ \bibinfo {author}
  {\bibfnamefont {Y.-A.}\ \bibnamefont {Chen}},\ }\bibfield  {title} {\bibinfo
  {title} {Generalized statistics on lattices},\ }\bibfield  {journal}
  {\bibinfo  {journal} {Physical Review X}\ }\textbf {\bibinfo {volume} {16}},\
  \href {https://doi.org/10.1103/6k88-w52n} {10.1103/6k88-w52n} (\bibinfo
  {year} {2026}),\ \Eprint {https://arxiv.org/abs/2412.01886}
  {arXiv:2412.01886} \BibitemShut {NoStop}%
\bibitem [{\citenamefont {Xue}(2026)}]{Xue2026StatisticsAbelian}%
  \BibitemOpen
  \bibfield  {author} {\bibinfo {author} {\bibfnamefont {H.}~\bibnamefont
  {Xue}},\ }\bibfield  {title} {\bibinfo {title} {Statistics of abelian
  topological excitations},\ }\bibfield  {journal} {\bibinfo  {journal}
  {Physical Review B}\ }\href {https://doi.org/10.1103/g3nc-fwqg}
  {10.1103/g3nc-fwqg} (\bibinfo {year} {2026})\BibitemShut {NoStop}%
\bibitem [{\citenamefont {Ji}\ and\ \citenamefont {Wen}(2020)}]{JW191213492}%
  \BibitemOpen
  \bibfield  {author} {\bibinfo {author} {\bibfnamefont {W.}~\bibnamefont
  {Ji}}\ and\ \bibinfo {author} {\bibfnamefont {X.-G.}\ \bibnamefont {Wen}},\
  }\bibfield  {title} {\bibinfo {title} {Categorical symmetry and noninvertible
  anomaly in symmetry-breaking and topological phase transitions},\ }\href
  {https://doi.org/10.1103/PhysRevResearch.2.033417} {\bibfield  {journal}
  {\bibinfo  {journal} {Phys. Rev. Res.}\ }\textbf {\bibinfo {volume} {2}},\
  \bibinfo {pages} {033417} (\bibinfo {year} {2020})},\ \Eprint
  {https://arxiv.org/abs/1912.13492} {arXiv:1912.13492} \BibitemShut {NoStop}%
\bibitem [{\citenamefont {{Kong}}\ \emph {et~al.}(2020)\citenamefont {{Kong}},
  \citenamefont {{Lan}}, \citenamefont {{Wen}}, \citenamefont {{Zhang}},\ and\
  \citenamefont {{Zheng}}}]{KZ200514178}%
  \BibitemOpen
  \bibfield  {author} {\bibinfo {author} {\bibfnamefont {L.}~\bibnamefont
  {{Kong}}}, \bibinfo {author} {\bibfnamefont {T.}~\bibnamefont {{Lan}}},
  \bibinfo {author} {\bibfnamefont {X.-G.}\ \bibnamefont {{Wen}}}, \bibinfo
  {author} {\bibfnamefont {Z.-H.}\ \bibnamefont {{Zhang}}},\ and\ \bibinfo
  {author} {\bibfnamefont {H.}~\bibnamefont {{Zheng}}},\ }\bibfield  {title}
  {\bibinfo {title} {{Algebraic higher symmetry and categorical symmetry: A
  holographic and entanglement view of symmetry}},\ }\href
  {https://doi.org/10.1103/PhysRevResearch.2.043086} {\bibfield  {journal}
  {\bibinfo  {journal} {Phys. Rev. Res.}\ }\textbf {\bibinfo {volume} {2}},\
  \bibinfo {pages} {043086} (\bibinfo {year} {2020})},\ \Eprint
  {https://arxiv.org/abs/2005.14178} {arXiv:2005.14178} \BibitemShut {NoStop}%
\bibitem [{\citenamefont {Levin}\ and\ \citenamefont {Wen}(2003)}]{LW0316}%
  \BibitemOpen
  \bibfield  {author} {\bibinfo {author} {\bibfnamefont {M.}~\bibnamefont
  {Levin}}\ and\ \bibinfo {author} {\bibfnamefont {X.-G.}\ \bibnamefont
  {Wen}},\ }\bibfield  {title} {\bibinfo {title} {Fermions, strings, and gauge
  fields in lattice spin models},\ }\href
  {https://doi.org/10.1103/physrevb.67.245316} {\bibfield  {journal} {\bibinfo
  {journal} {Phys. Rev. B}\ }\textbf {\bibinfo {volume} {67}},\ \bibinfo
  {pages} {245316} (\bibinfo {year} {2003})},\ \Eprint
  {https://arxiv.org/abs/cond-mat/0302460} {arXiv:cond-mat/0302460}
  \BibitemShut {NoStop}%
\bibitem [{\citenamefont {{Baez}}\ and\ \citenamefont
  {{Lauda}}(2003)}]{BLm0307200}%
  \BibitemOpen
  \bibfield  {author} {\bibinfo {author} {\bibfnamefont {J.~C.}\ \bibnamefont
  {{Baez}}}\ and\ \bibinfo {author} {\bibfnamefont {A.~D.}\ \bibnamefont
  {{Lauda}}},\ }\bibfield  {title} {\bibinfo {title} {{Higher-Dimensional
  Algebra V: 2-Groups}},\ }\href {https://doi.org/10.48550/arXiv.math/0307200}
  {\ ,\ \bibinfo {pages} {math/0307200} (\bibinfo {year} {2003})},\ \Eprint
  {https://arxiv.org/abs/math/0307200} {arXiv:math/0307200} \BibitemShut
  {NoStop}%
\bibitem [{\citenamefont {{Kapustin}}\ and\ \citenamefont
  {{Thorngren}}(2017)}]{KT13094721}%
  \BibitemOpen
  \bibfield  {author} {\bibinfo {author} {\bibfnamefont {A.}~\bibnamefont
  {{Kapustin}}}\ and\ \bibinfo {author} {\bibfnamefont {R.}~\bibnamefont
  {{Thorngren}}},\ }\bibinfo {title} {{Higher Symmetry and Gapped Phases of
  Gauge Theories}},\ in\ \href {https://doi.org/10.1007/978-3-319-59939-7_5}
  {\emph {\bibinfo {booktitle} {{Algebra, Geometry, and Physics in the 21st
  Century: Kontsevich Festschrift}}}},\ \bibinfo {editor} {edited by\ \bibinfo
  {editor} {\bibfnamefont {D.}~\bibnamefont {Auroux}}, \bibinfo {editor}
  {\bibfnamefont {L.}~\bibnamefont {Katzarkov}}, \bibinfo {editor}
  {\bibfnamefont {T.}~\bibnamefont {Pantev}}, \bibinfo {editor} {\bibfnamefont
  {Y.}~\bibnamefont {Soibelman}},\ and\ \bibinfo {editor} {\bibfnamefont
  {Y.}~\bibnamefont {Tschinkel}}}\ (\bibinfo  {publisher} {Birkh\"auser,
  Cham},\ \bibinfo {year} {2017})\ pp.\ \bibinfo {pages} {177--202},\ \Eprint
  {https://arxiv.org/abs/1309.4721} {arXiv:1309.4721} \BibitemShut {NoStop}%
\bibitem [{\citenamefont {Levin}\ and\ \citenamefont {Gu}(2012)}]{LG1209}%
  \BibitemOpen
  \bibfield  {author} {\bibinfo {author} {\bibfnamefont {M.}~\bibnamefont
  {Levin}}\ and\ \bibinfo {author} {\bibfnamefont {Z.-C.}\ \bibnamefont {Gu}},\
  }\bibfield  {title} {\bibinfo {title} {Braiding statistics approach to
  symmetry-protected topological phases},\ }\href
  {https://doi.org/10.1103/physrevb.86.115109} {\bibfield  {journal} {\bibinfo
  {journal} {Physical Review B}\ }\textbf {\bibinfo {volume} {86}},\ \bibinfo
  {pages} {115109} (\bibinfo {year} {2012})},\ \Eprint
  {https://arxiv.org/abs/1202.3120} {arXiv:1202.3120} \BibitemShut {NoStop}%
\bibitem [{\citenamefont {{Barkeshli}}\ \emph {et~al.}(2019)\citenamefont
  {{Barkeshli}}, \citenamefont {{Bonderson}}, \citenamefont {{Cheng}},\ and\
  \citenamefont {{Wang}}}]{BZ14104540}%
  \BibitemOpen
  \bibfield  {author} {\bibinfo {author} {\bibfnamefont {M.}~\bibnamefont
  {{Barkeshli}}}, \bibinfo {author} {\bibfnamefont {P.}~\bibnamefont
  {{Bonderson}}}, \bibinfo {author} {\bibfnamefont {M.}~\bibnamefont
  {{Cheng}}},\ and\ \bibinfo {author} {\bibfnamefont {Z.}~\bibnamefont
  {{Wang}}},\ }\bibfield  {title} {\bibinfo {title} {{Symmetry
  fractionalization, defects, and gauging of topological phases}},\ }\href
  {https://doi.org/10.1103/PhysRevB.100.115147} {\bibfield  {journal} {\bibinfo
   {journal} {Physical Review B}\ }\textbf {\bibinfo {volume} {100}},\ \bibinfo
  {pages} {115147} (\bibinfo {year} {2019})},\ \Eprint
  {https://arxiv.org/abs/1410.4540} {arXiv:1410.4540} \BibitemShut {NoStop}%
\bibitem [{\citenamefont {Wen}(2019)}]{W181202517}%
  \BibitemOpen
  \bibfield  {author} {\bibinfo {author} {\bibfnamefont {X.-G.}\ \bibnamefont
  {Wen}},\ }\bibfield  {title} {\bibinfo {title} {Emergent (anomalous) higher
  symmetries from topological order and from dynamical electromagnetic field in
  condensed matter systems},\ }\href
  {https://doi.org/10.1103/physrevb.99.205139} {\bibfield  {journal} {\bibinfo
  {journal} {Phys. Rev. B}\ }\textbf {\bibinfo {volume} {99}},\ \bibinfo
  {pages} {205139} (\bibinfo {year} {2019})},\ \Eprint
  {https://arxiv.org/abs/1812.02517} {arXiv:1812.02517} \BibitemShut {NoStop}%
\bibitem [{\citenamefont {Dijkgraaf}\ and\ \citenamefont
  {Witten}(1990)}]{DW9093}%
  \BibitemOpen
  \bibfield  {author} {\bibinfo {author} {\bibfnamefont {R.}~\bibnamefont
  {Dijkgraaf}}\ and\ \bibinfo {author} {\bibfnamefont {E.}~\bibnamefont
  {Witten}},\ }\bibfield  {title} {\bibinfo {title} {Topological gauge theories
  and group cohomology},\ }\href {https://doi.org/10.1007/bf02096988}
  {\bibfield  {journal} {\bibinfo  {journal} {Commun. Math. Phys.}\ }\textbf
  {\bibinfo {volume} {129}},\ \bibinfo {pages} {393} (\bibinfo {year}
  {1990})}\BibitemShut {NoStop}%
\bibitem [{\citenamefont {Freed}\ and\ \citenamefont
  {Quinn}(1993)}]{FreedQuinn1993}%
  \BibitemOpen
  \bibfield  {author} {\bibinfo {author} {\bibfnamefont {D.~S.}\ \bibnamefont
  {Freed}}\ and\ \bibinfo {author} {\bibfnamefont {F.}~\bibnamefont {Quinn}},\
  }\bibfield  {title} {\bibinfo {title} {Chern--simons theory with finite gauge
  group},\ }\href {https://doi.org/10.1007/BF02096860} {\bibfield  {journal}
  {\bibinfo  {journal} {Communications in Mathematical Physics}\ }\textbf
  {\bibinfo {volume} {156}},\ \bibinfo {pages} {435} (\bibinfo {year}
  {1993})}\BibitemShut {NoStop}%
\bibitem [{\citenamefont {Joyal}\ and\ \citenamefont
  {Street}(1993)}]{JoyalStreet1993}%
  \BibitemOpen
  \bibfield  {author} {\bibinfo {author} {\bibfnamefont {A.}~\bibnamefont
  {Joyal}}\ and\ \bibinfo {author} {\bibfnamefont {R.}~\bibnamefont {Street}},\
  }\bibfield  {title} {\bibinfo {title} {Braided tensor categories},\ }\href
  {https://doi.org/10.1006/aima.1993.1055} {\bibfield  {journal} {\bibinfo
  {journal} {Advances in Mathematics}\ }\textbf {\bibinfo {volume} {102}},\
  \bibinfo {pages} {20} (\bibinfo {year} {1993})}\BibitemShut {NoStop}%
\bibitem [{\citenamefont {Etingof}\ \emph
  {et~al.}(2015{\natexlab{a}})\citenamefont {Etingof}, \citenamefont {Gelaki},
  \citenamefont {Nikshych},\ and\ \citenamefont
  {Ostrik}}]{EtingofGelakiNikshychOstrik2015}%
  \BibitemOpen
  \bibfield  {author} {\bibinfo {author} {\bibfnamefont {P.}~\bibnamefont
  {Etingof}}, \bibinfo {author} {\bibfnamefont {S.}~\bibnamefont {Gelaki}},
  \bibinfo {author} {\bibfnamefont {D.}~\bibnamefont {Nikshych}},\ and\
  \bibinfo {author} {\bibfnamefont {V.}~\bibnamefont {Ostrik}},\ }\href@noop {}
  {\emph {\bibinfo {title} {Tensor Categories}}},\ \bibinfo {series}
  {Mathematical Surveys and Monographs}, Vol.\ \bibinfo {volume} {205}\
  (\bibinfo  {publisher} {American Mathematical Society},\ \bibinfo {address}
  {Providence, RI},\ \bibinfo {year} {2015})\BibitemShut {NoStop}%
\bibitem [{\citenamefont {Gu}\ and\ \citenamefont {Wen}(2014)}]{GuWen2014}%
  \BibitemOpen
  \bibfield  {author} {\bibinfo {author} {\bibfnamefont {Z.-C.}\ \bibnamefont
  {Gu}}\ and\ \bibinfo {author} {\bibfnamefont {X.-G.}\ \bibnamefont {Wen}},\
  }\bibfield  {title} {\bibinfo {title} {Symmetry-protected topological orders
  for interacting fermions: Fermionic topological nonlinear sigma models and a
  special group supercohomology theory},\ }\href
  {https://doi.org/10.1103/PhysRevB.90.115141} {\bibfield  {journal} {\bibinfo
  {journal} {Physical Review B}\ }\textbf {\bibinfo {volume} {90}},\ \bibinfo
  {pages} {115141} (\bibinfo {year} {2014})},\ \Eprint
  {https://arxiv.org/abs/1201.2648} {arXiv:1201.2648 [cond-mat.str-el]}
  \BibitemShut {NoStop}%
\bibitem [{\citenamefont {Kapustin}\ \emph {et~al.}(2015)\citenamefont
  {Kapustin}, \citenamefont {Thorngren}, \citenamefont {Turzillo},\ and\
  \citenamefont {Wang}}]{KTT1429}%
  \BibitemOpen
  \bibfield  {author} {\bibinfo {author} {\bibfnamefont {A.}~\bibnamefont
  {Kapustin}}, \bibinfo {author} {\bibfnamefont {R.}~\bibnamefont {Thorngren}},
  \bibinfo {author} {\bibfnamefont {A.}~\bibnamefont {Turzillo}},\ and\
  \bibinfo {author} {\bibfnamefont {Z.}~\bibnamefont {Wang}},\ }\bibfield
  {title} {\bibinfo {title} {Fermionic symmetry protected topological phases
  and cobordisms},\ }\href {https://doi.org/10.1007/jhep12(2015)052} {\bibfield
   {journal} {\bibinfo  {journal} {J. High Energ. Phys.}\ }\textbf {\bibinfo
  {volume} {2015}}\bibfield  {number} {\bibinfo  {number} { (12)},\ \bibinfo
  {pages} {1}},\ }\Eprint {https://arxiv.org/abs/1406.7329} {arXiv:1406.7329}
  \BibitemShut {NoStop}%
\bibitem [{\citenamefont {Wen}(2017)}]{W161201418}%
  \BibitemOpen
  \bibfield  {author} {\bibinfo {author} {\bibfnamefont {X.-G.}\ \bibnamefont
  {Wen}},\ }\bibfield  {title} {\bibinfo {title} {Exactly soluble local bosonic
  cocycle models, statistical transmutation, and simplest time-reversal
  symmetric topological orders in 3+1 dimensions},\ }\href
  {https://doi.org/10.1103/physrevb.95.205142} {\bibfield  {journal} {\bibinfo
  {journal} {Phys. Rev. B}\ }\textbf {\bibinfo {volume} {95}},\ \bibinfo
  {pages} {205142} (\bibinfo {year} {2017})},\ \Eprint
  {https://arxiv.org/abs/1612.01418} {arXiv:1612.01418} \BibitemShut {NoStop}%
\bibitem [{\citenamefont {Kapustin}\ and\ \citenamefont
  {Thorngren}(2017)}]{KT170108264}%
  \BibitemOpen
  \bibfield  {author} {\bibinfo {author} {\bibfnamefont {A.}~\bibnamefont
  {Kapustin}}\ and\ \bibinfo {author} {\bibfnamefont {R.}~\bibnamefont
  {Thorngren}},\ }\bibfield  {title} {\bibinfo {title} {Fermionic {SPT} phases
  in higher dimensions and bosonization},\ }\href
  {https://doi.org/10.1007/jhep10(2017)080} {\bibfield  {journal} {\bibinfo
  {journal} {J. High Energ. Phys.}\ }\textbf {\bibinfo {volume}
  {2017}}\bibfield  {number} {\bibinfo  {number} { (10)},\ \bibinfo {pages}
  {80}},\ }\Eprint {https://arxiv.org/abs/1701.08264} {arXiv:1701.08264}
  \BibitemShut {NoStop}%
\bibitem [{\citenamefont {Wang}\ and\ \citenamefont {Gu}(2018)}]{WangGu2018}%
  \BibitemOpen
  \bibfield  {author} {\bibinfo {author} {\bibfnamefont {Q.-R.}\ \bibnamefont
  {Wang}}\ and\ \bibinfo {author} {\bibfnamefont {Z.-C.}\ \bibnamefont {Gu}},\
  }\bibfield  {title} {\bibinfo {title} {Towards a complete classification of
  symmetry-protected topological phases for interacting fermions in three
  dimensions and a general group supercohomology theory},\ }\href
  {https://doi.org/10.1103/PhysRevX.8.011055} {\bibfield  {journal} {\bibinfo
  {journal} {Physical Review X}\ }\textbf {\bibinfo {volume} {8}},\ \bibinfo
  {pages} {011055} (\bibinfo {year} {2018})},\ \Eprint
  {https://arxiv.org/abs/1703.10937} {arXiv:1703.10937 [cond-mat.str-el]}
  \BibitemShut {NoStop}%
\bibitem [{\citenamefont {{Chatterjee}}\ and\ \citenamefont
  {{Wen}}(2023)}]{CW220303596}%
  \BibitemOpen
  \bibfield  {author} {\bibinfo {author} {\bibfnamefont {A.}~\bibnamefont
  {{Chatterjee}}}\ and\ \bibinfo {author} {\bibfnamefont {X.-G.}\ \bibnamefont
  {{Wen}}},\ }\bibfield  {title} {\bibinfo {title} {{Symmetry as a shadow of
  topological order and a derivation of topological holographic principle}},\
  }\href {https://doi.org/10.1103/PhysRevB.107.155136} {\bibfield  {journal}
  {\bibinfo  {journal} {Phys. Rev. B}\ }\textbf {\bibinfo {volume} {107}},\
  \bibinfo {pages} {155136} (\bibinfo {year} {2023})},\ \Eprint
  {https://arxiv.org/abs/2203.03596} {arXiv:2203.03596} \BibitemShut {NoStop}%
\bibitem [{\citenamefont {{Freed}}\ \emph {et~al.}(2022)\citenamefont
  {{Freed}}, \citenamefont {{Moore}},\ and\ \citenamefont
  {{Teleman}}}]{FT220907471}%
  \BibitemOpen
  \bibfield  {author} {\bibinfo {author} {\bibfnamefont {D.~S.}\ \bibnamefont
  {{Freed}}}, \bibinfo {author} {\bibfnamefont {G.~W.}\ \bibnamefont
  {{Moore}}},\ and\ \bibinfo {author} {\bibfnamefont {C.}~\bibnamefont
  {{Teleman}}},\ }\href@noop {} {\bibinfo {title} {{Topological symmetry in
  quantum field theory}}} (\bibinfo {year} {2022}),\ \Eprint
  {https://arxiv.org/abs/2209.07471} {arXiv:2209.07471} \BibitemShut {NoStop}%
\bibitem [{\citenamefont {Kong}\ \emph {et~al.}(2015)\citenamefont {Kong},
  \citenamefont {Wen},\ and\ \citenamefont {Zheng}}]{KZ150201690}%
  \BibitemOpen
  \bibfield  {author} {\bibinfo {author} {\bibfnamefont {L.}~\bibnamefont
  {Kong}}, \bibinfo {author} {\bibfnamefont {X.-G.}\ \bibnamefont {Wen}},\ and\
  \bibinfo {author} {\bibfnamefont {H.}~\bibnamefont {Zheng}},\ }\href@noop {}
  {\bibinfo {title} {Boundary-bulk relation for topological orders as the
  functor mapping higher categories to their centers}} (\bibinfo {year}
  {2015}),\ \Eprint {https://arxiv.org/abs/1502.01690} {arXiv:1502.01690}
  \BibitemShut {NoStop}%
\bibitem [{\citenamefont {{Thorngren}}\ and\ \citenamefont
  {{Wang}}(2024)}]{TW191202817}%
  \BibitemOpen
  \bibfield  {author} {\bibinfo {author} {\bibfnamefont {R.}~\bibnamefont
  {{Thorngren}}}\ and\ \bibinfo {author} {\bibfnamefont {Y.}~\bibnamefont
  {{Wang}}},\ }\bibfield  {title} {\bibinfo {title} {{Fusion Category Symmetry
  I: Anomaly In-Flow and Gapped Phases}},\ }\href
  {https://doi.org/10.1007/JHEP04(2024)132} {\bibfield  {journal} {\bibinfo
  {journal} {J. High Energ. Phys.}\ }\textbf {\bibinfo {volume} {2024}},\
  \bibinfo {pages} {132}},\ \Eprint {https://arxiv.org/abs/1912.02817}
  {arXiv:1912.02817} \BibitemShut {NoStop}%
\bibitem [{\citenamefont {Kong}\ and\ \citenamefont {Wen}(2014)}]{KW1458}%
  \BibitemOpen
  \bibfield  {author} {\bibinfo {author} {\bibfnamefont {L.}~\bibnamefont
  {Kong}}\ and\ \bibinfo {author} {\bibfnamefont {X.-G.}\ \bibnamefont {Wen}},\
  }\bibfield  {title} {\bibinfo {title} {Braided fusion categories,
  gravitational anomalies, and the mathematical framework for topological
  orders in any dimensions},\ }\href@noop {} {\  (\bibinfo {year} {2014})},\
  \Eprint {https://arxiv.org/abs/1405.5858} {arXiv:1405.5858} \BibitemShut
  {NoStop}%
\bibitem [{\citenamefont {{Lichtman}}\ \emph {et~al.}(2021)\citenamefont
  {{Lichtman}}, \citenamefont {{Thorngren}}, \citenamefont {{Lindner}},
  \citenamefont {{Stern}},\ and\ \citenamefont {{Berg}}}]{LB200304328}%
  \BibitemOpen
  \bibfield  {author} {\bibinfo {author} {\bibfnamefont {T.}~\bibnamefont
  {{Lichtman}}}, \bibinfo {author} {\bibfnamefont {R.}~\bibnamefont
  {{Thorngren}}}, \bibinfo {author} {\bibfnamefont {N.~H.}\ \bibnamefont
  {{Lindner}}}, \bibinfo {author} {\bibfnamefont {A.}~\bibnamefont {{Stern}}},\
  and\ \bibinfo {author} {\bibfnamefont {E.}~\bibnamefont {{Berg}}},\
  }\bibfield  {title} {\bibinfo {title} {{Bulk anyons as edge symmetries:
  Boundary phase diagrams of topologically ordered states}},\ }\href
  {https://doi.org/10.1103/PhysRevB.104.075141} {\bibfield  {journal} {\bibinfo
   {journal} {Physical Review B}\ }\textbf {\bibinfo {volume} {104}},\ \bibinfo
  {pages} {075141} (\bibinfo {year} {2021})},\ \Eprint
  {https://arxiv.org/abs/2003.04328} {arXiv:2003.04328} \BibitemShut {NoStop}%
\bibitem [{\citenamefont {{Gaiotto}}\ and\ \citenamefont
  {{Kulp}}(2021)}]{GK200805960}%
  \BibitemOpen
  \bibfield  {author} {\bibinfo {author} {\bibfnamefont {D.}~\bibnamefont
  {{Gaiotto}}}\ and\ \bibinfo {author} {\bibfnamefont {J.}~\bibnamefont
  {{Kulp}}},\ }\bibfield  {title} {\bibinfo {title} {{Orbifold groupoids}},\
  }\href {https://doi.org/10.1007/JHEP02(2021)132} {\bibfield  {journal}
  {\bibinfo  {journal} {J. High Energ. Phys.}\ }\textbf {\bibinfo {volume}
  {2021}}\bibfield  {number} {\bibinfo  {number} { (2)},\ \bibinfo {pages}
  {132}},\ }\Eprint {https://arxiv.org/abs/2008.05960} {arXiv:2008.05960}
  \BibitemShut {NoStop}%
\bibitem [{\citenamefont {Apruzzi}\ \emph {et~al.}(2023)\citenamefont
  {Apruzzi}, \citenamefont {Bonetti}, \citenamefont {Etxebarria}, \citenamefont
  {Hosseini},\ and\ \citenamefont {Schäfer-Nameki}}]{AS211202092}%
  \BibitemOpen
  \bibfield  {author} {\bibinfo {author} {\bibfnamefont {F.}~\bibnamefont
  {Apruzzi}}, \bibinfo {author} {\bibfnamefont {F.}~\bibnamefont {Bonetti}},
  \bibinfo {author} {\bibfnamefont {I.~G.}\ \bibnamefont {Etxebarria}},
  \bibinfo {author} {\bibfnamefont {S.~S.}\ \bibnamefont {Hosseini}},\ and\
  \bibinfo {author} {\bibfnamefont {S.}~\bibnamefont {Schäfer-Nameki}},\
  }\bibfield  {title} {\bibinfo {title} {Symmetry {TFTs} from string theory},\
  }\href {https://doi.org/10.1007/s00220-023-04737-2} {\bibfield  {journal}
  {\bibinfo  {journal} {Communications in Mathematical Physics}\ }\textbf
  {\bibinfo {volume} {402}},\ \bibinfo {pages} {895} (\bibinfo {year}
  {2023})},\ \Eprint {https://arxiv.org/abs/2112.02092} {arXiv:2112.02092
  [hep-th]} \BibitemShut {NoStop}%
\bibitem [{\citenamefont {{Chatterjee}}\ \emph {et~al.}(2022)\citenamefont
  {{Chatterjee}}, \citenamefont {{Ji}},\ and\ \citenamefont
  {{Wen}}}]{CW221214432}%
  \BibitemOpen
  \bibfield  {author} {\bibinfo {author} {\bibfnamefont {A.}~\bibnamefont
  {{Chatterjee}}}, \bibinfo {author} {\bibfnamefont {W.}~\bibnamefont {{Ji}}},\
  and\ \bibinfo {author} {\bibfnamefont {X.-G.}\ \bibnamefont {{Wen}}},\
  }\href@noop {} {\bibinfo {title} {{Emergent generalized symmetry and maximal
  symmetry-topological-order}}} (\bibinfo {year} {2022}),\ \Eprint
  {https://arxiv.org/abs/2212.14432} {arXiv:2212.14432} \BibitemShut {NoStop}%
\bibitem [{\citenamefont {Doplicher}\ \emph
  {et~al.}(1969{\natexlab{a}})\citenamefont {Doplicher}, \citenamefont {Haag},\
  and\ \citenamefont {Roberts}}]{DHR1969_fields_I}%
  \BibitemOpen
  \bibfield  {author} {\bibinfo {author} {\bibfnamefont {S.}~\bibnamefont
  {Doplicher}}, \bibinfo {author} {\bibfnamefont {R.}~\bibnamefont {Haag}},\
  and\ \bibinfo {author} {\bibfnamefont {J.~E.}\ \bibnamefont {Roberts}},\
  }\bibfield  {title} {\bibinfo {title} {Fields, observables and gauge
  transformations. {I}},\ }\href {https://doi.org/10.1007/BF01645267}
  {\bibfield  {journal} {\bibinfo  {journal} {Communications in Mathematical
  Physics}\ }\textbf {\bibinfo {volume} {13}},\ \bibinfo {pages} {1} (\bibinfo
  {year} {1969}{\natexlab{a}})}\BibitemShut {NoStop}%
\bibitem [{\citenamefont {Doplicher}\ \emph
  {et~al.}(1969{\natexlab{b}})\citenamefont {Doplicher}, \citenamefont {Haag},\
  and\ \citenamefont {Roberts}}]{DHR1969_fields_II}%
  \BibitemOpen
  \bibfield  {author} {\bibinfo {author} {\bibfnamefont {S.}~\bibnamefont
  {Doplicher}}, \bibinfo {author} {\bibfnamefont {R.}~\bibnamefont {Haag}},\
  and\ \bibinfo {author} {\bibfnamefont {J.~E.}\ \bibnamefont {Roberts}},\
  }\bibfield  {title} {\bibinfo {title} {Fields, observables and gauge
  transformations. {II}},\ }\href {https://doi.org/10.1007/BF01645674}
  {\bibfield  {journal} {\bibinfo  {journal} {Communications in Mathematical
  Physics}\ }\textbf {\bibinfo {volume} {15}},\ \bibinfo {pages} {173}
  (\bibinfo {year} {1969}{\natexlab{b}})}\BibitemShut {NoStop}%
\bibitem [{\citenamefont {{Lan}}\ \emph {et~al.}(2019)\citenamefont {{Lan}},
  \citenamefont {{Zhu}},\ and\ \citenamefont {{Wen}}}]{LW180901112}%
  \BibitemOpen
  \bibfield  {author} {\bibinfo {author} {\bibfnamefont {T.}~\bibnamefont
  {{Lan}}}, \bibinfo {author} {\bibfnamefont {C.}~\bibnamefont {{Zhu}}},\ and\
  \bibinfo {author} {\bibfnamefont {X.-G.}\ \bibnamefont {{Wen}}},\ }\bibfield
  {title} {\bibinfo {title} {Fermion decoration construction of symmetry
  protected trivial orders for fermion systems with any symmetries $g_f$ and in
  any dimensions},\ }\href {https://doi.org/10.1103/PhysRevB.100.235141}
  {\bibfield  {journal} {\bibinfo  {journal} {Phys. Rev. B}\ }\textbf {\bibinfo
  {volume} {100}},\ \bibinfo {pages} {235141} (\bibinfo {year} {2019})},\
  \Eprint {https://arxiv.org/abs/1809.01112} {arXiv:1809.01112} \BibitemShut
  {NoStop}%
\bibitem [{\citenamefont {Gaiotto}\ and\ \citenamefont
  {Kapustin}(2016)}]{GK150505856}%
  \BibitemOpen
  \bibfield  {author} {\bibinfo {author} {\bibfnamefont {D.}~\bibnamefont
  {Gaiotto}}\ and\ \bibinfo {author} {\bibfnamefont {A.}~\bibnamefont
  {Kapustin}},\ }\bibfield  {title} {\bibinfo {title} {Spin {TQFTs} and
  fermionic phases of matter},\ }\href
  {https://doi.org/10.1142/s0217751x16450445} {\bibfield  {journal} {\bibinfo
  {journal} {Int. J. Mod. Phys. A}\ }\textbf {\bibinfo {volume} {31}},\
  \bibinfo {pages} {1645044} (\bibinfo {year} {2016})},\ \Eprint
  {https://arxiv.org/abs/1505.05856} {arXiv:1505.05856} \BibitemShut {NoStop}%
\bibitem [{\citenamefont {Nussinov}\ and\ \citenamefont
  {Ortiz}(2009)}]{NOc0605316}%
  \BibitemOpen
  \bibfield  {author} {\bibinfo {author} {\bibfnamefont {Z.}~\bibnamefont
  {Nussinov}}\ and\ \bibinfo {author} {\bibfnamefont {G.}~\bibnamefont
  {Ortiz}},\ }\bibfield  {title} {\bibinfo {title} {Sufficient symmetry
  conditions for topological quantum order},\ }\href
  {https://doi.org/10.1073/pnas.0803726105} {\bibfield  {journal} {\bibinfo
  {journal} {Proc. Natl. Acad. Sci. U.S.A.}\ }\textbf {\bibinfo {volume}
  {106}},\ \bibinfo {pages} {16944} (\bibinfo {year} {2009})},\ \Eprint
  {https://arxiv.org/abs/cond-mat/0605316} {arXiv:cond-mat/0605316}
  \BibitemShut {NoStop}%
\bibitem [{\citenamefont {Gaiotto}\ \emph {et~al.}(2015)\citenamefont
  {Gaiotto}, \citenamefont {Kapustin}, \citenamefont {Seiberg},\ and\
  \citenamefont {Willett}}]{GW14125148}%
  \BibitemOpen
  \bibfield  {author} {\bibinfo {author} {\bibfnamefont {D.}~\bibnamefont
  {Gaiotto}}, \bibinfo {author} {\bibfnamefont {A.}~\bibnamefont {Kapustin}},
  \bibinfo {author} {\bibfnamefont {N.}~\bibnamefont {Seiberg}},\ and\ \bibinfo
  {author} {\bibfnamefont {B.}~\bibnamefont {Willett}},\ }\bibfield  {title}
  {\bibinfo {title} {Generalized global symmetries},\ }\href
  {https://doi.org/10.1007/jhep02(2015)172} {\bibfield  {journal} {\bibinfo
  {journal} {J. High Energ. Phys.}\ }\textbf {\bibinfo {volume}
  {2015}}\bibfield  {number} {\bibinfo  {number} { (2)},\ \bibinfo {pages}
  {172}},\ }\Eprint {https://arxiv.org/abs/1412.5148} {arXiv:1412.5148}
  \BibitemShut {NoStop}%
\bibitem [{\citenamefont {Etingof}\ \emph {et~al.}(2005)\citenamefont
  {Etingof}, \citenamefont {Nikshych},\ and\ \citenamefont
  {Ostrik}}]{EtingofNikshychOstrik2005}%
  \BibitemOpen
  \bibfield  {author} {\bibinfo {author} {\bibfnamefont {P.}~\bibnamefont
  {Etingof}}, \bibinfo {author} {\bibfnamefont {D.}~\bibnamefont {Nikshych}},\
  and\ \bibinfo {author} {\bibfnamefont {V.}~\bibnamefont {Ostrik}},\
  }\bibfield  {title} {\bibinfo {title} {On fusion categories},\ }\href
  {https://doi.org/10.4007/annals.2005.162.581} {\bibfield  {journal} {\bibinfo
   {journal} {Annals of Mathematics}\ }\textbf {\bibinfo {volume} {162}},\
  \bibinfo {pages} {581} (\bibinfo {year} {2005})}\BibitemShut {NoStop}%
\bibitem [{\citenamefont {Etingof}\ \emph
  {et~al.}(2015{\natexlab{b}})\citenamefont {Etingof}, \citenamefont {Gelaki},
  \citenamefont {Nikshych},\ and\ \citenamefont {Ostrik}}]{EGNO2015}%
  \BibitemOpen
  \bibfield  {author} {\bibinfo {author} {\bibfnamefont {P.}~\bibnamefont
  {Etingof}}, \bibinfo {author} {\bibfnamefont {S.}~\bibnamefont {Gelaki}},
  \bibinfo {author} {\bibfnamefont {D.}~\bibnamefont {Nikshych}},\ and\
  \bibinfo {author} {\bibfnamefont {V.}~\bibnamefont {Ostrik}},\ }\href
  {https://doi.org/10.1090/surv/205} {\emph {\bibinfo {title} {Tensor
  Categories}}},\ \bibinfo {series} {Mathematical Surveys and Monographs},
  Vol.\ \bibinfo {volume} {205}\ (\bibinfo  {publisher} {American Mathematical
  Society},\ \bibinfo {address} {Providence, RI},\ \bibinfo {year}
  {2015})\BibitemShut {NoStop}%
\bibitem [{\citenamefont {Baez}\ and\ \citenamefont
  {Lauda}(2004)}]{BaezLauda2004}%
  \BibitemOpen
  \bibfield  {author} {\bibinfo {author} {\bibfnamefont {J.~C.}\ \bibnamefont
  {Baez}}\ and\ \bibinfo {author} {\bibfnamefont {A.~D.}\ \bibnamefont
  {Lauda}},\ }\bibfield  {title} {\bibinfo {title} {Higher-dimensional algebra
  {V}: 2-groups},\ }\href@noop {} {\bibfield  {journal} {\bibinfo  {journal}
  {Theory and Applications of Categories}\ }\textbf {\bibinfo {volume} {12}},\
  \bibinfo {pages} {423} (\bibinfo {year} {2004})},\ \Eprint
  {https://arxiv.org/abs/math/0307200} {arXiv:math/0307200} \BibitemShut
  {NoStop}%
\bibitem [{\citenamefont {Douglas}\ and\ \citenamefont
  {Reutter}(2018)}]{DouglasReutter2018}%
  \BibitemOpen
  \bibfield  {author} {\bibinfo {author} {\bibfnamefont {C.~L.}\ \bibnamefont
  {Douglas}}\ and\ \bibinfo {author} {\bibfnamefont {D.~J.}\ \bibnamefont
  {Reutter}},\ }\bibfield  {title} {\bibinfo {title} {Fusion 2-categories and a
  state-sum invariant for 4-manifolds},\ }\href@noop {} {\bibfield  {journal}
  {\bibinfo  {journal} {arXiv preprint}\ } (\bibinfo {year} {2018})},\ \Eprint
  {https://arxiv.org/abs/1812.11933} {arXiv:1812.11933 [math.QA]} \BibitemShut
  {NoStop}%
\bibitem [{\citenamefont {Kong}\ \emph {et~al.}(2020)\citenamefont {Kong},
  \citenamefont {Tian},\ and\ \citenamefont {Zhou}}]{KongTianZhou2020}%
  \BibitemOpen
  \bibfield  {author} {\bibinfo {author} {\bibfnamefont {L.}~\bibnamefont
  {Kong}}, \bibinfo {author} {\bibfnamefont {Y.}~\bibnamefont {Tian}},\ and\
  \bibinfo {author} {\bibfnamefont {S.}~\bibnamefont {Zhou}},\ }\bibfield
  {title} {\bibinfo {title} {The center of monoidal 2-categories in 3+1d
  dijkgraaf--witten theory},\ }\href
  {https://doi.org/10.1016/j.aim.2019.106928} {\bibfield  {journal} {\bibinfo
  {journal} {Advances in Mathematics}\ }\textbf {\bibinfo {volume} {360}},\
  \bibinfo {pages} {106928} (\bibinfo {year} {2020})},\ \Eprint
  {https://arxiv.org/abs/1905.04644} {arXiv:1905.04644 [math.QA]} \BibitemShut
  {NoStop}%
\bibitem [{\citenamefont {D{\'e}coppet}\ \emph {et~al.}(2024)\citenamefont
  {D{\'e}coppet}, \citenamefont {Huston}, \citenamefont {Johnson-Freyd},
  \citenamefont {Nikshych}, \citenamefont {Penneys}, \citenamefont {Plavnik},
  \citenamefont {Reutter},\ and\ \citenamefont
  {Yu}}]{DecoppetEtAl2024Fusion2Categories}%
  \BibitemOpen
  \bibfield  {author} {\bibinfo {author} {\bibfnamefont {T.~D.}\ \bibnamefont
  {D{\'e}coppet}}, \bibinfo {author} {\bibfnamefont {P.}~\bibnamefont
  {Huston}}, \bibinfo {author} {\bibfnamefont {T.}~\bibnamefont
  {Johnson-Freyd}}, \bibinfo {author} {\bibfnamefont {D.}~\bibnamefont
  {Nikshych}}, \bibinfo {author} {\bibfnamefont {D.}~\bibnamefont {Penneys}},
  \bibinfo {author} {\bibfnamefont {J.}~\bibnamefont {Plavnik}}, \bibinfo
  {author} {\bibfnamefont {D.}~\bibnamefont {Reutter}},\ and\ \bibinfo {author}
  {\bibfnamefont {M.}~\bibnamefont {Yu}},\ }\bibfield  {title} {\bibinfo
  {title} {The classification of fusion 2-categories},\ }\href@noop {}
  {\bibfield  {journal} {\bibinfo  {journal} {arXiv preprint}\ } (\bibinfo
  {year} {2024})},\ \Eprint {https://arxiv.org/abs/2411.05907}
  {arXiv:2411.05907 [math.QA]} \BibitemShut {NoStop}%
\bibitem [{\citenamefont {Kitaev}(2001)}]{K0131}%
  \BibitemOpen
  \bibfield  {author} {\bibinfo {author} {\bibfnamefont {A.~Y.}\ \bibnamefont
  {Kitaev}},\ }\bibfield  {title} {\bibinfo {title} {Unpaired majorana fermions
  in quantum wires},\ }\href {https://doi.org/10.1070/1063-7869/44/10s/s29}
  {\bibfield  {journal} {\bibinfo  {journal} {Phys.-Usp.}\ }\textbf {\bibinfo
  {volume} {44}},\ \bibinfo {pages} {131} (\bibinfo {year} {2001})},\ \Eprint
  {https://arxiv.org/abs/cond-mat/0010440} {arXiv:cond-mat/0010440}
  \BibitemShut {NoStop}%
\bibitem [{\citenamefont {Costantino}(2005)}]{C0527}%
  \BibitemOpen
  \bibfield  {author} {\bibinfo {author} {\bibfnamefont {F.}~\bibnamefont
  {Costantino}},\ }\bibfield  {title} {\bibinfo {title} {A calculus for
  branched spines of 3-manifolds},\ }\href
  {https://doi.org/10.1007/s00209-005-0810-0} {\bibfield  {journal} {\bibinfo
  {journal} {Math. Z.}\ }\textbf {\bibinfo {volume} {251}},\ \bibinfo {pages}
  {427} (\bibinfo {year} {2005})},\ \Eprint
  {https://arxiv.org/abs/math/0403014} {math/0403014} \BibitemShut {NoStop}%
\bibitem [{\citenamefont {Chen}\ \emph {et~al.}(2013)\citenamefont {Chen},
  \citenamefont {Gu}, \citenamefont {Liu},\ and\ \citenamefont
  {Wen}}]{CGL1172}%
  \BibitemOpen
  \bibfield  {author} {\bibinfo {author} {\bibfnamefont {X.}~\bibnamefont
  {Chen}}, \bibinfo {author} {\bibfnamefont {Z.-C.}\ \bibnamefont {Gu}},
  \bibinfo {author} {\bibfnamefont {Z.-X.}\ \bibnamefont {Liu}},\ and\ \bibinfo
  {author} {\bibfnamefont {X.-G.}\ \bibnamefont {Wen}},\ }\bibfield  {title}
  {\bibinfo {title} {Symmetry protected topological orders and the group
  cohomology of their symmetry group},\ }\href
  {https://doi.org/10.1103/physrevb.87.155114} {\bibfield  {journal} {\bibinfo
  {journal} {Phys. Rev. B}\ }\textbf {\bibinfo {volume} {87}},\ \bibinfo
  {pages} {155114} (\bibinfo {year} {2013})},\ \Eprint
  {https://arxiv.org/abs/1106.4772} {arXiv:1106.4772} \BibitemShut {NoStop}%
\bibitem [{\citenamefont {Chen}\ \emph {et~al.}(2012)\citenamefont {Chen},
  \citenamefont {Gu}, \citenamefont {Liu},\ and\ \citenamefont
  {Wen}}]{CGL1204}%
  \BibitemOpen
  \bibfield  {author} {\bibinfo {author} {\bibfnamefont {X.}~\bibnamefont
  {Chen}}, \bibinfo {author} {\bibfnamefont {Z.-C.}\ \bibnamefont {Gu}},
  \bibinfo {author} {\bibfnamefont {Z.-X.}\ \bibnamefont {Liu}},\ and\ \bibinfo
  {author} {\bibfnamefont {X.-G.}\ \bibnamefont {Wen}},\ }\bibfield  {title}
  {\bibinfo {title} {Symmetry-protected topological orders in interacting
  bosonic systems},\ }\href {https://doi.org/10.1126/science.1227224}
  {\bibfield  {journal} {\bibinfo  {journal} {Science}\ }\textbf {\bibinfo
  {volume} {338}},\ \bibinfo {pages} {1604} (\bibinfo {year} {2012})},\ \Eprint
  {https://arxiv.org/abs/1301.0861} {arXiv:1301.0861} \BibitemShut {NoStop}%
\bibitem [{\citenamefont {Steenrod}(1947)}]{S4790}%
  \BibitemOpen
  \bibfield  {author} {\bibinfo {author} {\bibfnamefont {N.~E.}\ \bibnamefont
  {Steenrod}},\ }\bibfield  {title} {\bibinfo {title} {Products of cocycles and
  extensions of mappings},\ }\href {https://doi.org/10.2307/1969172} {\bibfield
   {journal} {\bibinfo  {journal} {The Annals of Mathematics}\ }\textbf
  {\bibinfo {volume} {48}},\ \bibinfo {pages} {290} (\bibinfo {year}
  {1947})}\BibitemShut {NoStop}%
\end{thebibliography}%
\end{document}